\newif\ifdraft \drafttrue
\newif\iffull \fulltrue
\newif\ifproofsout\proofsoutfalse
\newif\ifred \redfalse
\documentclass[11pt]{article}
\usepackage[margin=1in]{geometry}
\usepackage[margin=0.6in]{caption}
\usepackage[font={scriptsize}]{caption}
\usepackage[usenames,dvipsnames]{xcolor}
\usepackage{wrapfig}
\usepackage{float}
\usepackage{caption}
\usepackage{subcaption}
\usepackage{tikz}
\usepackage{xspace}
\usepackage{amsthm}
\usepackage{graphicx}
\usepackage{amsmath, amssymb, amsthm, amssymb}
\usepackage[utf8]{inputenc}
\usepackage{mathtools}
\usepackage{framed}
\usepackage{enumitem}
\usepackage[square,sort,comma,numbers]{natbib}
\usepackage{color}
\usepackage{mathrsfs}
\usepackage{amsfonts}
\usepackage{pifont}
\usepackage{url}

\newtheorem{thm}{Theorem}
\newtheorem{lemma}{Lemma}

\newtheorem{cor}{Corollary}
\newtheorem{example}{Example}
\newtheorem{assumpt}{Assumption}

\newtheorem{pro}{Proposition}
\newtheorem{definition}{Definition}

\newcommand{\s}{{\bf s}}
\newcommand{\si}[1]{\ensuremath{{\bf s}_{#1}}}
\newcommand{\spi}[1]{\ensuremath{{\bf s}'_{#1}}}
\newcommand{\sm}[1]{\ensuremath{{\bf s}_{-#1}}}

\usepackage[linesnumbered,ruled,vlined]{algorithm2e}

\newcommand{\nice}{well-behaved}
\newcommand{\maxcarnage}{maximum carnage}
\newcommand{\maxdisrupt}{maximum disruption}
\newcommand{\random}{random attack}

\renewcommand{\c}{\ensuremath{C_{\textrm{E}}\xspace}}
\renewcommand{\b}{\ensuremath{C_{\textrm{I}}\xspace}}
\newcommand{\cm}{\ensuremath{C_{\max}\xspace}}
\newcommand{\CC}{\ensuremath{CC\xspace}}
\newcommand{\D}{\ensuremath{\mathcal{H}}}

\newcommand{\I}{\ensuremath{\mathcal{I}}}
\newcommand{\U}{\ensuremath{\mathcal{U}}}
\newcommand{\V}{\ensuremath{\mathcal{V}}}
\newcommand{\T}{\ensuremath{\mathcal{T}}}

\begin{document}
\title{Strategic Network Formation with Attack and Immunization\footnote{The short version of this paper~\cite{GoyalJKKM16} appears in the proceedings of WINE-16.}}
\author{
Sanjeev Goyal\thanks{\scriptsize{University of Cambridge, Faculty of Economics. sg472@cam.ac.uk}}\hspace{5mm}
Shahin Jabbari\thanks{\scriptsize{University of Pennsylvania, Department of Computer and Information Sciences. jabbari@cis.upenn.edu}}\hspace{5mm}
Michael Kearns\thanks{\scriptsize{University of Pennsylvania, Department of Computer and Information Sciences. mkearns@cis.upenn.edu}}\hspace{5mm}
Sanjeev Khanna\thanks{\scriptsize{University of Pennsylvania, Department of Computer and Information Sciences. sanjeev@cis.upenn.edu}}\hspace{5mm}
Jamie Morgenstern\thanks{\scriptsize{University of Pennsylvania, Department of Computer and Information Sciences. jamiemor@cis.upenn.edu}}
}
\maketitle
\abstract{ 
Strategic network formation arises in settings where agents
receive some benefit from their connectedness to other agents, but
also incur costs for forming these links.  We consider a new network
formation game that incorporates an adversarial attack, as well as
{\em immunization\/} or protection against the attack.  An agent's network
benefit is the expected size of her connected component post-attack,
and agents may also choose to immunize themselves from attack at some
additional cost.  Our framework can be viewed as a stylized model of
settings where {\em reachability\/} rather than centrality is the
primary interest (as in many technological networks such as the
Internet), and vertices may be vulnerable to attacks (such as
viruses), but may also reduce risk via potentially costly measures
(such as an anti-virus software).

The reachability network benefit model has been studied in the setting
without attack or immunization~\cite{BalaG00}, where it is known that
the set of equilibrium networks is the empty graph as well as any
tree.  We show that the introduction of attack and immunization
changes the game in dramatic ways; in particular, many new equilibrium
topologies emerge, some more sparse and some more dense than trees.
Our interests include the characterization of equilibrium graphs, and
the social welfare costs of attack and immunization.

Our main theoretical contributions include a strong bound on the edge
density at equilibrium.  In particular, we show that under a very mild assumption
on the adversary's attack model, every equilibrium
network contains at most only $2n-4$ edges for $n \ge 4$, where $n$
denotes the number of agents and this upper bound is tight.  This
demonstrates that despite permitting topologies denser than trees, the
amount of ``over-building'' introduced by attack and immunization is
sharply limited.  We also show that social welfare does not 
significantly erode:
every non-trivial equilibrium in our model with respect to several adversarial attack models
asymptotically has social welfare at least as that of any equilibrium
in the original attack-free model.

We complement our sharp theoretical results with
simulations demonstrating fast convergence of a bounded
rationality dynamic, {\em swapstable best response\/}, which
generalizes linkstable best response but is considerably more powerful
in our model. The simulations further elucidate the wide variety of
asymmetric equilibria possible and demonstrate topological
consequences of the dynamics, including heavy-tailed degree
distributions arising from immunization.
Finally, we report on a behavioral experiment on
our game with over 100 participants, where 
despite the complexity of the game, the
resulting network was surprisingly close to equilibrium.}
\section{Introduction}
\label{sec:intro}
In network formation games, distributed and strategic agents
receive some benefit from their connectedness to others,
but also incur some cost for forming these links.  Much research in this area 
\cite{BalaG00, BlumeEKKT11, FabrikantLMPS03} studies the structure of equilibrium networks
formed as the result of various choices for the network benefit
function, as well as the social welfare in equilibria.  In many
network formation games, the costs incurred from forming links are 
direct: each edge costs $\c > 0$ for an agent to purchase.  
Recently, motivated by scenarios as diverse as financial crises,
terrorism and technological vulnerability, games with indirect
connectivity costs have been considered: an agent's connections expose
her to negative, contagious shocks the network might endure.

We begin with the simple and well-studied {\em
reachability} network formation game~\cite{BalaG00}, in which players
purchase links to each other, and enjoy a network benefit equal to the
size of their connected component in the collectively formed graph. We
modify this model by introducing an adversary who is allowed to
examine the network, and choose a single vertex or player to
attack. This attack then spreads throughout the entire connected
component of the originally attacked vertex, destroying all of these
vertices. Crucially however, players also have the option of purchasing {\em
immunization\/} against attack.  Thus the attack spreads only to
those non-immunized (or {\em vulnerable\/}) vertices reachable from
the originally attacked vertex.  We examine several natural adversarial attacks 
such as an adversary that seeks to maximize destruction, an adversary that randomly 
selects a vertex for the start of infection  and an adversary that seeks to minimize the social 
welfare of the network post-attack to name a few. A player's overall
payoff is thus the expected size of her post-attack component, minus
her edge and immunization expenditures.\footnote{The spread of the initial
attack to reachable non-immunized vertices is deterministic in our model, and the
protection of immunized vertices is absolute. It is also natural to
consider relaxations such as probabilistic attack spreading and
imperfect immunization, as well as generalizations such as
multiple initial attack vertices. See Section~\ref{sec:discussion} for a discussion.
However, as we shall see, even the basic model we study here
exhibits substantial complexity.}

Our game can be viewed as a stylized model for settings where
reachability rather than centrality is the primary interest in joining
a network vulnerable to adversarial attack. Examples include
technological networks such as the Internet, where packet transmission
times are sufficiently low that being ``central''~\cite{FabrikantLMPS03} or a ``hub'' ~\cite{BlumeEKKT11}
is less of a concern, but in the presence of attacks such as viruses
or DDoS, mere reachability may be compromised. Parties may reduce
risks via costly measures such as anti-virus. In a financial setting, vertices might represent banks and edges credit/debt agreements.
The introduction of an attractive but extremely risky asset is a threat or attack on the network that naturally seeks its largest 
accessible market, but can be mitigated by individual institutions adopting balance sheet requirements or leverage restrictions. 
In a biological setting, vertices could represent humans, and edges physical proximity or contact. The attack could
be an actual biological virus that randomly infects an individual and spreads by physical contact through the network;
again, individuals may have the option of immunization. While our simplified model is obviously
not directly applicable to any of these examples in detail, we do believe our results provide some high-level insights about the 
strategic tensions in such scenarios. See Section~\ref{sec:discussion} for discussion of some variants of our model.

Immunization against attack has recently been studied in games played
on a network where risk of contagious shocks are
present~\cite{CerdeiroDG15} but only in the setting in which the
network is first designed by a centralized party, after which agents
make individual immunization decisions.  We endogenize both these aspects, which leads to a
model incomparable to this earlier work.

The original reachability game~\cite{BalaG00}
permitted a sharp and simple characterization of all equilibrium
networks: any tree as well as the empty graph.  We 
demonstrate that once attack and immunization are introduced, the
set of possible equilibria becomes considerably more complex,
including networks that contain multiple cycles, as well as others
which are disconnected but nonempty. This diversity of equilibrium
topologies leads to our primary questions of interest: How dense
can equilibria become? In particular, does the presence of the
attacker encourage the creation of massive redundancy of connectivity?
Moreover, does the introduction of attack and immunization result in
dramatically lower social welfare compared to the original game?

\medskip {\bf Our Results and Techniques} The main theoretical
contributions of this work are to show that our game still exhibits
edge sparsity at equilibrium, and has high social welfare properties
despite the presence of attacks. First we show that under
a very mild assumption on the adversary's attack model, the equilibrium
networks with $n\geq 4$ players have at most $2n-4$ edges,
fewer than twice as many edges as
any nonempty equilibria of the original reachability game without attack.  We prove
this by introducing an abstract representation of the network and use tools from extremal
graph theory to upper bound the resources globally invested by the players to mitigate 
connectivity disruptions due to any attack, obtaining our sparsity result.

We then show that with respect to several adversarial attack models, in any equilibrium with at least one edge and one
immunized vertex, the resulting network is connected.  These results imply that any \emph{new} equilibrium network 
(i.e. one which was not an equilibrium of the original reachability game) is either a sparse but connected graph, 
or is a forest of unimmunized vertices.  The latter occurs only  in the rather unnatural case where the cost of 
immunization or edges grows with the population size, and in the former case we further show the social 
welfare is at least $n^2 - O(n^{5/3})$, which is asymptotically the maximum possible with a polynomial rate
of convergence.  These results provide us with a complete picture of social welfare in our model. We show the welfare lower
bound by first proving any equilibrium network with both immunization and an edge is connected, then showing that there cannot be many
targeted vertices who are \emph{critical} for global connectivity, where critical is defined formally in terms of both the vertex's
probability of attack and the size of the components remaining after the attack.  Thus players myopically optimizing their own utility
create highly resilient networks in presence of attack.

We complement our theory with simulations demonstrating fast and general convergence of  {\em swapstable\/} best 
response, a type of limited best response which generalizes linkstable best response but is much more powerful in our
game. The simulations provide a
dynamic counterpart to our static equilibrium characterizations and illustrate a number of interesting further features of equilibria,
such as heavy-tailed degree distributions.

We conclude by reporting on a behavioral experiment on our network formation game with over 100 participants, where 
despite the complexity of the game, the resulting network was surprisingly close to equilibrium and echoes
many of the theoretical and simulation analyses. 

\medskip {\bf Organization} We formally present our model and review some related work
in Section~\ref{sec:model}. In Section~\ref{sec:eq-ex} we briefly describe some interesting topologies that arise as equilibria
in our model illustrating the richness of the solution space. We present our sparsity result and  lower bound on welfare in
Sections~\ref{sec:sparsity}~and~\ref{sec:welfare}, respectively.   Sections~\ref{sec:exp}~and~\ref{sec:beh} describe our simulations and 
behavioral experiment, respectively.  We conclude with some directions for future work in Section~\ref{sec:discussion}.
\section{Model}
\label{sec:model}

We assume the $n$ vertices of a graph (network) correspond to individual
players. Each player has the choice to purchase edges to other
players at a cost of $\c>0$ per edge. Each player additionally decides
whether to immunize herself at a
cost of $\b> 0$ or remain \emph{vulnerable}. 

A (pure) \emph{strategy} for player $i$ (denoted by $s_{i}$) is a pair
consisting of the subset of players $i$ purchased an edge to and
her immunization choice.  Formally, we denote the subset of
edges which $i$ buys an edge to as $x_i \subseteq \{1, \ldots, n\}$,
and the binary variable $y_i\in \{0,1\}$ as her immunization choice
($y_i=1$ when $i$ immunizes). Then $s_i = (x_i, y_i)$.  {\em We
assume that edge purchases are unilateral i.e. players do not need
approval or reciprocation in order to purchase an edge to another
but that the connectivity benefits and risks are bilateral.} We
restrict our attention to pure strategy equilibria and  our
results show they exist and are structurally diverse.
  
Let $\s = (s_1,\ldots,s_n)$ denote the strategy profile for all the
players.  Fixing $\s$, the set of edges purchased by all the
players induces an undirected graph and the set of immunization
decisions forms a bipartition of the vertices. We denote a game
\emph{state} as a pair $(G, \I)$, where $G=(V,E)$ is the undirected
graph induced by the edges purchased by the players and
$\I \subseteq V$ is the set of players who decide to immunize.  
We use the notation $\U = V\setminus \I$ to denote the 
vulnerable vertices i.e. the players who decide not to immunize.
We refer to a subset of vertices of $\U$ as a \emph{vulnerable region}
if they form a maximally connected component. We denote the
set of vulnerable regions by $\V=\{\V_1, \ldots, V_k\}$ where each 
$\V_i$ is a vulnerable region.

Fixing a game state $(G, \I)$, the adversary inspects the formed
network and the immunization pattern and chooses to attack some vertex. 
If the adversary attacks a vulnerable vertex $v\in\U$,
then the attack starts at $v$ and spreads, killing $v$ and any other
vulnerable vertices reachable from $v$.  
Immunized vertices act as
``firewalls'' through which the attack cannot spread.~\emph{We point out that in this work we restrict the adversary to only pick one seed to start the attack.}

More precisely, the adversary is specified by a function that defines a probability 
distribution over vulnerable regions. 
We refer to a vulnerable region with non-zero probability of attack as a \emph{targeted region}
and the vulnerable vertices inside of a targeted region as \emph{targeted vertices}. 
We denote the targeted regions by $\T = \{\T_1, \ldots, \T_{k'}\}$ where each $\T'\in\T$ denotes a targeted region.\footnote{Since 
every targeted region is vulnerable, the index $k'\leq k$ in the definition of $\T$ (see $k$ in the definition of $\V$).}

$\T=\emptyset$ corresponds to the adversary making no attack, so player $i$'s \emph{utility} (or \emph{payoff}) is
equal to the size of her connected component minus her expenses (edge purchases and
immunization). When $|\T| > 0$, then player's $i$ expected utility
(fixing a game state) is equal to the expected size of her connected
component\footnote{The size of the connected component of a vertex is defined to be zero in
the event she is killed.} less her expenditures, where the expectation is taken over
the adversary's choice of attack (a distribution on $\T$). 
Formally, let $\Pr[\T']$ denote the probability of attack to 
targeted region $\T'$ and $\CC_i(\T')$ 
the size of the connected component of player $i$ 
post-attack to $\T'$. Then the expected utility of
$i$ in strategy profile $s$ denoted by $u_i(s)$ is precisely
\vspace{-5pt}
\begingroup\makeatletter\def\f@size{8}\check@mathfonts
\def\maketag@@@#1{\hbox{\m@th\large\normalfont#1}}
\[
u_i(\s) =\sum_{\T'\in\T}\Big(\Pr\left[\T'\right]\CC_i\left(\T'\right)\Big)-|x_i| \c-y_i\b.
\]
\endgroup
We refer to the sum of expected utilities of all the
players playing $\s$ as the \emph{(social) welfare} of $\s$.

\noindent{\bf Examples of Adversaries}
We highlight several natural adversaries that fit into our framework. 
We begin with a natural adversary
whose goal is to maximize the number of agents killed.
\begin{definition}
\label{def:max-carnage}
The \emph{\maxcarnage}~adversary attacks the vulnerable region of maximum size. If there are multiple
such regions, the adversary picks one of them uniformly at random. 
Once a targeted region is selected for the attack, the adversary selects a vertex inside of that
region uniformly at random to start the attack.
\end{definition}
Then a targeted region with respect to a \maxcarnage~adversary is a vulnerable region of maximum size and 
the adversary defines a uniform distribution over such regions (see Figure~\ref{fig:attacks}).
We now introduce another natural but less sophisticated adversary
which starts an attack by picking a vulnerable vertex at random.

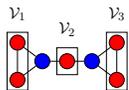
\begin{figure}[h]
\centering
\begin{minipage}[c]{.22\textwidth}
\centering
\scalebox{.55}{
\begin{tikzpicture}
[scale=0.6, every node/.style={circle,draw=black}, red node/.style = {circle, fill = red, draw},  gray node/.style = {circle, fill = blue, draw}]
\node [red node] (1) at  (0, 0){};
\node [red node] (6) at  (-2, 0.75){};
\node [red node] (7) at  (-2, -0.75){};
\node [red node] (4) at  (2, 0.75){};
\node [red node] (5) at  (2, -0.75){};
\node [gray node] (2) at (1,0){};
\node [gray node] (3) at (-1, 0){};
\draw (1) to (2);\draw (1) to (3);\draw (2) to (4);\draw (2) to (5);\draw (4) to (5);\draw (3) to (7);\draw (7) to (6);
\draw (6) to (3);
\node [draw,rectangle,color=black,minimum width=0.5cm,minimum height=1.4cm,label=$\V_3$] at (2, 0) {};
\node [draw,rectangle,color=black,minimum width=0.5cm,minimum height=0.65cm,label=$\V_2$] at (0, 0) {};
\node [draw,rectangle,color=black,minimum width=0.5cm,minimum height=1.4cm,label=$\V_1$] at (-2, 0) {};
\end{tikzpicture}}
\end{minipage}
\begin{minipage}[c]{.75\textwidth}
\caption{Blue and red vertices denote $\I$ and $\U$, respectively.
The probability of attack to the vulnerable regions denoted by $\V_1, \V_2$ and $\V_3$ 
(in that order) for each adversary are as follows.
\maxcarnage: 0.5, 0, 0.5; \random: 0.4, 0.2, 0.4; \maxdisrupt: 0, 1, 0.
\label{fig:attacks}}
\end{minipage}
\end{figure}

\begin{definition}
\label{def:random-attack}
The \emph{\random}~adversary attacks a vulnerable vertex uniformly at random.
\end{definition}
So every vulnerable vertex is targeted with respect to the \random~adversary and the adversary 
induces a distribution over targeted regions such that the probability
of attack to a targeted region is proportional to its size (see Figure~\ref{fig:attacks}). 
Lastly, we define another natural adversary whose goal is to minimize the post-attack social welfare.
\begin{definition}
\label{def:max-disruption}
The \emph{\maxdisrupt}~adversary attacks the vulnerable region which minimizes the post-attack 
social welfare. If there are multiple such regions, the adversary picks one of them uniformly at random. 
Once a targeted region is selected for the attack, the adversary selects a vertex inside of that
region uniformly at random to start the attack.
\end{definition}
This adversary only attacks those vulnerable regions which minimize the post-attack welfare
and the adversary defines a uniform distribution over such regions 
(again see Figure~\ref{fig:attacks}).

\noindent{\bf Equilibrium Concepts}
We analyze the networks formed in our game under
two types of equilibria.  We model each of the $n$ players as
strategic agents who choose deterministically which edges to purchase and whether or not to
immunize, knowing the exogenous behavior
of the adversary defined as above. We say a strategy profile $\s$ is a
\emph{pure-strategy Nash equilibrium} (Nash equilibrium for short) if,
for any player $i$, fixing the behavior of the other players to be
$\sm{i}$, the expected utility for $i$ cannot strictly increase playing any action $\spi{i}$
over $\si{i}$.

In addition to Nash, we study another equilibrium concept that is
closely related to linkstable equilibrium
(see e.g.~\cite{BlockJ06}), a bounded-rationality generalization of Nash.
We refer to this concept
as \emph{swapstable equilibrium}.\footnote{
This equilibrium concept was first introduced by~\citet{Lenzner12} under the name
\emph{greedy equilibrium}.}
A strategy profile is a swapstable equilibrium if no individual agent's expected utility
(fixing other agent's strategies) can strictly improve under 
any of the following \emph{swap deviations:}
(1) Dropping any single purchased edge, (2) Purchasing any single unpurchased edge,
(3) Dropping any single purchased edge and purchasing any single unpurchased edge,
(4) Making any one of the deviations above, and also changing the immunization status.

The first two deviations correspond to the standard
linkstability.  The third permits the more powerful {\em swapping\/}
of one purchased edge for another.  The last additionally allows reversing immunization status. Our interest in swapstable
networks derives from the fact that while they only consider
``simple'' or ``local'' deviation rules, they share several properties
with Nash networks that linkstable networks do not. In that sense,
swapstability is a bounded rationality concept that moves us closer
to full Nash.  Intuitively, in our game (and in many of our proofs),
we exploit the fact that if a player is connected to some other set of
vertices via an edge to a targeted vertex, and that set also
contains an immune vertex, the player would prefer to connect to the
immune vertex instead. This deviation involves a swap not
just a single addition or deletion.  It is worth mentioning explicitly
that by definition every Nash equilibrium is a swapstable
equilibrium and every swapstable equilibrium is a linkstable
equilibrium. The reverse of none of these
statements are true in our game. 
See Appendix~\ref{sec:ne-vs-le} for more details.
We also point out that the set of equilibrium
networks with respect to adversaries defined in Definitions~\ref{def:max-carnage}, \ref{def:random-attack} and 
\ref{def:max-disruption} are disjoint. 
See Appendix~\ref{sec:compare} for more details.
\subsection{Related Work}

 Our paper is a contribution to the study of strategic network design and
defense.  This problem has been extensively studied in economics,
electrical engineering, and computer science
(see~e.g.~\cite{AlpcanB11, Anderson08, Goyal15, Royetal10}).  Most of
the existing work takes the network as given and examines optimal
security choices (see e.g.~\cite{AspnesCY06, Cunningham85, GueyeWA11,
  KearnsOrtiz03, LaszkaSB12}).  To the best of our knowledge, our
paper offers the first model in which both links and defense (immunization) are 
chosen by the players.

Combining  linking and immunization within
a common framework yields new insights.  We start with a discussion of
the network formation literature. In a setting with no attack, our
model reduces to the original model of one-sided reachability network
formation of \citet{BalaG00}.  They showed that a Nash
equilibrium network is either a tree or an empty network. By
contrast, we show that in the presence of a security threat,
Nash networks exhibit very different properties: both networks
containing cycles and partially connected networks can emerge in
equilibrium. Moreover, we show that while networks may contain cycles,
they are sparse (we provide a tight upper bound on the number of links
in any equilibrium network of our game).
 
Regarding security, a recent paper by \citet{CerdeiroDG15} studies optimal design of
networks in a setting where players make immunization choices against a \maxcarnage~adversary but the
network design is given. They show that an optimal network is either a hub-spoke or a
network containing $k$-critical vertices\footnote{Vertex $v$ is
  $k$-critical in a connected network if the size of the largest
  connected component after removing $v$ is $k$.}  or a
partially connected network (observe that a $k$-critical vertex can
secure $n-k$ vertices by immunization).  Our analysis extends this
work by showing that there is a pressure toward the emergence of
$k$-critical vertices even when linking is decentralized. We also contribute 
to the study of welfare costs of decentralization.
\citet{CerdeiroDG15} show that the Price of Anarchy (PoA) is bounded, when
the network is centrally designed while immunization is decentralized
(their welfare measure includes the edge expenditures of the
planner).  By contrast, we show that the PoA
is unbounded when both decisions are decentralized.  
Although we also show that non-trivial equilibrium networks with respect to various adversaries 
have a PoA very near 1. This highlights the key role of linking and resonates with the original results on the
PoA in the context of pure network formation games (see e.g.~\cite{Goyal07}).

Recently \citet{BlumeEKKT11} study network formation where new links generate direct
(but not reachability) benefits, infection can flow through paths of connections and
immunization is not a choice.  They demonstrate a fundamental
tension between socially optimal and stable networks: the former lie
just below a linking threshold that keeps contagion under check, while
the latter admit linking just above this threshold, leading to
extensive contagion and very low payoffs. 

Furthermore, \citet{Kliemann11} introduced a reachability network
formation game with attacks but without defense. In their model, the
attack also happens after the network is formed and the adversary
destroys exactly one \emph{link} in the network (with no spread)
according to a probability distribution over links that can depend on
the structure of the network. They show equilibrium networks in their
model are chord-free and hence sparse. 
We also show an abstract representation of equilibrium networks
in our model corresponds to chord-free graphs and then use this observation
to prove sparsity. While both models
lead to chord-free graphs in equilibria, the analysis of \emph{why} these graphs are
chord-free is quite different. In their model, the deletion of a
single link destroys at most one path between any pair of vertices. 
So if there were two edge-disjoint paths between any pairs of vertices, 
they will certainly remain connected after any attack. In our
model the adversary attacks a vertex and the attack can spread and 
delete many links. This leads to a more delicate analysis. The welfare
analysis is also quite different, since the deletion of an edge can
cause a network to have at most two connected components, while the
deletion of (one or more) vertices might lead to many connected components.

Finally, very recently,~\citet{br16} studied the complexity of computing 
Nash best response for our game with respect to the \maxcarnage~and \random~adversaries.
\section{Diversity of Equilibrium Networks}
\label{sec:eq-ex}
In contrast to the original reachability network formation
game~\cite{BalaG00}, our game exhibits equilibrium networks which
contain cycles, as well as non-empty graphs which are not
connected.\footnote{See Appendix~\ref{sec:no-attack} for more details
on the original reachability network formation game.}
Figure~\ref{fig:eq} gives several examples of specific Nash
equilibrium networks with respect to the \maxcarnage~adversary for small populations, each of which is
representative of a broad family of equilibria for large populations
and a range of values for $\c$ and $\b$ as formalized in 
Appendix~\ref{sec:missing-proofs}.\footnote{Throughout
we represent immunized and vulnerable vertices as blue and red,
respectively. Although we treat the networks as
undirected graphs (since the connectivity
benefits and risks are bilateral), we use directed edges in some
figures to denote which player purchased the edge e.g.
$i\rightarrow j$ means that $i$ has purchased an
edge to $j$. Finally, we use the~\maxcarnage~adversary in many of our illustrations throughout 
  because both the adversary's choice of attack and verifying certain properties are the easiest in this model compared to other natural
  models of Section~\ref{sec:model}.}
These examples show that the tight characterization of the
reachability game, where equilibrium networks are either
empty graph or trees, fails to hold for our more general game.\footnote{The empty
  graph and trees can also form at equilibrium in our
  game.} However, in the
following sections, we show that an approximate version of this
characterization continues to hold for several adversaries.
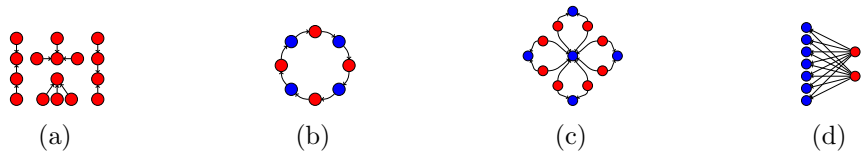
\begin{figure}[h]
\centering
\begin{subfigure}[b]{0.2\textwidth}
\centering
\scalebox{.45}{
\begin{tikzpicture}
[scale = 0.6, every node/.style={circle, fill=red, draw=black}, gray node/.style = {circle, fill = gray, draw}]
\node (1) at  (0, 0){};\node (2) at (-0.7,-1){};\node (3) at (0, -1){};
\node (4) at (0.7, -1){};\node (5) at  (-2, 2){};\node (6) at (-2,1){};
\node (7) at (-2, 0){};\node (8) at (-2, -1){};\node (16) at (2, 2){};
\node (9) at (2,0){};\node (10) at (2, 1){};\node (11) at (2, -1){};
\node (12) at  (0, 1){};\node (13) at  (0, 2){};\node (14) at  (1, 1){};\node (15) at  (-1, 1){};
\draw[->] (2)  to (1);\draw[->] (3)  to (1);\draw[->] (4)  to (1);
\draw[->] (5)  to (6);\draw[->] (7)  to (6);\draw[->] (8)  to (7);
\draw[->] (10)  to (9);\draw[->] (11)  to (9);\draw[->] (16)  to (10);
\draw[->] (13)  to (12);\draw[->] (14)  to (12);\draw[->] (15)  to (12);
\end{tikzpicture}}
\caption{\label{fig:eq-forest}}
\end{subfigure}
\begin{subfigure}[b]{0.2\textwidth}
\centering
\scalebox{.45}{
\begin{tikzpicture}
[scale = 0.50, every node/.style={circle, fill=red, draw=black}, gray node/.style = {circle, fill = blue, draw}]
\node[gray node] (1) at  (0.59, 0.59){};\node[gray node] (3) at  (0.59, 3.41){};
\node[gray node] (5) at  (3.41, 3.41){};\node[gray node] (7) at  (3.41, 0.59){};
\node (2) at  (0, 2){};\node (4) at (2, 4){};
\node (6) at (4, 2){};\node (8) at (2, 0){};
\draw[->] (1) [out=135,in=270] to (2);\draw[->] (2) [out=90,in=225] to (3);
\draw[->] (3) [out=45,in=180] to (4);\draw[->] (4) [out=0,in=135] to (5);
\draw[->] (5) [out=315,in=90] to (6);\draw[->] (6) [out=270,in=45] to (7);
\draw[->] (7) [out=225,in=0] to (8);\draw[->] (8) [out=180,in=315] to (1);
\end{tikzpicture}}
\caption{\label{fig:eq-cycle}}
\end{subfigure}
\begin{subfigure}[b]{0.2\textwidth}
\centering
\scalebox{.36}{
\begin{tikzpicture}
[scale = 0.55, every node/.style={circle, fill= red, draw=black}, gray node/.style = {circle, fill = blue, draw}]
\node[gray node] (1) at  (0, 0){};\node[gray node] (3) at  (3, 0){};
\node[gray node] (6) at  (-3, 0){};\node[gray node] (9) at  (0, 3){};
\node[gray node] (12) at  (0, -3){};
\node (2) at  (2, 1){};\node (4) at  (2, -1){};\node (5) at  (-2, -1){};
\node (7) at  (-2, 1){};\node (8) at  (-1, 2){};\node (10) at  (1, 2){};
\node (11) at  (1, -2){};\node (13) at  (-1, -2){};
\draw[->] (2) [out=180,in=30]  to (1);\draw[->] (4) to [out=0,in=240] (3);
\draw[->] (7) [out=180,in=60] to (6);\draw[->] (5) [out=0,in=220]  to (1);
\draw[->] (13) [out=270,in=150] to (12);\draw[->] (8) [out=270,in=120]  to (1);
\draw[->] (10) [out=90,in=330]  to (9);\draw[->] (2) to [out=0,in=120] (3);
\draw[->] (4) [out=180,in=330]  to (1);\draw[->] (7) [out=0,in=150]  to (1);
\draw[->] (5) [out=180,in=300] to (6);\draw[->] (11) [out=270,in=30]  to (12);
\draw[->] (11) to [out=90,in=300]  (1);\draw[->] (13) [out=90,in=240]  to (1);
\draw[->] (10) [out=270,in=60]  to (1);\draw[->] (8) [out=90,in=210] to (9);
\end{tikzpicture}}
\caption{\label{fig:eq-flower}}
\end{subfigure}
\begin{subfigure}[b]{0.2\textwidth}
\centering
\scalebox{.36}{
\begin{tikzpicture}
[scale=0.45, every node/.style={circle,draw=black}, gray node/.style = {circle, fill = blue, draw}, red node/.style = {circle, fill = red, draw}]
\node [red node] (1) at  (2, 0){};\node [red node] (2) at  (2, 2){};
\node [gray node] (3) at  (-2, 0){};\node [gray node] (4) at  (-2, 1){};
\node [gray node] (5) at  (-2, 3){};\node [gray node] (6) at  (-2, 2){};\node [gray node] (7) at  (-2, 4){};
\node [gray node] (8) at  (-2, -1){};\node [gray node] (9) at  (-2, -2){};
\draw[->] (1) to (3);\draw[->] (1) to (4);\draw[->] (1) to (5);\draw[->] (1) to (6);\draw[->] (1) to (7);\draw[->] (1) to (8);\draw[->] (1) to (9);
\draw[->] (2) to (3);\draw[->] (2) to (4);\draw[->] (2) to (5);\draw[->] (2) to (6);\draw[->] (2) to (7);\draw[->] (2) to (8);\draw[->] (2) to (9);
\end{tikzpicture}}
\caption{\label{fig:eq-biclique}}
\end{subfigure}
\caption{Examples of equilibria with respect to the \maxcarnage~adversary:
(\ref{fig:eq-forest}) Forest equilibrium, $\c=1$ and $\b=9$; (\ref{fig:eq-cycle}) cycle equilibrium, $\c=1.5$ and $\b=3$; 
(\ref{fig:eq-flower}) 4-petal flower equilibrium, $\c=0.1$ and $\b=3$, (\ref{fig:eq-biclique}) Complete bipartite equilibrium, $\c=0.1$ and $\b=4$.
}
\label{fig:eq}
\end{figure}

On the one hand, examples in Figure~\ref{fig:eq} show that equilibrium networks can be denser in
our game compared to the non-attack reachability game. It is thus
natural to ask just how dense they can be. In
Section~\ref{sec:sparsity}, we prove that (under a mild assumption on the adversary) 
the equilibria of our game cannot contain more than $2n-4$ edges when $n\ge 4$. So while these
networks can be denser than trees, they remain quite sparse, and thus the threat of attack does not result 
in too much ``over-building'' or redundancy of connectivity at equilibrium. Our density upper bound is 
tight, as the generalized complete bipartite graph in Figure~\ref{fig:eq-biclique} has exactly $2n-4$ edges.

On the other hand, the examples also show that equilibrium networks can be
disconnected (even before the attack) and this might raise concerns
regarding the welfare compared to the reachability game.  
In Section~\ref{sec:welfare}, we show that for several adversarial attacks, 
all equilibria in our game which contain
at least one edge and at least one immunized vertex (and are thus
{\em non-trivial\/} in the sense that are different than any equilibrium
of the reachability game without attack) are connected
and have immunization patterns such that even
\emph{after} the attack the network remains highly connected.
This allows us to prove that such equilibria in fact enjoy very good welfare.
\section{Sparsity}
\label{sec:sparsity}
We show that despite the existence of equilibria
containing cycles as shown in
Section~\ref{sec:eq-ex}, 
under a very mild restriction on the adversary,
\emph{any} (Nash, swapstable or linkstable) equilibrium network of our game has at most $2n-4$ edges and is
thus quite sparse. 
Moreover, this upper bound is tight as the generalized complete
bipartite graph in Figure~\ref{fig:eq-biclique} has exactly $2n-4$
edges.

The rest of this section is organized as follows.
We start by defining a natural restriction on the adversary. We then 
propose an abstract view of the networks in our game and proceed to show
that the abstract network is chord-free in equilibria with respect to the restricted adversary. 
We finally derive the edge density of the original network by connecting its edge density 
to the density of the abstract network.
We start by defining equivalence classes for networks.
\begin{definition}
\label{def:eq-1}
Let $G_1=(V, E_1)$ and $G_2=(V, E_2)$ be two networks. $G_1$ and $G_2$ are \emph{equivalent}
if for all vertices $v\in V$, the connected component of $v$ is the same in both $G_1$
and $G_2$ for  
every possible choice of initial attack vertex in $V$.
\end{definition}
Based on equivalence, we make the following natural restriction on the adversary.
\begin{assumpt}
\label{assumpt:equiv}
An adversary is \emph{\nice} if on any pair of equivalent networks $G_1=(V,E_1)$ and $G_2=(V,E_1)$, 
the probability that a vertex $v \in V$ is chosen for attack, is the same. 
\end{assumpt}
We point out that the adversaries in Definitions~\ref{def:max-carnage},~\ref{def:random-attack}~and~\ref{def:max-disruption} are all \nice.
We proceed to abstract the network formed by the agents
and argue about the edge density in this abstraction. 

Let $G=(V,E)$ be any network, $\I\subseteq V$ the immunized vertices in $G$
and $\V_1, \ldots, \V_k$ the vulnerable regions in $G$.
In the abstract network every vulnerable region in $G$ is 
contracted to a single vertex.
More formally, let $G'=(V', E')$ be the abstract network.
Define $V'=\I\cup \{u_1, \ldots u_k\}$ where each $u_i$ represents a contracted
vulnerable region of $G$. 
Moreover, $E'$ is constructed from $E$ as follows. 
For any edge $(v_1,v_2)\in E$
such that $v_1, v_2\in\I$ there is an edge $(v_1, v_2)\in E'$. For any edge
$(v_1,v_2)\in E$ such that $v_1\in \V_i$ for some $i$ and $v_2\in \I$ there is an edge $(u_i, v_2)\in E'$
where $u_i$ denotes the contracted vulnerable region of $G$ that $v_1$ belongs to.
For any edge $(v_1, v_2)$ such that $v_1, v_2\in \V_i$ for some $i$ there is no edge in $G'$ (see Figure~\ref{fig:abstract}).
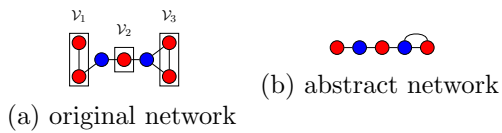
\begin{figure}[h]
\centering
\begin{minipage}[c]{.2\textwidth}
\centering
\scalebox{.5}{
\begin{tikzpicture}
[scale=0.6, every node/.style={circle,draw=black}, red node/.style = {circle, fill = red, draw},  gray node/.style = {circle, fill = blue, draw}]
\node [red node] (1) at  (0, 0){};
\node [red node] (6) at  (-2, 0.75){};
\node [red node] (7) at  (-2, -0.75){};
\node [red node] (4) at  (2, 0.75){};
\node [red node] (5) at  (2, -0.75){};
\node [gray node] (2) at (1,0){};
\node [gray node] (3) at (-1, 0){};
\draw (1) to (2);\draw (1) to (3);\draw (2) to (4);\draw (2) to (5);\draw (4) to (5);\draw (3) to (7);\draw (7) to (6);
\node [draw,rectangle,color=black,minimum width=0.5cm,minimum height=1.4cm,label=$\V_3$] at (2, 0) {};
\node [draw,rectangle,color=black,minimum width=0.5cm,minimum height=0.65cm,label=$\V_2$] at (0, 0) {};
\node [draw,rectangle,color=black,minimum width=0.5cm,minimum height=1.4cm,label=$\V_1$] at (-2, 0) {};
\end{tikzpicture}}
\subcaption{original network}
\end{minipage}
\begin{minipage}[c]{.2\textwidth}
\centering
\scalebox{.5}{
\begin{tikzpicture}
[scale=0.6, every node/.style={circle,draw=black}, red node/.style = {circle, fill = red, draw},  gray node/.style = {circle, fill = blue, draw}]
\node [red node] (1) at  (0, 0){};
\node [red node] (2) at  (-2, 0){};
\node [red node] (3) at  (2, 0){};
\node [gray node] (4) at (1,0){};
\node [gray node] (5) at (-1, 0){};
\draw (1) to (4);\draw (4) to (3);\draw (1) to (5);\draw (2) to (5);
\draw[-] (3) [out=90,in=90] to (4);
\end{tikzpicture}}
\subcaption{abstract network}
\end{minipage}
\begin{minipage}[c]{.45\textwidth}
\centering
\caption{Example of original and abstract network. Blue: immunized vertices in both networks. 
Red: the vulnerable vertices and regions in the original and abstract network, respectively.
\label{fig:abstract}}
\end{minipage}
\end{figure}

We next show that if $G$ is an equilibrium network then $G'$ is a chord-free graph.
\begin{lemma}
\label{lem:achordal-graph}
Let $G=(V,E)$ be a Nash, swapstable or linkstable equilibrium network and $G'=(V', E')$ the
abstraction of $G$. Then $G'$ is a chord-free graph if the adversary is \nice.
\end{lemma}
\ifproofsout
\else
\ifproofsout
\begin{proof}[Proof of Lemma~\ref{lem:achordal-graph}]
\else
\begin{proof}
\fi
We first show that $G'$ is a graph and not a multi-graph.
By construction, we only need to show that no two vertices in any of the vulnerable regions of $G$
are connected to the same immunized vertex.

Suppose by contradiction that there exist an immunized vertex $v\in\I$ and 
vulnerable vertices $v_1, v_2\in \V_i$ (for some vulnerable region $\V_i$ of $G$) 
such that $(v_1, v)$ and $(v_2, v)$ are both in $E$. Given any attack, $v_1$ and $v_2$ would 
either both survive or die. In the former, one of the edges $(v_1, v)$ or $(v_2, v)$ can be dropped
while maintaining the same connectivity benefit for all the survived vertices post-attack because the adversary
is \nice. In the latter,
neither $(v_1, v)$ nor $(v_2, v)$ provide any connectivity benefit for any of the vertices in $\{v, v_1, v_2\}$ post-attack
and dropping one of this edges would strictly increase the utility of the player who purchased that edge (again note that
the distribution of attack remains unchanged because the adversary is~\nice). 
Therefore, $(v_1, v)$ and $(v_2, v)$ cannot both be in $E$ when $G$ is an equilibrium network; a contradiction.

We next show that $G'$ is chord-free.
  Suppose by contradiction that $G'$ has a chord. Then there exists a cycle of size at least $4$ in $G'$
  that has a chord.  Consider any such cycle. By definition there
  exist vertices $u,v, y$ and $z\in V'$ such that \emph{(i)} there are
  at least two vertex disjoint paths between $u$ and $v$,
  \emph{(ii)} $y$ is on the path from $u$ to $v$, $z$ is on the other path,
  and \emph{(iii)} $(y,z)\in E'$. 
  We show that dropping
  the edge between $y$ and $z$ would be a linkstable deviation (and
  hence a swapstable and Nash deviation) that increases the expected payoff of the vertex that
  purchased this edge.  This would contradict our assumption that $G$ is an
  equilibrium network.

 First observe that dropping the edge $(y,z)$ would result in an equivalent 
 network to $G$. Since the adversary is \nice, the distribution of attack 
 in $G$ before and after the deviation is the same.
 Second, by construction of the abstract graph, at most one of $y$ or $z$ dies in any attack. If they both survive, 
 at least one of the vertex disjoint paths between them survives because at most 
 one vertex in $G'$ would die after any attack. So the edge $(y, z)$ is redundant.
 If one of them say $y$ dies then $z$ would be still connected to to the entirety of this cycle and its
  neighborhood. So the edge $(y, z)$ is redundant in this case too.
\end{proof}

\fi

\ifproofsout
We defer all the proofs in this section to Appendix~\ref{sec:missing-proofs-sparse}.
\fi
As the next step we bound the edge density of chord-free networks
in Theorem~\ref{thm:achordal} using
Theorem~\ref{thm:mader} from the graph theory literature. 
\begin{thm}[Mader~\cite{Mader78}]\label{thm:mader}
  Let $G = (V, E)$ be an undirected graph with minimum degree
  $d$. Then there is an edge $(u,v)\in E$ such that there are $d$
  vertex-disjoint paths from $u$ to $v$.
\end{thm}
\begin{thm}\label{thm:achordal}
  Let $G=(V, E)$ be a chord-free graph on $n \ge 4$ vertices. Then $|E|\leq 2n-4$.~
  \footnote{\citet{Kliemann11} proved Theorem~\ref{thm:achordal} with a different technique
for a density bound of $2n-1$ for all $n$.}
\end{thm}
\ifproofsout
\else
\ifproofsout
\begin{proof}[Proof of Theorem~\ref{thm:achordal}]
\else
\begin{proof}
\fi
While $G$ contains a vertex of degree at most 2, we remove this
  vertex from $G$ and repeat this process until either the number of 
  remaining vertices falls to 
  $4$ or the minimum degree in the residual graph is at least $3$.  
  Let $\tilde{G}(\tilde{V},\tilde{E})$ be the
  resulting graph upon termination of this process, and let  $\tilde{n} \ge 4$ 
  denote the number of vertices in $\tilde{G}$. 
  
If $\tilde{n} = 4$, then the assertion of the theorem follows from the following 
two observations: \emph{(i)} we 
removed at most $2(n - 4)$ edges in the process, and \emph{(ii)} any chord-free 
graph on $4$ vertices contains 
at most $4$ edges. Combining these observations together, we can 
conclude that the total number of 
edges in $G$ is at most $2(n-4) + 4 = 2n -4$.

Otherwise, $\tilde{G}$ is a graph with minimum degree of at least $3$. 
Moreover $\tilde{G}$ is chord-free (since $G$
  is chord-free and vertex deletion maintains the chord-free
  property).  Now by Theorem~\ref{thm:mader}, $\tilde{G}$ contains an edge
  $(u,v)$ such that there are at least $3$ vertex-disjoint paths
  connecting $u$ and $v$. This implies that there are at least two vertex
  disjoint paths connecting $u$ and $v$, other than the edge
  $(u,v)$. So there exists some cycle that contains $u$ and $v$ (but not the edge
  $(u,v)$) with length at least $4$.  However, the edge between $u$
  and $v$ would be a chord for such a cycle. This is a contradiction since
  $\tilde{G}$ is chord-free. So, $\tilde{G}$ must be a graph with $4$ vertices, and hence
  there must be at most $2n-4$ edges in $G$.
\end{proof}

\fi
Theorem~\ref{thm:achordal} implies the edge density of the abstract network $G'=(V', E')$
is at most $2|V'|-4$. To derive the edge density of the original network, 
we first show that any vulnerable region in $G$ (contracted vertices in $G'$) is a tree when $G$ is
an equilibrium network.
\begin{lemma}
\label{lem:tree}
Let $G=(V,E)$ be a Nash, swapstable or linkstable equilibrium network. Then 
all the vulnerable regions in $G$ are trees if the adversary is~\nice.
\end{lemma}
\ifproofsout
\else
\ifproofsout
\begin{proof}[Proof of Lemma~\ref{lem:tree}]
\else
\begin{proof}
\fi
Suppose by contradiction that there exists a vulnerable region $\V'$ in $G$ with a cycle. After any attack, the vertices in $\V'$
would either all survive or die. In both cases, any edge beyond a tree is redundant since (1) it 
provides no connectivity benefit and only increases the expenditure (2) the distribution of the attack
would be the same with or without such edge because the adversary is \nice. 
So $\V'$ can't have any cycles and hence is a tree when $G$ is an equilibrium network.
\end{proof}

\fi
We use Lemmas~\ref{lem:achordal-graph},~\ref{lem:tree} and 
Theorem~\ref{thm:achordal} to prove a density bound on the equilibrium networks.
\begin{thm}
\label{thm:sparse-general}
Let $G=(V,E)$ be a Nash, swapstable or linkstable equilibrium network on $n\geq 4$ vertices. Then $|E|\leq 2n-4$ 
for any~\nice~adversary.
\end{thm}
\ifproofsout
\else
\ifproofsout
\begin{proof}[Proof of Theorem~\ref{thm:sparse-general}]
\else
\begin{proof}
\fi
  Let $G'=(V', E')$ be the abstract graph composed from $G$ on $n'$ vertices. 
  We consider two cases based on the number of vertices $n'$ in $G'$: (1) $n'\geq 4$ or (2) $n'\leq 3$.
  Observe that each vertex
  $v'\in V'$ actually represents a tree in $G$ because 
  each vertex is either a singleton immunized vertex which is a tree by definition or a 
  contracted vertex which is tree in $G$ by Lemma~\ref{lem:tree} since the adversary is \nice and $G$
  is an equilibrium network.
  
  In case (1), since the adversary is \nice~and $G'$ is a chord-free graph by Lemma~\ref{lem:achordal-graph}, 
  Theorem~\ref{thm:achordal}
  implies $G'$ has
  at most $|E'|\leq 2n'-4$ edges (since $n'\ge 4$).  
  For every $v'\in V'$, if $v'$
  represents $k_{v'}$ vertices in $G$, this implies that
  $n' = n - \Sigma_{v'\in V'} (k_{v'}-1)$. Thus, $G$ can have at most
  \begin{align*}|E| = 2n'-4 &+ \sum_{v'\in V'} (k_{v'} - 1)
  = 2\left(n - \sum_{v'\in V'} (k_{v'}-1)\right)-4 + \sum_{v'\in V'} (k_{v'}-1) 
 \leq 2n-4
  \end{align*}
  edges, as desired.
  
  In case (2), $|E'|\le n'$ since $n'\leq 3$. Again for every $v'\in V'$, if $v'$
  represents $k_{v'}$ vertices in $G$, this implies that $n' = n - \Sigma_{v'\in V'} (k_{v'}-1)$. Hence,
  \begin{align*}
  |E| \leq |E'| + \sum_{v'\in V'} (k_{v'} - 1) \leq n' + \Sigma_{v'\in V'} (k_{v'}-1) = n
  \end{align*}
  which is at most $2n-4$ when $n\ge 4$.
\end{proof}
\fi
\section{Connectivity and Social Welfare in Equilibria}

\label{sec:welfare}
The results of Section~\ref{sec:sparsity} show that despite
the potential presence of cycles at equilibrium, there are still sharp
limits on collective expenditure on edges in our game. However, they
do not directly lower bound the welfare, due to
connectivity concerns: if the graph could become highly fragmented
after the attack, or is sufficiently fragmented prior
to the attack, the reachability benefits to players could be sharply
lower than in the attack-free reachability game.  In this section we
show that when $\b$ and $\c>1$ are both constants with respect to
$n$,\footnote{We view this condition as the most
  interesting regime of our model, since in natural circumstances we do
  not expect the cost of edge formation or immunization to grow with
  the population size.}  none of these concerns are realized
in any ``interesting'' equilibrium network, described precisely below.

In the original reachability game~\cite{BalaG00},
the \emph{maximum} welfare achievable in any equilibrium is $n^2-O(n)$. 
Here we will
show that the welfare achievable in any ``non-trivial'' equilibrium
is $n^2-O(n^{5/3})$.  Obviously with no
restrictions on the adversary and the parameters this cannot be true.
Just as in the original game, for $\c>1$, the empty graph remains an
equilibrium in our game with respect to all the natural adversaries in Section~\ref{sec:model}.
The empty graph has a social welfare of only $O(n)$ (each
vertex has an expected payoff of $1-1/n$).  
We thus assume the
equilibrium network contains at least \emph{one} edge and at least
\emph{one} immunized vertex.  We refer to all equilibrium networks
that satisfy the above assumption as \emph{non-trivial}
equilibria. They capture the equilibria that are new to our game
compared to the original attack-free setting --- the network is not
empty, and at least one player has chosen immunization.

Limiting attention to non-trivial equilibria is \emph{necessary} if we hope to
guarantee that the welfare at equilibrium is $\Omega(n^2)$ when
$\c>1$. As already noted, without the edge assumption, the empty graph is an
equilibrium with respect to several natural adversaries.
Furthermore, without the immunization assumption, $n/3$ disjoint
components where each component consists of 3 vulnerable vertices is
an equilibrium (for carefully chosen $\c$ and $\b$) with respect to e.g. the \maxcarnage~adversary. 
In both cases, the social welfare is only $O(n)$.

Similar to the sparsity section,  
to get any meaningful results for the welfare we need to restrict the adversary's power. 
To simplify presentation, for the most of this section we 
state and analyze our results for the 
\maxcarnage~adversary. At the end of this section, we show how these results
(or their slight modifications) can be extended to several other adversaries.

Consider any connected component that contains an immunized vertex and an edge
 in a non-trivial equilibrium network with respect to
the \maxcarnage~adversary.
We first show that any targeted region in such component (if exists)
has size one when $\c>1$.
\begin{lemma}
\label{lem:singletons}
Let $G$ be a non-trivial Nash or swapstable equilibrium network with respect to the \maxcarnage~adversary. 
Then in any component of $G$ with at least one immunized vertex and at least one edge, the targeted regions (if they exist)
are singletons when $\c>1$.
\end{lemma}
\begin{proof}
Suppose by contradiction there exist a component $\hat{G}$  with at least one immunized vertex and at least one 
edge and a targeted region $\T$ 
with size strictly bigger than 1 in $\hat{G}$.
Note that $\T$ is a vulnerable region of maximum size in this case.
By Lemma~\ref{lem:tree}, $\T$ is a tree. 
Since $|\T|>1$, then this tree must have at least two leaves $x,y\in \T$.
We claim that there is some vertex in $\T$ who would strictly
prefer to \emph{swap} her edge to some immunized vertex in $\hat{G}$ rather than an edge
which connects her to the remainder of $\T$. 
  
Consider two cases:
(1) one of $x$ or $y$ buys her edge in the tree or (2) neither $x$ nor $y$ buys her edge in the tree.

In case (1), suppose
  without loss of generality that $x$ has bought an edge in the tree. Since $\hat{G}$
  is connected, there exists an immunized vertex $z$ which is connected
  to some vertex in $\T$.  If $x$ is not connected to $z$, then $x$
  would strictly prefer to buy an edge to $z$ over buying her tree
  edge. By this deviation, the probability of
  attack to $x$ is strictly decreased. Furthermore, in any other attack
  outside of $\T$,
  $x$ would at least get the same connectivity benefit. Finally, if the attack happens to the part of $\T$
  that got disconnected from $x$ after the deviation, she would get a non-zero benefit whereas before the
  deviation such attacks would have killed $x$ as well.
  
  So suppose $x$ is connected to $z$. Then if $y$ also bought
  her tree edge, she would also strictly prefer an edge to $z$. 
  So suppose $y$ did not buy her tree edge. Observe that
  $y$ cannot be connected to $z$ because one of the edges $(x,z)$ or $(y,z)$ would 
  be redundant.
  Now consider the edge that connects $y$ to the tree $\T$. Then $y$'s 
  parent in the tree must have bought this edge; since $\c>1$,
  this implies $y$ must be connected to some immunized vertex $z'$ (or
  it would not be worth connecting to $y$); 
  Also observe that $y$'s parent can be connected to $z$ because either the edge
  between $x$ and $z$ or $y$'s parent and $z$ is redundant.
  However, $y$'s parent would
  strictly prefer to buy an edge to $z'$ over an edge to $y$.  
  Thus, $x$ cannot have bought her tree edge; either $y$ or her parent would
  like to re-wire if this were the case.

  In case (2), since $\c>1$, both $x$ and $y$ must have immunized neighbors or
  their edges being purchased by $x$'s parent and $y$'s parent would not be best
  responses by those vertices. 
  Let $z$ and $z'$ denote the immunized vertices connected to $x$ and $y$, respectively.
  Note that $z\ne z'$ otherwise one of the edges $(z,x)$ or $(z', y)$ would be redundant.
  But then, both $x$'s parent and $y$'s parent in the tree $\T$ would
  strictly prefer to buy an edge to $z$ and $z'$
  rather than to $x$ and $y$, respectively.
\end{proof}

We then show that non-trivial equilibrium networks with respect to
the \maxcarnage~adversary are connected when $\c>1$.
We defer the omitted proofs of this section
to Appendix~\ref{sec:missing-connectivity}.
\begin{thm}
\label{thm:connect}
Let $G$ be a non-trivial Nash, swapstable or linkstable equilibrium network with respect to the \maxcarnage~adversary. 
Then, $G$ is a connected graph when $\c>1$.
\end{thm}

Together, Lemma~\ref{lem:singletons} and Theorem~\ref{thm:connect} imply that 
any non-trivial equilibrium network with respect to~\maxcarnage~adversary is a connected
network with targeted regions of size 1.
Finally, we state our main result regarding the welfare in such
non-trivial equilibria. 
\begin{thm}
\label{thm:welfare-new}
Let $G$ be a non-trivial Nash or swapstable equilibrium network on $n$ vertices 
with respect to the \maxcarnage~adversary.  If $\c$ and $\b$
are constants (independent of $n$) and $\c>1$ then the welfare of $G$ is
$n^2 - O(n^{5/3})$.
\end{thm}

\noindent\textbf{Block-Cut Tree Decomposition:}
Before proving Theorem~\ref{thm:welfare-new}, we describe the 
notion of block-cut tree decomposition of a graph.
The \emph{block-cut tree} decomposition
(see e.g.~\cite{WestGraphTheory2000}) of an undirected graph
$G=(V,E)$, denoted by $T = (B \cup C, E')$, 
is defined
as follows. A vertex
$b\in B$ (called a \emph{block}) corresponds to some subset $V_b$ of
$V$ which is a maximal two-connected component in $G$. A vertex
$v\in C$ (called a \emph{cut} vertex) corresponds to some vertex
$v\in V$, the removal of which would increase the number of connected
components in $G$; an edge $e = (b, v) \in E'$ means that $v \in V_b$,
 and that the removal
of $v$ from $G$ would disconnect $V_b\setminus\{v\}$ from some other
part of $G$. 
In contrast to the standard convention,  we assume throughout that cut vertices are not part 
of the blocks their removal would disconnect.  This is simply to avoid over-counting. 
Also, note that all the leaves in $T$ must be blocks since any cut vertex has degree at least $2$.
The decomposition of any undirected graph $G$ can be
efficiently computed in $O(|E|+|V|)$ time.

We define the \emph{size} of a block $b$ (denoted by $|b|$) to be
number of vertices in $V_b$ (which is $|V_b|$).  Also we define the
size of a subtree $T_v$, rooted at $v\in B\cup C$ (denoted by $|T_v|$) to
be the number of vertices contained in the union of all blocks and cut vertices in
$T_v$. We now sketch the proof of Theorem~\ref{thm:welfare-new} and defer the
full proof, which is rather involved, to Appendix~\ref{sec:missing-connectivity}.\\
\noindent
\emph{Proof Sketch for Theorem~\ref{thm:welfare-new}.}
Theorem~\ref{thm:connect} implies that $G$ is connected.
  Also, Lemma~\ref{lem:singletons} implies 
  that all the targeted regions of $G$
  (if there are any) are singletons.
Furthermore, since
  there are at most $2n-4$ edges in $G$ by
  Theorem~\ref{thm:sparse-general} and the number of immunized
  vertices is at most $n$, the collective expenditure of vertices in
  $G$ is at most $\cm:=(2n-4) \c + n \b$.

Let $T=(B\cup C, E')$ be the block-cut tree decomposition of $G$.  
An attack to a targeted non-cut
vertex in any block leaves $G$ with a single connected
component after attack. However, an attack to a targeted cut vertex 
can disconnect $G$. So to analyze welfare, we only consider
the targeted cut vertices in $T$. Moreover we only focus on
targeted cut vertices of $T$ that an attack on such
vertices sufficiently reduces the size of the largest connected
component post-attack. Let
$\epsilon=2 \sqrt{\c}/n^{1/3}$.  A targeted cut vertex $v$
is \emph{heavy} if after an attack to $v$, the size of
the largest connected component in $G\setminus\{v\}$ is strictly less
than $(1-\epsilon)n$. If $G$ is a non-trivial
equilibrium, we show that the total probability of
attack to heavy cut vertices is small. So with high probability
the network retains a large connected component after attack thus the welfare is high.

Root $T$ arbitrarily on some targeted cut vertex $r\in C$.  If there
is no such cut vertex, then the size of largest connected component in
$G$ after any attack is at least $n-1$. So the social welfare is at least $(n-1)^2-\cm$ and we are done.  So assume $r$ exists
and consider the set of cut vertices $\D_r\subseteq C$ such that for
all $v\in \D_r$
(a) $v$ is targeted,
(b) $|T_v|\ge \epsilon  n$, and
(c) no targeted cut vertex $v' \in T_v \setminus \{ v \}$ has the
property that $|T_{v'}|\ge \epsilon n$ i.e. $v$ is the deepest
targeted cut vertex in $T_v$ satisfying (b).
Each $v\in \D_r$ is a heavy cut vertex (but there might be other 
heavy cut vertices in $T$).  Consider
two cases based on the size of $\D_r$: (1) $|\D_r|=1$
and (2) $|\D_r| > 1$.

In case (1) where $|\D_r|=1$, let $\D_r=\{v\}$. Consider the
following two cases: 1(a) $v=r$ and 1(b) $v\ne r$ where $r$ is the root
of the tree.

\begin{figure}[h]
\centering
\begin{minipage}[c]{.2\textwidth}
\centering
\scalebox{.5}{
\begin{tikzpicture}
[scale=0.65, every node/.style={circle,draw=black, minimum size=0.7cm}, red node/.style = {circle, fill = red, draw},  gray node/.style = {circle, fill = blue, draw}]
\node [red node] (2) at  (0, 10){$v$};
\node [draw,rectangle,color=white,minimum width=1cm,minimum height=0.6cm,label=$$] (4) at (0, 8.4) {$$};
\node [draw,rectangle,color=white,minimum width=1.2cm,minimum height=1cm,label=$$] (7) at (-2, 8.5) {$$};
\node [draw,rectangle,color=white,minimum width=1cm,minimum height=1.5cm,label=$$] (8) at (2, 8.5) {$$};
\draw(2) to (4);\draw(2) to (7);\draw(2) to (8);
\draw (0,8.9)--(0.5,8)--(-0.5,8)--cycle;
\draw (1.3,9.05)--(0.8,7.1)--(1.8,7.1)--cycle;
\draw (-1.1,9.15)--(-0.6,7.2)--(-1.6,7.2)--cycle;
\end{tikzpicture}}
\caption{\label{fig:case2a}}
\end{minipage}
\begin{minipage}[c]{.25\textwidth}
\centering
\scalebox{.5}{
\begin{tikzpicture}
[scale=0.65, every node/.style={circle,draw=black, minimum size=0.7cm}, red node/.style = {circle, fill = red, draw},  gray node/.style = {circle, fill = blue, draw}]
\node [red node] (1) at  (0, 4){$v$};
\node [red node] (2) at  (0, 10){$r$};
\node (3) at  (0, 7){};
\node (9) at (2,7){};
\node [draw,rectangle,color=black,minimum width=1cm,minimum height=0.6cm,label=$$] (4) at (0, 8.5) {$$};
\node [draw,rectangle,color=black,minimum width=1cm,minimum height=0.6cm,label=$$] (7) at (-2, 8.5) {$b$};
\node [draw,rectangle,color=black,minimum width=1cm,minimum height=0.6cm,label=$$] (8) at (2, 8.5) {$$};
\node [draw,rectangle,color=black,minimum width=1cm,minimum height=0.6cm,label=$$] (5) at (0, 5.5) {$$};
\node [draw,rectangle,color=black,minimum width=1cm,minimum height=0.6cm,label=$$] (10) at (2, 5.5) {$$};
\draw (0,4.8)--(-1.2,3.5)--(1.2,3.5)--cycle;
\draw(2) to (4);\draw(4) to (3);\draw(2) to (7);\draw(2) to (8);\draw(3) to (5);\draw(1) to (5);
\draw(9) to (10); \draw(9) to (8);
\end{tikzpicture}}
\caption{\label{fig:case2b2}}
\end{minipage}
\begin{minipage}[c]{0.39\textwidth}
\centering
\scriptsize{
Figure~\ref{fig:case2a}: Case 1(a); $v$ is the only heavy cut vertex
 and is the root of $T$. The triangles denote the subtrees
 rooted at the child blocks of $v$.\\
Figure~\ref{fig:case2b2}: Case 1(b2); $v\ne r$ and either $r$ or a vertex in $b$ has a 
beneficial deviation. The triangle denotes
the subtree rooted at $v$.}
\end{minipage}
\end{figure}

In 1(a), let $p$
be the probability of attack to $v$.
We show that $p$
is small or else $v$
  would immunize.  Also if any vertex other than $v$
is attacked, the size of the largest connected component 
post-attack is at least $(1-\epsilon)n$
(see Figure~\ref{fig:case2a}).  These imply  the claimed welfare.

For $1(b)$, observe that the targeted cut vertices on the path
from $v$ to $r$ (the root) are the only possible heavy cut
vertices (counting both $v$ and $r$).  Let $p_v$ denote the probability that some heavy cut
vertex on this path is attacked. 
Consider two cases: 1(b1) $ p_v \leq \sqrt{\c} n^{-1/3} $, and 1(b2)
$p_v > \sqrt{\c} n^{-1/3}$.  In 1(b1) the welfare is
as claimed because the 
probability of attack to heavy cut vertices is small. Moreover, 
1(b2) cannot happen at equilibrium because an immunized vertex
in a child block of $r$ which is not on the path to $v$ has a profitable deviation (see Figure~\ref{fig:case2b2}).

In case (2), let $r'$ be a cut vertex that is the \emph{lowest common
  ancestor} of vertices in $\D_{r}$. If $r'\ne r$, we root the tree on
$r'$ and repeat the process of finding heavy cut vertices. Note that
$\D_{r}\subseteq \D_{r'}$ since we might add some additional heavy cut
vertices to $\D_{r'}$ (see Figures~\ref{fig:before}~and~\ref{fig:after}).

\begin{figure}[h]
\centering
\begin{minipage}[b]{0.22\textwidth}
\centering
\scalebox{.48}{
\begin{tikzpicture}
[scale=0.7, every node/.style={circle,draw=black, minimum size=0.7cm}, red node/.style = {circle, fill = red, draw},  gray node/.style = {circle, fill = blue, draw}]
\node [red node] (1) at  (0, 4){$v_2$};
\node [red node] (12) at  (-2.5, 4){$v_1$};
\node [red node] (2) at  (0, 10){$r$};
\node (3) at  (0, 7){$r'$};
\node [draw,rectangle,color=black,minimum width=1cm,minimum height=0.6cm,label=$$] (4) at (0, 8.5) {$$};
\node [draw,rectangle,color=black,minimum width=1cm,minimum height=0.6cm,label=$$] (7) at (-2, 8.5) {$$};
\node [draw,rectangle,color=black,minimum width=1cm,minimum height=0.6cm,label=$$] (8) at (2, 8.5) {$$};
\node [draw,rectangle,color=black,minimum width=1cm,minimum height=0.6cm,label=$$] (11) at (-2, 5.5) {$$};
\node [draw,rectangle,color=black,minimum width=1cm,minimum height=0.6cm,label=$$] (5) at (0, 5.5) {$$};
\draw (0,4.8)--(-1.2,3.5)--(1.2,3.5)--cycle;
\draw (-2.2, 4.9)--(-4.1,3.5)--(-1.8,3.5)--cycle;
\draw(2) to (4);\draw(4) to (3);\draw(2) to (7);\draw(2) to (8);\draw(3) to (5);\draw(1) to (5);
\draw(11) to (12); \draw(11) to (3);
\end{tikzpicture}}
\caption{\label{fig:before}}
\end{minipage}
\begin{minipage}[b]{0.24\textwidth}
\centering
\scalebox{.5}{
\begin{tikzpicture}
[scale=0.7, every node/.style={circle,draw=black, minimum size=0.7cm}, red node/.style = {circle, fill = red, draw},  gray node/.style = {circle, fill = blue, draw}]
\node [red node] (1) at  (0, 4){$v_2$};
\node [red node] (12) at  (-2.5, 4){$v_1$};
\node [red node] (15) at  (2.5, 4){$r$};
\node (3) at  (0, 7){$r'$};
\node [draw,rectangle,color=black,minimum width=1cm,minimum height=0.6cm,label=$$] (11) at (-2, 5.5) {$$};
\node [draw,rectangle,color=black,minimum width=1cm,minimum height=0.6cm,label=$$] (5) at (0, 5.5) {$$};
\node [draw,rectangle,color=black,minimum width=1cm,minimum height=0.6cm,label=$$] (10) at (2, 5.5) {$$};
\draw(3) to (5);\draw(1) to (5);\draw(11) to (12); \draw(11) to (3);\draw(10) to (3);\draw(15) to (10);
\draw (0,4.8)--(-1.2,3.5)--(1.2,3.5)--cycle;
\draw (-2.2, 4.9)--(-4.1,3.5)--(-1.8,3.5)--cycle;
\draw (2.2, 4.9)--(1.8,3.5)--(4,3.5)--cycle;
\end{tikzpicture}}
\caption{\label{fig:after}}
\end{minipage}
\begin{minipage}[b]{0.42\textwidth}
  \scriptsize{An example of re-rooting in case 2. Heavy cut vertices in $\D$ are
  in red. The small rectangles and circles denote blocks and cut
  vertices, respectively. The triangles denote the subtrees rooted at
  critical cut vertices. Fig.~\ref{fig:before} is before
  and Fig.~\ref{fig:after} is after re-rooting.}
  \end{minipage}
\end{figure}

Note that the vertices in $\D_{r'}$ and the targeted cut vertices
on the path from a $v\in\D_{r'}$ to $r'$ (new root) are the only
possible heavy cut vertices.  Let $p_v$ denote the
probability that some targeted cut vertex on the path from $v$ to $r'$
is attacked.  
Consider two cases: 2(a)
$\Sigma_{v\in\D_{r'}} p_v \leq n^{-1/3}$, and 2(b)
$\Sigma_{v\in\D_{r'}} p_v > n^{-1/3}$. In 2(a) the
welfare is as claimed because the 
probability of attack to heavy cut vertices is small. Finally 2(b) cannot
happen at equilibrium because an immunized vertex
in a child block of one of vertices in $\D_{r'}$ has a profitable deviation.
\qed

Lastly, although non-trivial linkstable equilibrium networks with respect to the \maxcarnage~adversary
are connected when $\c>1$, the size of targeted regions in such networks can be bigger than 1.
So our proof techniques for Theorem~\ref{thm:welfare-new} might not extend to such networks.
\noindent {\bf Remarks}
We proved our sparsity result with a rather mild restriction on the adversary. 
However, we presented our welfare results with respect to a very specific adversary
-- the \maxcarnage~adversary.
The reader might have noticed that our proofs in this section essentially relied only on
the following two properties of 
the \maxcarnage~adversary: (1) Adding an edge between any 2 vertices (at least 1 of which is immunized)
does not change the distribution of the attack and (2) Breaking a link inside of a
targeted region does not increase the probability of attack to the targeted
region  while at the same time does not decrease the probability of attack to any other vulnerable regions. 
These same properties hold for the \random~adversary and other adversaries that set the probability
of attack to a vulnerable region directly proportional to an increasing function of the size of the vulnerable region.
Thus our welfare results extend to \random~adversary and other such adversaries without any modifications.

However, other natural adversaries might not satisfy these properties (e.g. the \maxdisrupt~adversary does not satisfy the first
property). While the techniques in the welfare proofs are not directly applicable to such adversaries, it is still possible to 
reason about the welfare with respect to such adversaries using different techniques e.g. we
can show that in any non-trivial and \emph{connected} equilibrium with respect to the \maxdisrupt~adversary,
when $\c$ and $\b$ are 
constants (independent of $n$) and $\c>1$, then the welfare is $n^2-O(n^{5/3})$. See Appendix~\ref{sec:max-disruption-welfare} for more details. 
Note that this is
slightly weaker than the statement with respect to the \maxcarnage~adversary, because
we cannot show any non-trivial Nash equilibrium network with respect to the \maxdisrupt~adversary is connected
when $\c>1$.\footnote{In fact, this last statement does not hold when we restrict our attention to non-trivial swapstable equilibrium
networks with respect to the \maxdisrupt~adversary even when $\c>1$.} 
We leave the question of whether arguing about welfare is possible using unified techniques 
for a wide class of adversaries as future work.
\section{Simulations}
\label{sec:exp}

We complement our theory
with simulations investigating various properties of
swapstable best response dynamics. 
Again we focused on the \maxcarnage~adversary and
implemented a simulation allowing the specification of the
following parameters: number of players $n$; edge cost
$\c$; immunization cost $\b$; and initial edge density.  The first
three of these parameters are as discussed before
but the last is new and specific to the simulations. Note that for any
$\c \geq 1$,  empty graph is a Nash equilibrium. Thus to sensibly study
any type of best response dynamics, it is necessary to ``seed'' the
process with at least some initial connectivity.  As for motivation, one could view the initial edge purchases as 
occurred prior to the introduction of attack and immunization. 
We examine simulations starting both from very sparse initial
connectivity  and rather dense initial connectivity, for varying combinations 
of the other parameters.  In all cases the initial connectivity was chosen randomly via the Erd\H{o}s-Renyi model.
 
Our simulations proceed in {\em rounds\/}, where each round consists
of a {\em swapstable best response update\/} for all $n$ players in
some fixed order. More precisely, in the update for player $i$ we fix
the edge and immunization purchases of all other players, and compute
the expected payoff of $i$ if she were to alter her current action
according to swap deviations stated in Section~\ref{sec:model}.
Swapstable dynamics is a rich but ``local'' best response process, and
thus more realistic than full Nash best response dynamics\footnote{The
computational complexity of Nash best response was unknown to us at the time of preparing this document.
Very recently, this question has been studied
by \citet{br16} for our game with respect to the \maxcarnage~and~\random~adversaries.} from a bounded rationality
perspective.  We also note that the phenomena we report on here appear
to be qualitatively robust to a variety of natural modifications of
the dynamics, such as restriction to linkstable best
response instead of swapstable,
changes to the ordering of updates, and so on.  Recall that
all of our formal results hold for swapstable as well as Nash
equilibria, so the theory remains relevant for the simulations.
\begin{figure}[h]
\begin{minipage}[b]{0.39\textwidth}
\centering
\captionsetup{width=0.7\textwidth}
 {\includegraphics[width=0.9\textwidth]{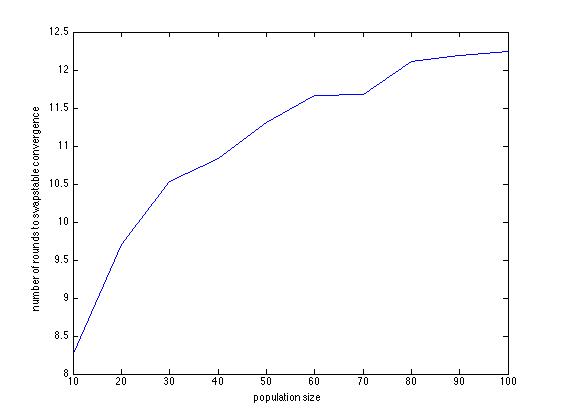}} 
 \caption{
   Average number of rounds for swapstable convergence vs. $n$, for
   $\c = \b = 2$.  \label{fig:convergence}}
\end{minipage}
\begin{minipage}[b]{0.54\textwidth}
\small{
The first question that arises in the consideration of any kind of
best response dynamic is whether and how quickly it will converge to the
corresponding equilibrium notion. Interestingly, empirically it
appears that swapstable best response dynamics {\em always\/}
converges rather rapidly. In Figure~\ref{fig:convergence} we
show the average number of rounds to convergence over many trials,
starting from dense initial connectivity (average degree 5), for
varying values of $n$. The growth in rounds
appears to be strongly sublinear in $n$ (recall that each round
updates all $n$ players, so the overall amount of computation is still
superlinear in $n$).  Thus we conjecture the general  and
 fast convergence of swapstable dynamics. See Appendix~\ref{sec:br-cycles} for more details.
}
\end{minipage}
\end{figure}

In Section~\ref{sec:eq-ex}, we gave a number of formal examples of Nash and swapstable equilibria
with respect to the \maxcarnage~adversary.
These examples tended to exhibit a large amount of symmetry, especially those
containing cycles, due to the large number of cases that need to be considered in the proofs.
Figure~\ref{fig:equil} shows a sampling of ``typical'' equilibria found via simulation for $n = 50$, 
\footnote{In these simulations the initial edge density was only $1/(2n)$, so the initial graph was
very sparse and fragmented.}
which
exhibit interesting asymmetries and illustrate the effects of the parameters.

\begin{figure*}[h]
\centering
\begin{minipage}[c]{.3\textwidth}
{\includegraphics[width=0.8\textwidth]{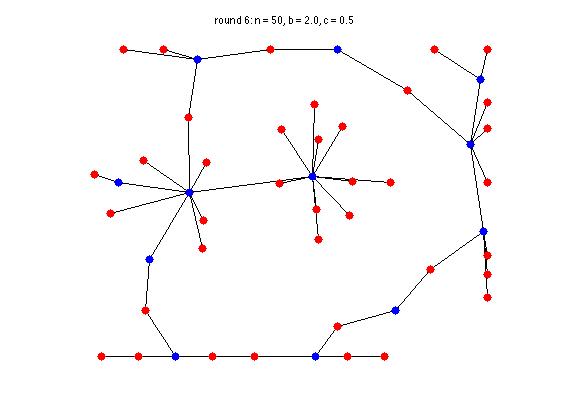}}\qquad
\end{minipage}
\begin{minipage}[c]{.3\textwidth}
{\includegraphics[width=0.8\textwidth]{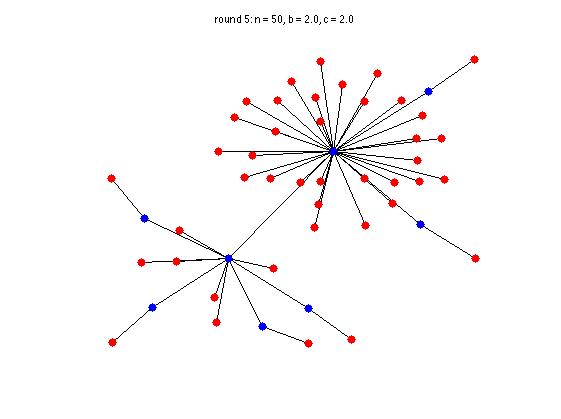}}\qquad
\end{minipage}
\begin{minipage}[c]{.3\textwidth}
{\includegraphics[width=0.8\textwidth]{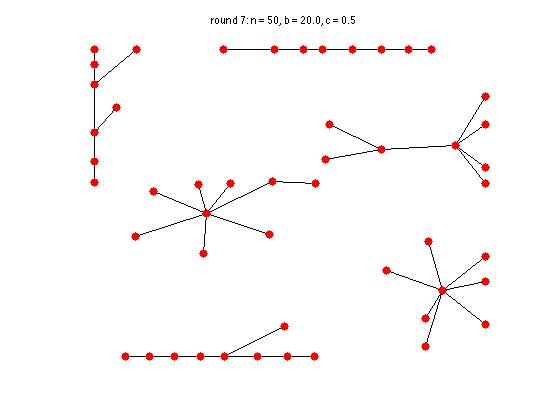}}
\end{minipage}
\caption{Sample equilibria reached by swapstable best response dynamics for $n = 50$.
	Left: $\c = 0.5$, $\b = 2$.
	Middle:  $\c = 2$, $\b = 2$.
	Right: $\c = 0.5$, $\b = 20$.
	\label{fig:equil}}
\end{figure*}

In the left panel of Figure~\ref{fig:equil}, $\c = 0.5$ and $\b = 2$. 
Thus players have an incentive to buy edges even to isolated
vertices as long as they do not increase their vulnerability to the
attack. In this regime, despite the initial disconnectedness of the
graph, we often see equilibria with a long cycle (as shown), with
various tree-like structures attached. In the middle panel we left $\b = 2$ but increased $\c$ to 2.
In this regime cycles are less common due to the higher $\c$.
The equilibrium illustrated is a tree formed by a connected ``backbone'' of immunized players,
each with varying numbers of vulnerable children. Finally, in the
right panel we return to inexpensive edges ($\c = 0.5$), but
greatly increased $\b$ to 20. In this
regime, we see fragmented equilibria with no immunizations.
We note that  unlike the right example which is \emph{trivial}, the examples in the left and
middle are non-trivial equilibria  with high social welfare as predicted by theory.

\begin{figure*}[h]
\centering
\begin{minipage}[c]{.3\textwidth}
{\includegraphics[width=0.8\textwidth]{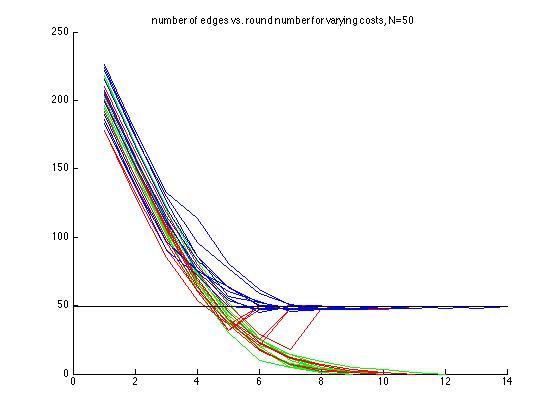}}
\end{minipage}
\begin{minipage}[c]{.3\textwidth}
{\includegraphics[width=0.8\textwidth]{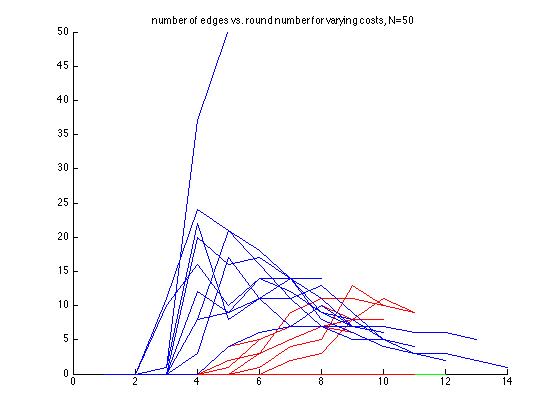}}
\end{minipage}
\begin{minipage}[c]{.3\textwidth}
{\includegraphics[width=0.8\textwidth]{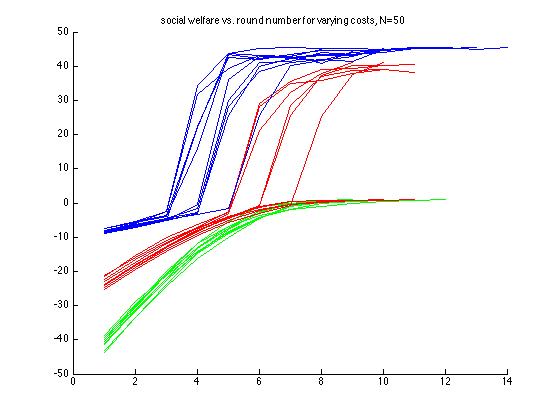}}
\end{minipage}
\caption{Number of edges (left panel), number of immunizations (middle panel), and average welfare (right panel)
vs. number of rounds, for $N = 50$ and varying values for $\b$ and $\c$.
See text for discussion.
\label{fig:dynamics}}
\end{figure*}

Figure~\ref{fig:equil} provides snapshots only at the conclusion
of swapstable dynamics while Figure~\ref{fig:dynamics} examines entire
paths, again at $n = 50$. We started from denser
initial graphs (average degree 5), and each panel
visualizes a different quantity per number of rounds,
for 3 cost regimes:
inexpensive cost $\c = \b = 2$ (blue);
moderate cost $\c = \b = 6$ (red);
and
expensive cost $\c = \b = 10$ (green). In each panel there are 10 simulations for each cost regime.

In the left panel, we show the evolution of the total number of edges ($y$ axis)
in the graph over successive rounds ($x$ axis). In all
regimes, there is initially a precipitous decline in connectivity, as
the overly dense initial graph cannot be supported at equilibrium. 
So in the early rounds all players are
dropping edges. The ultimate connectivity, however,
depends on the cost regime. In the inexpensive regime,
connectivity falls monotonically until it levels out very near the
threshold for global connectivity at $n-1$ (horizontal black line),
resulting in trees or perhaps just one cycle.  In the moderate 
regime, we see a bifurcation; in some trials,
connectivity fall all the way to the empty graph at
equilibrium, while in others fall well below the
$n-1$ tree threshold, but then ``recover'' back to that threshold
(which we discuss shortly). 
In the expensive regime, all trials again
result in a monotonic fall of connectivity all the way
to the empty graph.

For the same cost regimes and trials, the middle panel shows the
number of immunizations purchased over successive rounds. In the
inexpensive regime, immunizations, sometimes many,
are purchased in early rounds. These act as a
``safety net'' that prevents connectivity from falling below the tree
threshold. Typically immunizations grow initially and then decline. 
In the moderate regime, we see that the
explanation for the connectivity bifurcation discussed above can be
traced to immunization decisions. In the trials where connectivity is
recovered, some players eventually choose to immunize and thus provide
the focal points for edge repurchasing. In many trials
resulted in the empty graph, immunizations never occurred (these remain at $y = 0$). In the
expensive regime, no trials are visible because immunizations are
never purchased.

Finally, the right panel shows the evolution of the average social welfare per player
over successive rounds. In the inexpensive regime, welfare increase slowly and modestly
from negative values in the initial graph, then increase dramatically as the benefits
of immunization are realized. In the moderate regime, we see a bifurcation of
welfare corresponding directly to the bifurcation of connectivity. In the expensive regime, all
trials converge from below to the minimum (1-$1/n$) welfare of the empty graph.
Again as theory suggested, the relationship between $\c, \b$ and $n$
is determining whether convergence is to a non-trivial equilibrium and thus
high social welfare, or to a highly fragmented network with no immunizations and low social welfare.

We conclude by noting that for many
parameters, the dynamics above result in
heavy-tailed degree distributions --- a property commonly observed in
large-scale social networks that is easy to capture in
stochastic generative models (such as preferential attachment), but more
rare in strategic network formation.  Across 200 simulations for $n = 100$, $\c = 0.5$ 
and $\b = 2$, we computed the ratio of the maximum
to the average degree in each equilibrium found.  The
lowest, average and maximum ratio observed were 6, 15.8, and 41, respectively (so the highest degree is consistently an order of magnitude
greater than the average or more). Moreover, in all 200 trials the
highest-degree vertex chose immunization, despite the average rate of
immunization of 23\% across the population. Thus an amplification process seems to be 
 at work, where vertices that immunize early become the recipients of many edge purchases, 
 since they provide other vertices connectivity benefits that are relatively secure against
attack without the cost of immunization. 
\section{A Behavioral Experiment}
\label{sec:beh}
To complement our theory
 and simulations, 
we conducted a
behavioral experiment on our game with 118
participants. The participants were students in an undergraduate
survey course on network science at the University of Pennsylvania. As
training, participants were given a detailed document and lecture on
the game, with simple examples of payoffs for players on small graphs
under various edge purchase and immunization decisions.  (See
\url{http://www.cis.upenn.edu/~mkearns/teaching/NetworkedLife/NetworkFormationExperiment2015.pdf}
for the training document provided to participants.)  Participation
was a course requirement, and students were instructed that their
grade on the assignment would be exactly equal to their payoffs
according to the rules of the game. Students thus had strong
motivation to think carefully about the game. There was a 2-day gap
between the training lecture and the experiment, so subjects had time
to contemplate strategies; they were instructed not to discuss
the experiment with each other during this period, and that doing so
would be considered cheating on the assignment.

We again focused on the \maxcarnage~adversary in this section. 
Also costs of $\c = 5$ and
$\b = 20$ were used for the following 
twofold reasons.  First, with $n = 118$ participants (so a
maximum connectivity benefit of 118 points), it felt that these
values made edge purchases and immunization significant expenses and
thus worth careful deliberation.  
Second, running swapstable best response simulations using
these values generally resulted in non-trivial
equilibria with high welfare, whereas raising $\c$ and $\b$
significantly generally resulted in empty or fragmented graphs with
low welfare.

In a game of such complexity, with so many participants, it is 
unreasonable and uninteresting to formulate the experiment as a
one-shot simultaneous move game. Rather, some form of 
communication must be allowed. We chose to conduct the
experiment in an open courtyard
with the single ground rule that {\em all conversations be
  quiet and local\/} i.e. in order to
hear what a participant was saying to others, one should have to stand
next to them. The goal was to permit communication amongst small
groups of participants but to prevent global coordination
via broadcasting. The ``quiet rule'' was enforced by several proctors
for the experiment.

\begin{figure}
\centering
\captionsetup{width=0.8\textwidth}
{\includegraphics[width=0.3\textwidth]{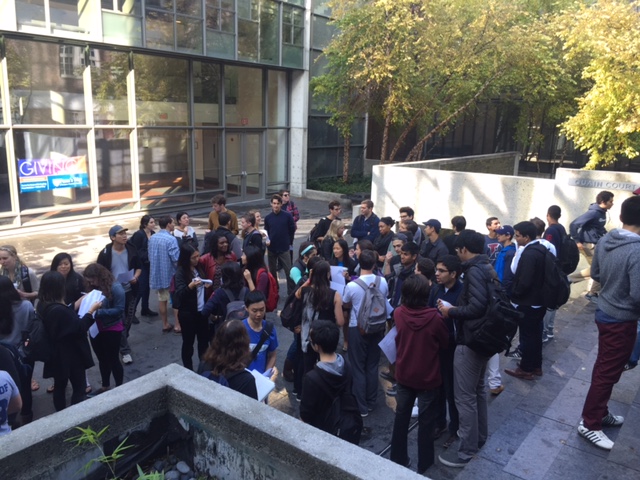}}\qquad
{\includegraphics[width=0.3\textwidth]{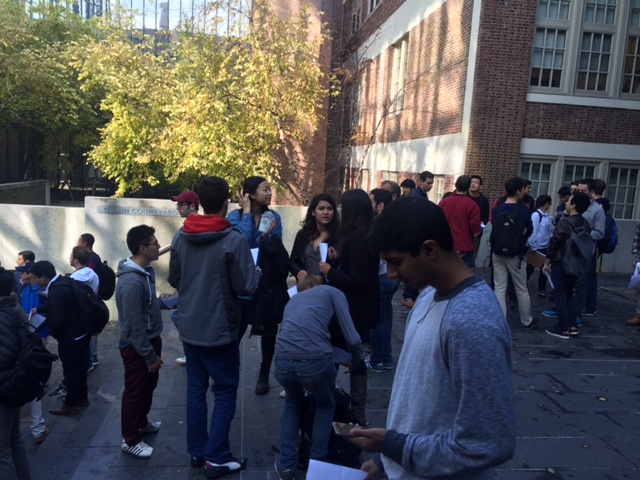}}
\caption{\label{fig:courtyard}\small Experiment participants gathered
  quickly in small groups that reformed frequently during the
  experiment.  }
\end{figure}

Other than the quiet rule, participants were told there were no
restrictions on the nature of conversations. In particular, they
were free to enter agreements, make promises or threats  and move freely in the courtyard.
However, it was also made clear to them that any agreements or
bargains struck would {\em not\/} be enforced by the rules of the
experiment and thus were non-binding. Each subject was given a handout
that simply required them to indicate which other subjects they chose
to purchase edges to (if any), and whether or not they chose to
purchase immunization. The handout contained a list of subject names,
along with a unique identification number for each subject that was
used to indicate edge purchases. Thus subjects knew the actual names
of the others as well as their assigned ID numbers.  
An entire class session was devoted to the experiment, but subjects were
free to (irrevocably) turn in their handout at any time and leave the
experiment. Thus subjects committed to their actions and exited
sequentially, and the entire duration was approximately 30 minutes.
During the experiment, subjects tended to gather quickly in small
discussion groups that reformed frequently, with subjects moving
freely from group to group. (See Figure~\ref{fig:courtyard}). It is clear from the outcome that despite adherence
to the quiet rule, the subjects engaged in widespread coordination via
this rapid mixing of small groups.

In the left panel of Figure~\ref{fig:finalbeh}, we show the final undirected network formed by the
edge purchases and immunization decisions. The graph is clearly anchored by two main immunized
hub vertices, each with many spokes who purchased their single edge to the respective hub.
These two large hubs are both directly connected, as well as by a longer ``bridge'' of 
three vulnerable vertices. There is also a smaller hub with just a handful of spokes, again connected
to one of the larger hubs via a chain of two vulnerable vertices.

In terms of the payoffs, inspection of the behavioral network reveals
that there are two groups of three vertices that are the largest
vulnerable connected components, and thus are the targets of the
attack.  These 6 players are each killed with probability $1/2$ for an expected 
payoff that is only half that of the wealthiest
players (the vulnerable spokes of degree 1). In between are the players 
who purchased immunization including the
three hubs as well as two immunized spokes.  The immunized spoke of
the upper hub is unnecessarily so, while the immunized spoke in the
lower hub is in fact best responding --- had they not purchased
immunization, they would have formed a unique largest vulnerable
component of size 4 and thus been killed with certainty.
\begin{figure}[h]
\centering
\begin{minipage}[c]{.19\textwidth}
\centering
\includegraphics[width=0.9\textwidth]{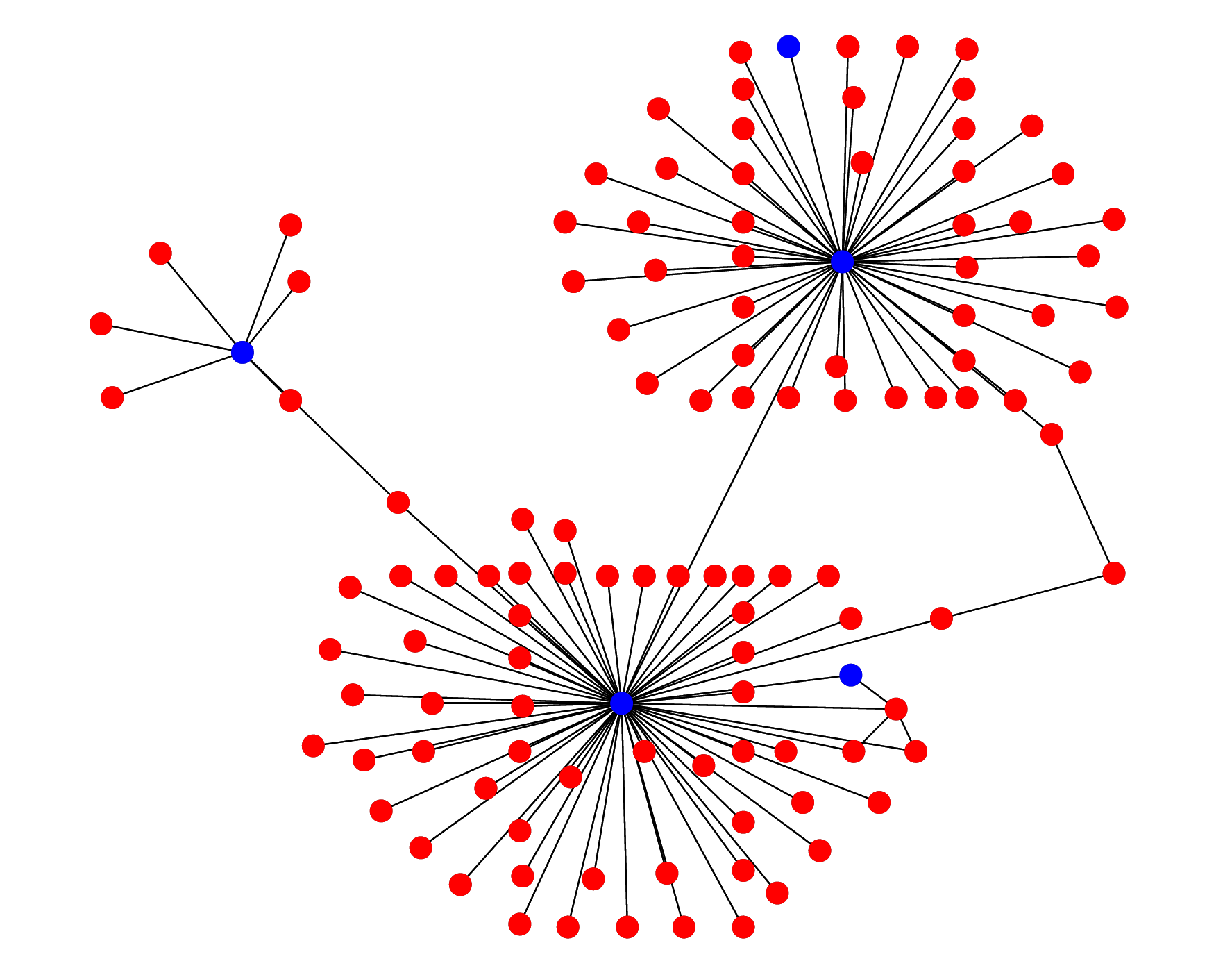}
\end{minipage}
\begin{minipage}[c]{.19\textwidth}
\includegraphics[width=0.9\textwidth]{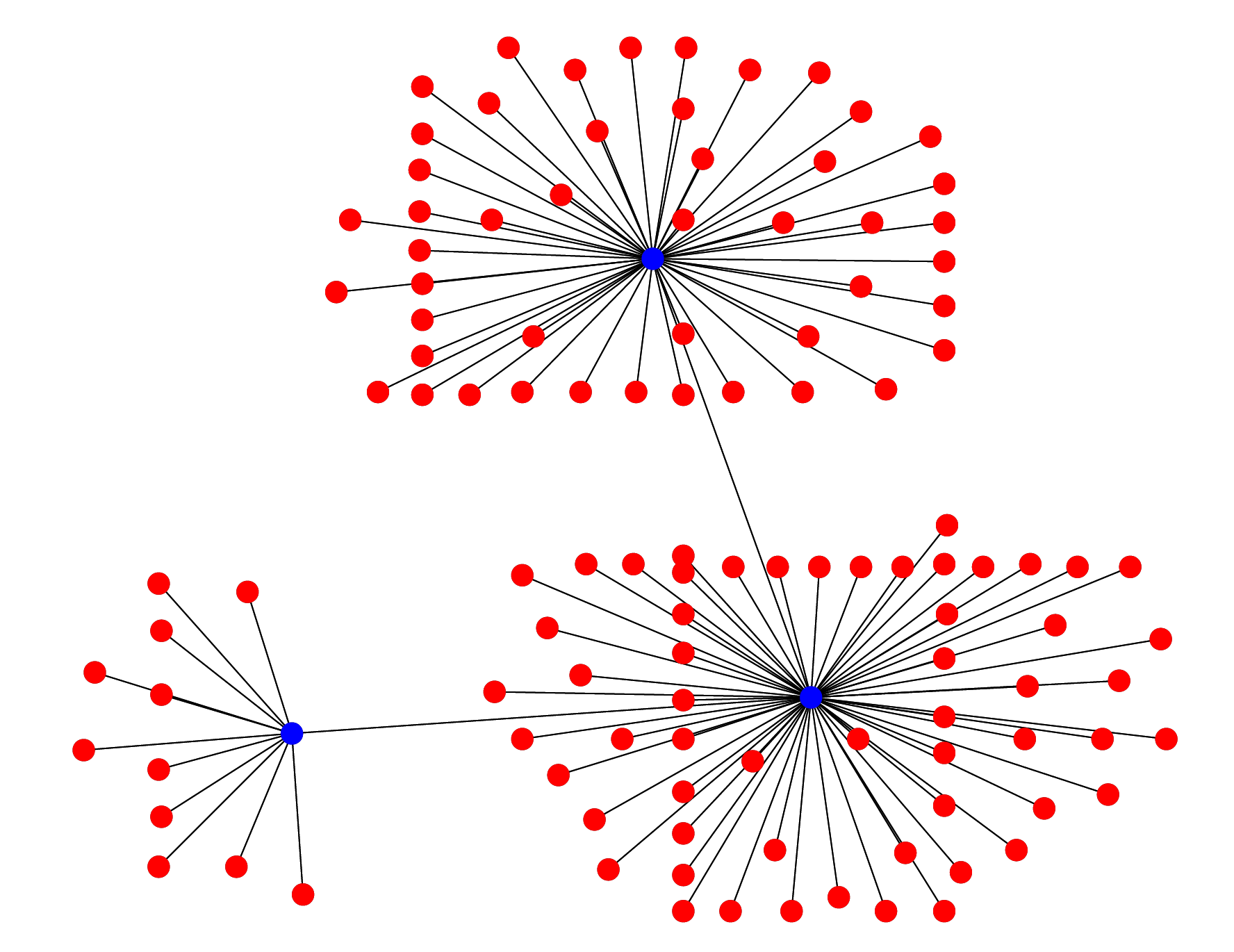}
\end{minipage}
\begin{minipage}[c]{.55\textwidth}
\caption{ Left: the final undirected network formed
	by the edge purchases and immunization decisions (blue for immunized, red for vulnerable).
	Right: a ``nearby'' Nash network.}
\label{fig:finalbeh}
\end{minipage}
\end{figure}

It is striking how many properties the behavioral network shares with
the theory in
Sections~\ref{sec:sparsity} and~\ref{sec:welfare}
 and the simulations in
Section~\ref{sec:exp}: 
multiple hub-spoke structures with sparse
connecting bridges, resulting in high welfare and a
heavy-tailed degree distribution; a couple of cycles; and multiple
targeted components. 
We can quantify how far the behavioral network is
from equilibrium by using it as the starting point for swapstable best
response dynamics and running until convergence. In the right panel of
Figure~\ref{fig:finalbeh}, we show the resulting Nash 
network reached from the behavioral network in only 4 rounds
of swapstable dynamics, and with only 15 of 118 vertices updating their
choices. The dynamics simply ``clean up'' some suboptimal behavioral
decisions --- the vulnerable bridges between hubs are replaced by
direct edges, the other targeted group of three spokes drops theirs
fatal edges, and immunizing spokes no longer do so.

Participants
were also required to complete an email survey shortly after the experiment, 
in which they were asked to comment
on any strategies they contemplated prior to the experiment; whether and how those strategies
changed during the experiment; and what strategies or behaviors they observed in other participants. 
The responses to these surveys are quite illuminating regarding both individual and collective behavior.
We point out some of these responses in the rest of this section.

Many subjects reported entering the experiment with not just a strategy for their own behavior, but
in fact some kind of ``master plan'' they hoped to convince others to join. One frequently reported
master plan involved variations on simple cycles:
\begin{small}
\begin{itemize}
\item
\emph{Going into the experiment my goal was to have everyone connect with one another and take the immunity in a circle.}
\item
\emph{I tried to create a cycle to start. Then I wanted to convince everyone to join our cycle. I figured this would work well because if any of us got immunity it would decrease the probability of being infected.}
\end{itemize}
\end{small}
Interestingly, little thought seems to have been given to how to actually quickly coordinate a global ordering of the participants in a cycle
via only quiet conversation in small groups.

Another frequently cited plan involved the hub and spoke structure that was largely realized:
\begin{small}
\begin{itemize}
\item
\emph{I thought about the possibility of everybody in the class connecting to one person and that person getting immunized.}
\item
\emph{I thought of the hub and spoke model with a hub being immunized. If everyone agreed to this plan then one individual would take a 20 point hit, and every other player would take a 5 +(x-1)/x*number of vertices in component hit. This would maximize everyone's points. Thankfully we had some volunteers to make use of this method. I believe if everyone did as they should, then we will all have our maximum scores.}
\end{itemize}
\end{small}

The strategies above are largely based on mathematical abstractions, but others reported planning to
use real-world social relationships in their strategies:

\begin{small}
\begin{itemize}
\item
\emph{Beforehand, I planned to connect to other students that I knew personally, to have around 3-4 connections. As 
	well, I was going to make one long-distance ``random edge" to increase the size/diversity of my component.}
\item
\emph{I thought of trying to connect with the decently large group of friends who sat in front as well as one or two other random people and then giving myself immunity in order to ensure that I would be safe if the large component I was connected to got compromised by the virus.}
\item
\emph{Before the experiment my strategy was to create a cycle with my friends in the class and then connect with one person that has immunity [who] could potentially connect me with another component.}
\item
\emph{Being a freshman, and knowing that the freshman class of NETS is pretty tight-knit I thought that it would be a good idea to choose someone from that freshman group, because I thought that the component would be of a decent size, but not the giant component itself.}
\end{itemize}
\end{small}

Of course, of particular interest are the surveys of the two large hubs, one of whom was male (who we shall thus refer to as Hub M, and had ID number 128)
and the other female (Hub F, ID number 127). Hub F reports:\\

\begin{small}
\noindent
\emph{The idea was that there would be one hub, with immunity, that everyone would connect to. This seemed like the best idea, because any other idea would pit communities against each other, and the highest payoff that could be expected from any fight between communities would be 49\%. I was willing to be the hub with a lower score of n-20 so everyone else could get a score of (n-1/n)(n-5) rather than have everyone fight...
I went to my friends and told them to spread the word for no one to buy immunity and for everyone to put down 127 as the only person to whom they'd connect, and explained to my friends my reasoning so they could explain to everyone else. A few minutes in, I realized someone else had had the same idea as me (conveniently, number 128) and so I had to adapt my plan because so many people had already put down 128...
We decided to team up and connect with each other, and both opt for immunity, because that would give us the same results as we'd both originally intended. Technically, only one had to connected to the other, and originally, 128 connected with me, but I decided for diplomatic reasons to connect back, so that it would be somewhat more fair to both sides...
A few people who approached me and 128 ended up switching from 128 to me because, apparently, ``girls are more trustworthy".}\\
\end{small}

Hub F thus seems to report an altruistic motivation for purchasing immunization, hoping to maximize social welfare.
In contrast, Hub M displays a more Machiavellian attitude:\\

\begin{small}
\noindent
\emph{I planned to form a radial network with one person being at the center and all other people will connect to this person and this person only. 
And this person at the central position will buy vaccine and no edges... 
However, the catch of this strategy is that the person at the center will score at least 15 points lower than all other nodes that surround him or her. Assuming that the end result is curved, I theorized that if we partition all the NETS students into two groups, those who are inside this radial network (Group A) and those who are not (Group B). The person at the center will still be better off...
the person at the center of the radial network will not end up being the last person in the class and will achieve a decent score by absolute value (thus there will be less of a disincentive for him to be at the center).}\\
\end{small}

Thus Hub M was willing to immunize in the hopes of actually creating three distinct groups of participants: the ``winners''
who would connect to Hub M; Hub M himself, with slightly lower payoff; and then a large group of ``losers'' who would be deliberately left out of the
immunized hub and spoke structure.

It is clear from the surveys that word quickly spread during the experiment to connect to Hubs F and M, and that many participants
joined this strategy --- though not without some reported mistrust, hesitation and manipulation:

\begin{small}
\begin{itemize}
\item
\emph{Some bought in immediately, some are skeptical. Nature of the conversation is usually about: is the person in the middle going to keep the his promise? What if he screws everyone else?}
\item
\emph{During the experiment I tried to convince as many people as possible to connect with me.  I lied about how many people I was actually connected with.}
\item
\emph{Something I found interesting with this experiment was how people found ways to create artificial trust in others. Something that I myself did and something that others did as well, was to put trust in number 127 by saying that since she marked her paper in pen, circling the option to buy immunity we could trust her. She very easily could have crossed that option off later, the fact that she used pen doesn't mean she for sure was trustworthy, but it is comforting to find a way to trust someone and that was one of the strategies that people, including myself, used. Other ways I overheard others create artificial trust was by saying, lets mark 127 or 128 based on their gender. Some girls marked 127 because she was a girl and some guys marked 128 because he was a guy.}
\end{itemize}
\end{small}

Finally, perhaps the most poignant remarks came from the participant who was responsible for creating one of the two
targeted components, only realizing so in hindsight:\\

\begin{small}
\emph{I missed the collaboration part of the experiment where the entire class came up with a plan so I pretty much used a similar strategy to what I had previously thought up.  I wrote down the name of the girl in my main group (that would have created a chain) and then wrote down two other people that would be my ``random'' connections.  However, I didn’t get immunity because I assumed that one of the three people I had connected to bought immunity...
Upon speaking with someone else about my strategy, I saw how flawed it was.  They explained to me what the class had done...
there is a high chance that I may have killed myself and the three people I connected to.  
If the three people I connected to die, it is because I did not buy immunity.  
If they die, it will be my fault --- not theirs.  I regret implementing my strategy.}
\end{small}

\section{Discussion and Future Work}
\label{sec:discussion}
We mention some areas for further study.
Within our model, the question of whether swapstable best response provably 
converges (as seen empirically) is open.
The benefit function considered here is one of many possible natural choices.
It would be interesting to consider other functions. Another extension includes
\emph{imperfect} immunization which fails with some probability e.g. as a
function of the amount of investment.

We mention two natural variants.
The first is the combination of
our original model with a standard diffusion model for the spread of attack. 
For example, combining with the independent cascade model~\cite{KempeKT03},
 when a targeted vertex is attacked, 
the infection spreads with probability $p$ 
along the  edges from the attacked point for which both endpoints
are unimmunized. This spread then continues until we reach
immunized vertices which again act as firewalls.
Again different adversaries can have different objectives e.g.
the \maxcarnage~adversary will pick an attack point which maximizes the expected spread. 
\citet{WangCW12} showed that computing the spread
in the independent cascade model is  \#P-complete.
This suggests that even before considering the complexity of analysis,
agents' reasoning  about the choice of attack by the adversary can become quite complicated.
Furthermore, due to the probabilistic nature of the spread, it is nontrivial to  establish any 
sparsity properties of the equilibria, because additional overbuilding might
occur to hedge against uncertainty of how the infection will spread. Welfare is yet more difficult
to analyze; unlike the deterministic spread, it is no longer obvious that
a vertex likely to end up in a small component post-attack has a
\emph{single} fixed edge purchase that would greatly improve her
utility, since different spread patterns can disconnect her from
different regions.

The second variant is identical to our current model except that it requires edge purchases 
to be \emph{bilateral}. In this variant, the concept of equilibrium might be replaced by the
notion of \emph{pairwise stability} (see e.g.~\cite{Jackson08}). As a majority of our results hinge upon the analysis of 
unilateral deviations, our current 
analysis cannot be easily modified to accommodate this change. As a first step towards 
this goal, the game we study could be modified by adding a \emph{blocking} action with 0 cost while 
maintaining the unilateral edge formation. Namely
for any edge purchased from player $i$ to $j$, player $j$ can block the edge with no cost.
The blocking action removes both the potential connectivity benefit or risk of contagion from edge 
$(i, j)$ for \emph{both} $i$ and $j$. The first observation is that there are equilibrium networks in our game 
which are not equilibria in this new game with blocking.  Moreover, we can show that in any equilibrium of 
the new game, no player blocks any of edges purchased to her. Finally we can show that all the properties 
of our game (sparsity, connectivity and social welfare) hold in the new game with blocking as well.
\section*{Acknowledgments}
We thank Chandra Chekuri, Yang Li and Aaron Roth for useful suggestions.
We also thank anonymous reviewers for detailed suggestions regarding
some of the proofs. Sanjeev Khanna is supported in part by National Science 
Foundation grants CCF-1552909, CCF-1617851, and IIS-1447470.
\appendix
\bibliographystyle{ACM-Reference-Format-Journals}
\bibliography{bib}
\section*{APPENDIX}
\section{Difference Between Solution Concepts}
\label{sec:ne-vs-le}
As we mentioned in Section~\ref{sec:model}, linkstable equilibria and
swapstable equilibria are both generalizations of Nash equilibria. In
particular, this implies that \emph{any} Nash equilibrium is also a
swapstable and linkstable equilibrium.  Furthermore, since linkstable
equilibria is also a generalization of swapstable equilibria, then
\emph{any} swapstable equilibrium is also a linkstable equilibrium.

In this section we show that these solutions concepts are in fact
different in our game. To do so, 
we focus on the \maxcarnage~adversary.
We first show an
example of a swapstable equilibrium
with respect to the \maxcarnage~adversary
 which is not a Nash equilibrium
(Example~\ref{ex:ne-vs-swap-1}). We then show an example of a
linkstable equilibrium \maxcarnage~adversary which is neither a swapstable equilibrium nor a
Nash equilibrium (Example~\ref{exp:ne-vs-ls-2}).

First, we state the following useful Lemma.
\begin{lemma}
\label{lem:tree2}
Let $G=(V,E)$ be a Nash, swapstable or linkstable equilibrium network
with respect to the \maxcarnage~adversary. 
The number of targeted regions cannot be one when $|V|>1$.
\end{lemma}
\begin{proof}
Suppose by contradiction that there exists only
  one targeted region $\T$.  Then $\T$ must be a singleton vertex; otherwise, 
 some player in $\T$ must have purchased an edge and the utility of this player is negative as the attack uniquely targets 
 $\T$ (killing it with probability $1$). Let $u$ denote the singleton vertex in $\T$.
 Then the expected utility of $u$ is 0, and neither $u$ nor any other vertex will purchase an edge incident
  to $u$.

  Now since the number of vertices is strictly bigger than 1, there exists
  some other vertex $v\in V$; this other vertex (and any other
  vertices besides $u$) must be immunized since $\T$ (which is a singleton) is the unique
  targeted region. First suppose that $G$ is the empty graph (i.e., $E=\emptyset$), 
  then for $v$'s behavior to be a best response, it
  must be that $\b \leq 1/2$, since she could drop her
  immunization for expected payoff $1/2$ rather than her current expected payoff of
  $1-\b$. Then, $u$'s expected payoff after immunizing would be
  $1-\b \geq 1/2 > 0$, higher than her current expected payoff; a contradiction to $G$ being an equilibrium. 
  Now suppose $G$ is not an empty graph.
  So there is some immunized vertex $v\in V\setminus \{u\}$ who purchases an
  edge to some other $v'\in V\setminus \{u, v\}$. Let $B$ denote the connected component in $G$ that
  contains $v, v'$. Since $v$ is best responding, it
  must be that $|B| - \b - \c \geq 1/2$, since $v$ could have
  chosen to not buy and edge to $v'$ and not to immunize for an expected utility of at least
  $1/2$. This implies that $u$ cannot be best responding in this case, since buying an
  edge to $v$ and immunizing would give $u$ an expected utility of 
  $(|B|+1) - \b - \c \geq 3/2 > 0$, a contradiction to $G$ being an equilibrium.
\end{proof}

We then show that the set of swapstable
equilibria is indeed larger than the set of Nash equilibria in our 
game (see Figure~\ref{fig:ex-ne-vs-swap-1}).
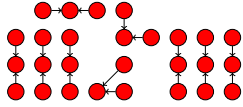
\begin{figure}[h]
\centering\begin{minipage}[c]{.3\textwidth}
\scalebox{.6}{
\centering
\begin{tikzpicture}
[scale = 0.6, every node/.style={circle, fill = red, draw=black}, gray node/.style = {circle, fill = blue, draw}]
\node (1) at  (0, 0){};\node (2) at (-1,-1){};\node (3) at (0, -1){};\node (25) at  (-3, 2){};
\node (26) at (-2,2){};\node (27) at (-1, 2){};\node (6) at (-2,1){};\node (7) at (-2, 0){};
\node (8) at (-2, -1){};\node (16) at (-3,1){};\node (17) at (-3, 0){};\node (18) at (-3, -1){};
\node (19) at (-4,1){};\node (20) at (-4, 0){};\node (21) at (-4, -1){};\node (11) at (1, 1){};
\node (12) at  (0, 1){};\node (13) at  (0, 2){};\node (4) at (2,1){};\node (5) at (2, 0){};
\node (9) at (2, -1){};\node (10) at (3,1){};\node (14) at (3, 0){};
\node (15) at (3, -1){};\node (22) at (4,1){};\node (23) at (4, 0){};\node (24) at (4, -1){};
\draw[->] (1)  to (2);\draw[->] (3)  to (2);\draw[->] (6)  to (7);
\draw[->] (8)  to (7);\draw[->] (4)  to (5);\draw[->] (9)  to (5);
\draw[->] (10)  to (14);\draw[->] (15)  to (14);\draw[->] (11)  to (12);
\draw[->] (13) to (12);\draw[->] (16)  to (17);\draw[->] (18)  to (17);
\draw[->] (19)  to (20);\draw[->] (21)  to (20);\draw[->] (22)  to (23);
\draw[->] (24)  to (23);\draw[->] (25)  to (26);\draw[->] (27)  to (26);
\end{tikzpicture}}
\end{minipage}
\begin{minipage}[c]{.4\textwidth}
\centering
\caption{An example of swapstable equilibrium with respect to the \maxcarnage~adversary which is not a Nash equilibrium. $\c=1$ and $\b=4$.
\label{fig:ex-ne-vs-swap-1}}
\end{minipage}
\end{figure}

\begin{example}
\label{ex:ne-vs-swap-1}
Let $n= 3k$ and consider $k$ disjoint trees of size $3$.  In each tree,
there exists a root vertex that both other vertices purchase an edge to the root.  
When $c\in(0,3/2)$, $\b\in[4,6]$ and $k\ge 9$, the mentioned
network is a swapstable equilibrium with respect to the \maxcarnage~adversary
but is not a Nash equilibrium.
\end{example}
\begin{proof}
  Due to the symmetry in this network, we only need to consider the
  deviations for two types of vertices: the root vertex, and the vertex
  that purchases an edge to the root.

  Let's consider the root vertex first. Her
  utility is $3(1-1/k)$ and her swapstable deviations are as follows:
\begin{enumerate}
\item adding one edge.
\item adding one edge and immunizing.
\item immunizing.
\end{enumerate} 

We show that none of these deviations are beneficial by showing that
the utility before the deviation is always (weakly) bigger than the
utility after deviation.

Case 1 trivially does not happen, because if the root adds an edge
to a different tree, she would be a part of the unique targeted region
which cannot happen in any equilibrium by Lemma~\ref{lem:tree2}. She
also does not want to purchase an edge to any other vertex in her tree, since she is
connected to all the other vertices already.

In case 2, she will survive with probability 1 after immunization. As
far as the edge purchasing decision, she would get the maximum utility
if she purchases an edge to a different tree.
\begin{align*}
\b \ge 4 \text{ and } \c \ge 0  \implies &\b+\c\geq 4 \text{ and } k\geq 9
\implies \\ &3\left(1-\frac{1}{k}\right) \ge \left(3 + 3(1-\frac{1}{k-1})\right) - \c - \b.
\end{align*}

In case 3, she will survive with probability one but she has to pay
for immunization.
\begin{align*}
\b \ge 4 \text{ and } k\geq 9 \implies 3(1-\frac{1}{k}) \ge 3 - \b.
\end{align*}

Next, consider a vertex that purchased an edge. Such
vertex has a utility of $3(1-1/k)-\c$ and her swapstable deviations
are as follows:
\begin{enumerate}
\item dropping her purchased edge.
\item dropping her purchased edge and immunizing.
\item adding one more edge.
\item adding one more edge and immunizing.
\item immunizing.
\item swapping her edge.
\item swapping her edge and immunizing.
\end{enumerate}

We again show that none of these deviations are beneficial.

In case 1, she would survive with probability 1 after dropping her
edge but the size of her connected component will also decrease.
\[
\c\in(0,\frac{3}{2})\text{ and } k\geq 9 \implies 3(1-\frac{1}{k}) - \c \geq 1.
\]

In case 2, once she drops her purchased edge, she is no longer a
targeted vertex. So immunization has no benefits in this case and as
long as case 1 is not beneficial, case 2 cannot be beneficial either.

The analysis of Case 3 is exactly the same as the analysis of case 1 of the
root vertex.

In case 4, she will survive with probability 1 after immunization. As
far as the edge purchasing decision, she would get the maximum utility
if she purchases an edge to a different tree.
\begin{align*}
\b \ge 4 \text{ and } \c \ge 0  \implies & \b+\c\geq 4 \text{ and } k\geq 9
 \\ 
 \implies & 3(1-\frac{1}{k}) - \c \ge
\left(3 + 3(1-\frac{1}{k-1})\right) - 2\cdot\c - \b.
\end{align*}

In case 5, she will survive with probability one but she has to pay
for immunization which is costly.
\begin{align*}
\b \ge 4 \text{ and } k\geq 9 \implies 3(1-\frac{1}{k}) - \c \ge 3 - \c - \b.
\end{align*}

In case 6, swapping her purchased edge to a vertex in another tree
will cause her to be a part of the unique targeted region which cannot
happen in any equilibrium by Lemma~\ref{lem:tree2}.  Also, obviously,
swapping to another vertex in her tree will leave her utility
unchanged.

In case 7, swapping her purchased edge to a vertex in the same tree
and immunizing is not beneficial as long case 5 is not beneficial (she
has the same expenditure and benefit as in case 5).  So we only need
to consider the case when she swaps her edge to a vertex in a
different tree and immunizes.
\[
\b\ge 4 \implies
3(1-\frac{1}{k}) \geq \left(1 + 3(1-\frac{1}{k})\right) - \b.
\]

Finally, it is easy to come up with a strictly beneficial Nash deviation which implies that the above network is not a Nash equilibrium
with respect to the \maxcarnage~adversary.
Consider any vertex that did not purchase an edge. She has a utility of $3(1-1/k)$. Consider
her utility when she buys one edge to all the other trees and immunizes herself.
\begin{align*}
\b \leq 5 \text{ and } & k\geq 9 \text{ and } \c\leq \frac{3}{2} < 2 \implies \\
&(3k-3)-\c\cdot(k-1)-\b \ge(3k-3)-2(k-1)-5
= k-6\geq 3.
\end{align*}
So her utility strictly increases by this deviation, implying that such network cannot be a Nash equilibrium.
\end{proof}

Finally, we show that the set of linkstable equilibria is indeed larger than 
the set of swapstable equilibria (which itself is larger than the set of Nash equilibria).  See Figure~\ref{fig:ex-ne-vs-swap-2}.
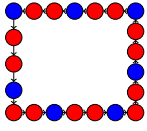
\begin{figure}[h]
\centering
\begin{minipage}[c]{.25\textwidth}
\scalebox{.6}{
\centering
\begin{tikzpicture}
[scale = 0.45, every node/.style={circle, fill = red, draw=black}, gray node/.style = {circle, fill = blue, draw}]
\node (1) at  (-3, 0){};\node (2) at (-2,0){};\node (4) at (0,0){};
\node (5) at (1, 0){};\node (7) at (3,0){};\node (8) at (3,1){};
\node (10) at (3,3){};\node (11) at (3, 4){};\node (13) at  (2, 5){};
\node (14) at (1, 5){};\node (16) at (-1,5){};\node (17) at (-2, 5){};
\node (19) at (-3,3.7){};\node (20) at (-3, 2.4){};
\node [gray node] (3) at (-1, 0){};\node [gray node] (6) at (2,0){};
\node [gray node] (9) at (3,2){};\node [gray node] (12) at  (3, 5){};
\node [gray node] (15) at (0, 5){};\node [gray node] (18) at (-3, 5){};
\node [gray node] (21) at (-3, 1.1){};
\draw[->] (1)  to (2);\draw[->] (2)  to (3);\draw[->] (3)  to (4);\draw[->] (4)  to (5);\draw[->] (5)  to (6);\draw[->] (6)  to (7);
\draw[->] (7)  to (8);\draw[->] (8)  to (9);\draw[->] (9)  to (10);\draw[->] (10)  to (11);\draw[->] (11)  to (12);\draw[->] (12)  to (13);
\draw[->] (13)  to (14);\draw[->] (14)  to (15);\draw[->] (15)  to (16);\draw[->] (16)  to (17);\draw[->] (17)  to (18);\draw[->] (18)  to (19);
\draw[->] (19)  to (20);\draw[->] (20)  to (21);\draw[->] (21)  to (1);
\end{tikzpicture}}
\end{minipage}
\begin{minipage}[c]{.5\textwidth}
\centering
\caption{An example of linkstable equilibrium with respect to the \maxcarnage~adversary 
which is not a swapstable equilibrium. $\c=2$ and $\b=4$.
\label{fig:ex-ne-vs-swap-2}}
\end{minipage}
\end{figure}

\begin{example}
  \label{exp:ne-vs-ls-2}
Consider a cycle consisting of $n=3k$ vertices. If
\emph{(i)} every player buys the edge in her counter clockwise direction on the cycle,
\emph{(ii)} every third vertex in the cycle immunizes (so there are $k$
immunized vertices in the cycle) and \emph{(ii)}
$\c\in(0,n/2-5), \b\in(3, n/2+3)$ and $k\geq 7$, then the cycle
is a linkstable equilibrium with respect to the \maxcarnage~adversary 
but not a swapstable equilibrium.
\end{example}
\begin{proof}
  First, consider any immunized vertex. This vertex clearly cannot
  change her immunization, regardless of how she changes her edge
  purchases. Since she is always connected to the vulnerable vertex to
  her counter clockwise direction, changing the immunization, will
  result in forming the unique largest targeted region which by
  Lemma~\ref{lem:tree2} cannot happen in any equilibrium. So as long as
  the payoff of the immunized vertex is greater than zero before the
  deviation, she will not change her immunization decision.
\begin{align*}
&\c < \frac{n}{2}-5 \text{ and } \b < \frac{n}{2}+3 \implies
\c + \b < n-2
\implies  (n-2) - \c - \b > 0.
\end{align*}
Furthermore, fixing the immunization decision, the linkstable edge
purchasing deviations for any of the immunized vertices are as
follows.
\begin{enumerate}
\item adding one more edge.
\item dropping her purchased edge.
\end{enumerate}

In each case we consider the utilities after and before the deviation
and show that given the conditions in the statement of the example,
the deviation in not beneficial.

In case 1, before the deviation the immunized vertex remains connected
to any vertex that survives.  So adding more edges will only increase
the expenditure while the connectivity benefit is the same. Since $\c>0$,
this deviation is not beneficial.

In case 2, dropping her edge might cause the network to become disconnected 
after the attack and, hence, it decreases the connectivity benefit of the vertex.
\begin{align*}
& \c < \frac{n}{2}-5 < \frac{n}{2}-\frac{3}{2}=\frac{3}{2}(k-1)
\implies  (n-2) - \c - \b >
\frac{1}{k}\Big(1 + 4 + \ldots + \left(3k-2\right)\Big) - \b.
\end{align*}


Now, we consider the vulnerable vertices. Any such vertex (if survives) will
remain connected to any other survived vertex. So no vulnerable vertex
wants to add more edges (while keeping her purchased edge).  So the
possible deviations of such vertices that we need to consider are as
follows.
\begin{enumerate}
\item immunizing.
\item dropping her purchased edge.
\item dropping her purchased edge and immunizing.
\end{enumerate}

Similar to the case of the immunized vertex, we compare the utilities after
and before the deviation and show that given the conditions in the
statement of the example, the deviation is not beneficial.

We divide the vulnerable vertices into two disjoint categories:
\emph{(i)} one that purchases an edge to an immunized vertex (type i)
and \emph{(ii)} one that purchases an edge to a vulnerable vertex
(type ii).  For the 2nd and 3rd deviation, we need to distinguish
between these two types.

In case 1, the vertex who immunizes survives. Also her vulnerable neighbor is not
targeted anymore.
\begin{align*}
\b > 3 & \implies \b > 3-\frac{6}{n} = 3-\frac{2}{k}
\implies (1-\frac{1}{k})(n - 2) - \c > (n-2) - \c - \b.
\end{align*}

In case 2, dropping her edge might cause the network to become disconnected 
after the attack and, hence, it decreases the connectivity benefit of the vertex. For a type i vertex,
\begin{align*}
& \c <  \frac{n}{2}-5 < \frac{n}{2}- \frac{7}{2} + \frac{6}{n}
\implies
(1-\frac{1}{k})(n - 2) - \c > \frac{1}{k}\Big(3 + 6 + \ldots +(3k-3)\Big).
\end{align*}
For a type ii vertex, the analysis is slightly different since after dropping her purchased 
edge, the vertex will not be a part of any targeted region anymore.
\begin{align*}
& \c < \frac{n}{2}-5 < \frac{n}{2} - 4 + \frac{6}{n}
\implies (1-\frac{1}{k})(n - 2) - \c >
\left(2 + \frac{1}{k-1}\Big(0+3 + 6 +\ldots +(3k-6)\Big)\right).
\end{align*}

In case 3, a type i vertex always survives if she immunizes but dropping the edge
will decrease her connectivity benefit.
\begin{align*}
\c &< \frac{n}{2}-5 \text{ and } \b > 3 
\implies \c - \b <  \frac{n}{2}-8 <\frac{n}{2} - 5 + \frac{6}{n}\\
&\implies (1-\frac{1}{k})(n - 2) - \c >
\frac{1}{k-1}\Big(3 + 6 + \ldots +\left(3k-3\right)\Big) - \b.
\end{align*}

Note that a type ii vertex survives if she drops her edge. So
immunization will only increase her expenditure.  So a type ii vertex
always prefers the deviation in case 2 to the deviation in case 3. And since 
we showed that the deviation in case 2 is not beneficial for a type ii vertex,
the deviation in case 3 cannot be beneficial either.

Now it is easy to see that this network cannot form in any swapstable
equilibrium. Consider two vulnerable vertices (denoted by $u$ and
$v$) that are connected to each other with an edge. Suppose without
loss of generality that $u$ has purchased the edge between $u$ and
$v$. Now, it is to see that $u$ can get a strictly higher payoff by
dropping the edge purchased to $v$ and instead buying an edge to the
immunized vertex that $v$ is connected to. This way, $u$ is not a part
of any targeted region anymore (so she survives with probability $1$
instead of $1-1/k$) but she is still connected to every vertex she was
connected to before the deviation. Since her 
expenditure is exactly the same this deviation will increase her utility and show
that such network cannot form in any swapstable equilibrium.
\end{proof}
\section{Comparison of Different Attack Models}
\label{sec:compare}
Throughout we showed that the equilibrium networks in the three stylized models introduced in Section~\ref{sec:model}
exhibit similar behaviors regarding sparsity and welfare.
So in this section we ask whether the equilibrium
networks in these models are the same?
%
We answer this question negatively. In particular, in
Example~\ref{example:dis-ne-car} we show that there exist (connected) Nash equilibrium
networks with respect to the~\maxdisrupt~adversary that are not equilibria with respect to the 
\maxcarnage~adversary.
We then in Example~\ref{example:car-ne-dis} show
that there exist (connected) Nash equilibrium networks with respect to the \maxcarnage~adversary that
are not equilibria with respect to the \maxdisrupt~adversary.
While we can also show the set of equilibrium networks with respect to the random attack adversary
is also disjoint from the the set of equilibrium networks with respect to the other two adversaries, 
we omit the details due to similarities.
\begin{example}
\label{example:dis-ne-car}
There exists a Nash equilibrium network with respect to the~\maxdisrupt~adversary which is not a Nash equilibrium with respect to the \maxcarnage~adversary.
\end{example}
\begin{proof}
Consider a complete binary tree with 
 $n=2^{n'+1}-1$ vertices (so the height of the tree is $n'$). Suppose
 every player in height $i > 0$ 
 buys the two edges to the players in height $i-1$ and every vertex in height $i>0$ 
 immunizes (see Figure~\ref{fig:max-carnage-ne-disruption}). 
 Then for $\b \in  (2, n-2)$, $\c \in (0, 1- 1/2^{n'})$ and $n'\geq 3$ this tree is a Nash equilibrium 
 with respect to the \maxdisrupt~adversary but not with respect to the \maxcarnage~adversary.

\begin{figure}[ht]
\centering\begin{minipage}[c]{.4\textwidth}
\scalebox{.5}{
\begin{tikzpicture}
[scale = 0.50, every node/.style={circle, fill=red, draw=black}, gray node/.style = {circle, fill = blue, draw}]
\node[gray node] (0) at  (0, 0){};
\node[gray node] (1) at  (-4, -1){};\node[gray node] (2) at  (+4, -1){};
\node[gray node] (3) at  (-6, -2){};\node[gray node] (4) at  (-2, -2){};\node[gray node] (5) at  (2, -2){};\node[gray node] (6) at  (6, -2){};
\node (7) at  (-7, -3){};\node (8) at  (-5, -3){};\node (9) at  (-3, -3){};\node (10) at  (-1, -3){};
\node (11) at  (1, -3){};\node (12) at  (3, -3){};\node (13) at  (5, -3){};\node (14) at  (7, -3){};
\draw[->] (0)  to (1);\draw[->] (0) to (2);
\draw[->] (1)  to (3);\draw[->] (1) to (4);\draw[->] (2)  to (5);\draw[->] (2) to (6);
\draw[->] (3)  to (7);\draw[->] (3) to (8);\draw[->] (4)  to (9);\draw[->] (4) to (10);
\draw[->] (5)  to (11);\draw[->] (5) to (12);\draw[->] (6)  to (13);\draw[->] (6) to (14);
\end{tikzpicture}}
\end{minipage}
\begin{minipage}[c]{.5\textwidth}
\centering
\caption{An example of Nash equilibrium with respect to the \maxdisrupt~adversary which is not a Nash equilibrium with respect to the \maxcarnage~adversary. $\c=15/16$ and $\b=33/16$.
\label{fig:max-carnage-ne-disruption}}
\end{minipage}
\end{figure}
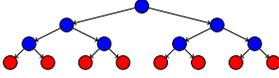

We first show that the tree in Figure~\ref{fig:max-carnage-ne-disruption} is an equilibrium with respect to the~\maxdisrupt~adversary. First, consider any immunized
vertex other than the root of the tree. If such vertex changes her immunization decision (regardless of 
how she changes her edge purchasing decision), she would deterministically get killed by the~\maxdisrupt~adversary. 
So she would not change her immunization decision as long as her utility is positive
which means $(n-1) - \b - 2\c > 0$. The same applies to the root as long as the root maintains one of her
purchased edges. It is easy to see that the deviation that she drops both of her edges and changes
her immunization is also not beneficial as long as $(n-1) - \b - 2\c > 1$ since she would not get attacked by the
adversary in this case.

Any vertex who survives after an attack will remain connected to any other surviving vertex; which implies no vertex
would strictly prefer to add another edge. We now show that no vertex would like to drop any of her purchased edges
and leaf vertices in height 0 would prefer to remain targeted. To show the former, note that the benefit
from each edge purchase is at least $1- 1/2^{n'}$ and this comes from edges purchased to the targeted leaves.
So no vertex would prefer to drop any of her edges since the benefit is strictly more than $\c$. To show the latter, the 
utility of a leaf in the tree is $(1-1/{2^{n'}})(n-1)$. If she immunizes her connectivity benefit would become $n-1$ but she
has to pay a price of $\b$ which is strictly higher than $(n-1)/2^{n'}$. So no leaf vertex would prefer to immunize.

To show that the tree in Figure~\ref{fig:max-carnage-ne-disruption} is not a Nash equilibrium with respect to the \maxcarnage~adversary, we show that any vertex in 
height 2 can strictly increase her utility by changing her immunization. Her utility before the deviation
is $(n-1)-2\c-\b$ and her utility after the deviation is $(1-1/(2^{n'}+1))(n-1) -2\c$. So as long as 
$\b > (n-1)/(2^{n'}+1) = 2^{n'+1}/(2^{n'}+1) = 2 - 2/(2^{n'}+1)$, then this deviation is beneficial for 
a height 2 vertex.
\end{proof}

\begin{example}
\label{example:car-ne-dis}
There exists a Nash equilibrium network with respect to the \maxcarnage~adversary which is not a Nash equilibrium 
network with respect to the~\maxdisrupt~adversary.
\end{example}
\begin{proof}
Consider two cycles of alternating immunized and vulnerable vertices of size $2k$ each which are connected
to each other through a vulnerable vertex outside of the cycles (so $n=4k+1$). 
Suppose the edges are purchased clockwise in the cycles and the outside vulnerable vertex purchases an edge to an 
immunized vertex in each of the cycles (see Figure~\ref{fig:max-disruption-ne-carnage}). 
 Then for $\b \in  (2, n/4)$, $\c \in [1, n/4-2)$ and $k\geq 3$ this configuration is a Nash equilibrium 
with respect to the \maxcarnage~adversary but not with respect to the \maxdisrupt~adversary.

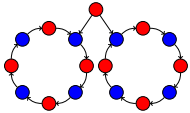
\begin{figure}[ht]
\centering\begin{minipage}[c]{.25\textwidth}
\scalebox{.5}{
\centering
\begin{tikzpicture}
[scale = 0.50, every node/.style={circle, fill=red, draw=black}, gray node/.style = {circle, fill = blue, draw}]
\node[gray node] (1) at  (0.59, 0.59){};\node[gray node] (3) at  (0.59, 3.41){};
\node[gray node] (5) at  (3.41, 3.41){};\node[gray node] (7) at  (3.41, 0.59){};
\node (2) at  (0, 2){};\node (4) at (2, 4){};
\node (6) at (4, 2){};\node (8) at (2, 0){};
\draw[->] (1) [out=135,in=270] to (2);\draw[->] (2) [out=90,in=225] to (3);
\draw[->] (3) [out=45,in=180] to (4);\draw[->] (4) [out=0,in=135] to (5);
\draw[->] (5) [out=315,in=90] to (6);\draw[->] (6) [out=270,in=45] to (7);
\draw[->] (7) [out=225,in=0] to (8);\draw[->] (8) [out=180,in=315] to (1);
\node[gray node] (9) at  (5.59, 0.59){};\node[gray node] (11) at  (5.59, 3.41){};
\node[gray node] (13) at  (8.41, 3.41){};\node[gray node] (15) at  (8.41, 0.59){};
\node (10) at  (5, 2){};\node (12) at (7, 4){};
\node (14) at (9, 2){};\node (16) at (7, 0){};
\draw[->] (9) [out=135,in=270] to (10);\draw[->] (10) [out=90,in=225] to (11);
\draw[->] (11) [out=45,in=180] to (12);\draw[->] (12) [out=0,in=135] to (13);
\draw[->] (13) [out=315,in=90] to (14);\draw[->] (14) [out=270,in=45] to (15);
\draw[->] (15) [out=225,in=0] to (16);\draw[->] (16) [out=180,in=315] to (9);
\node (17) at  (4.5, 5){};
\draw[->] (17) to (5);\draw[->] (17) to (11);
\end{tikzpicture}}
\end{minipage}
\begin{minipage}[c]{.55\textwidth}
\centering
\caption{An example of Nash equilibrium network with respect to the \maxcarnage~adversary 
which is not a Nash equilibrium with respect to the \maxdisrupt~adversary. $\c=1$ and $\b=3$. 
\label{fig:max-disruption-ne-carnage}}
\end{minipage}
\end{figure}

We first show that the configuration in Figure~\ref{fig:max-disruption-ne-carnage} is an equilibrium
with respect to the \maxcarnage~adversary. We consider three types of vertices: (1) immunized vertices
(2) the targeted vertex that connects the two cycles (henceforth the \emph{connecting vertex}) 
and (3) other targeted vertices.

First, consider an immunized vertex.\footnote{It suffices to only consider 
the immunized vertex adjacent to the connecting vertex because
every deviation of such vertex would result in at least the same utility as
the same deviation for any other immunized vertex.}
The utility of such vertex is $(n-1)2k/(2k+1)+(n-1)/(2(2k+1))-\b-\c$.
Her possible deviations are as follows.
\begin{enumerate}
\item dropping her purchased edge.
\item keeping her purchased edge and adding new edge(s).
\item swapping her purchased edge.
\item swapping her purchased edge and adding new edge(s).
\item changing her immunization.
\item dropping her purchased edge and changing her immunization.
\item keeping her purchased edge, adding new edge(s) and changing her immunization.
\item swapping her purchased edge and changing her immunization.
\item swapping her purchased edge, adding new edge(s) and changing her immunization.
\end{enumerate}

Comparing the utility before and after deviation shows that the deviation in case 1 is not beneficial.
\begin{align*}
k&\ge 3\implies k^2>k\implies \left(\frac{2k}{2k+1}\left(n-1\right)+\frac{1}{2k+1}\left(\frac{n-1}{2}\right)\right)-\b-\c \\
&>
\left(\frac{k}{2k+1}\left(n-1\right)+\frac{1}{2k+1}\left(\frac{n-1}{2}\right)+\frac{1}{2k+1}\left(\sum_{i=1}^{k}\left(\frac{n-1}{2}+1+2i-1\right)\right)\right)-\b.
\end{align*}

Consider case 2. An immunized  vertex might lose
connectivity to some part of the surviving network in case the connecting vertex gets attacked. But this happens with probability of only $1/(2k+1)$.
However, the extra connectivity benefit from purchasing one additional edge is at most $2k$ (which happens when the edge
is bought from an immunized vertex in one cycle to another immunized vertex in the other cycle). Since $\c>1$, 
purchasing this additional edge would never be beneficial. Similarly, purchasing more than one additional edges
would not be beneficial as well because it is easy to observe that with only one additional edge each vertex (when survived)
will remain connected to every other surviving vertex.

In case 3, swapping her purchased edge to any vertex in her cycle cannot be strictly better than her current edge purchase because
the current edge purchase guarantees that she remains connected to every surviving vertex in her cycle after any attack.
Swapping her purchased to the other cycle will not be beneficial either because
\begin{align*}
k&\ge 3\implies \left(\frac{2k}{2k+1}\left(n-1\right)+\frac{1}{2k+1}\left(\frac{n-1}{2}\right)\right)-\b-\c \\
&>
\left(\frac{k+1}{2k+1}\left(n-1\right)+\frac{1}{2k+1}\left(\sum_{i=1}^{k}\left(\frac{n-1}{2}+2i-1\right)\right)\right)-\b-\c.
\end{align*}

In case 4, note that two edge purchases (one in each cycle) are sufficient to connect the immunized vertex
to any other surviving vertex. So it suffices to consider the sub-case that the immunized vertex swaps her 
current edge and buys one more edge. So the deviation in case 4 cannot more beneficial than the deviation
in case 2 which is not beneficial itself.

Note that no immunized vertex with positive utility would change her immunization (regardless
of her edge purchases) because otherwise she would form the unique targeted region. 
So the deviations in cases 5, 6, 7, 8 and 9 are not beneficial because
\[
\b < \frac{n}{4} \text{ and } \c < \frac{n}{4}-2 \implies
\left(\frac{2k}{2k+1}\left(n-1\right)+\frac{1}{2k+1}\left(\frac{n-1}{2}\right)\right)-\b-\c > 0.
\]

Next, consider the connecting vertex. Her utility is $(n-1)2k/(2k+1)-2\c$. Her deviations are as follows. 
\begin{enumerate}
\item dropping any of her purchased edges.
\item keeping her purchased edges and adding new edge(s).
\item dropping any of her purchased edges and adding new edge(s).
\item swapping any of her purchased edges.
\item swapping her purchased edges and adding new edge(s).
\item immunizing.
\item dropping any of her purchased edges and immunizing.
\item keeping her purchased edges, adding new edge(s) and immunizing.
\item dropping any of her purchased edges, adding new edge(s) and immunizing.
\item swapping any of her purchased edges and immunizing.
\item swapping her purchased edges, adding new edge(s) and immunizing.
\end{enumerate}

In case 1, due to symmetry, we consider the deviation that the connecting vertex drops one of her edges. This is 
not beneficial because
\[
\c < \frac{n}{4}-2 < \frac{4k^2}{2k+1} \implies \frac{2k}{2k+1}\left(n-1\right)-2\c > \frac{2k}{2k+1}\frac{n-1}{2}-\c.
\]

In case 2, the connecting vertex 
(when survived) would remain connected to other surviving vertices. So she would not purchase any extra edges. 
Also Similar to her current strategy, with two edge purchases (one to each cycle), the connecting vertex 
would remain connected to any other surviving vertex. 
So the deviations in cases 3, 4 and 5 (which all involve buying at least two edges) are not beneficial.

In case 6,
\[
\b > 2 \geq \frac{4k}{2k+1} \implies \frac{2k}{2k+1}\left(n-1\right)-2\c > (n-1)-2\c-\b,
\]
implies that changing the immunization is not beneficial.

In case 7, the deviation is not beneficial because
\[
\c < \frac{n}{4}-2 \text{ and } \b > 1\implies \c-\b < \frac{n}{4}-3 < \frac{n-1}{2} \implies \frac{2k}{2k+1}\left(n-1\right)-2\c > \frac{n-1}{2}-\c-\b.
\]

In case 8, the connecting vertex 
(when survived) would remain connected to other surviving vertices. So she would not purchase any extra edges. 
Also as we showed in case 6 immunization is not helpful either (with two purchased edge).
Lastly, the deviations in cases 9, 10 and 11 are all dominated by the deviation in case 8. So none of them are beneficial.

Finally, consider any other targeted vertex (that is not the connecting vertex).\footnote{It suffices to only consider 
the targeted vertex that the immunized vertex adjacent to the connecting vertex have purchased an edge to.
This is because every deviation of such vertex would result in at least the same utility as
the same deviation for any other targeted vertex that is not the connecting vertex.}
Her utility is $(n-1)2k/(2k+1)-\c$. Her deviations are as follows. 
\begin{enumerate}
\item dropping her purchased edge.
\item keeping her purchased edge and adding new edge(s).
\item swapping her purchased edge.
\item swapping her purchased edge and adding new edge(s).
\item immunizing.
\item dropping her purchased edge and immunizing.
\item keeping her purchased edge, adding new edge(s) and immunizing.
\item swapping her purchased edge and immunizing.
\item swapping her purchased edge, adding new edge(s) and immunizing.
\end{enumerate}

Comparing the utility before and after deviation shows that the deviation in case 1 is not beneficial.
\begin{align*}
k&\ge 3 > 0\implies \left(\frac{2k}{2k+1}\left(n-1\right)\right)-\c \\
&>
\left(\frac{k}{2k+1}\left(n-1\right)+\frac{1}{2k+1}\left(\sum_{i=1}^{k-1}\left(\frac{n-1}{2}+1+2i\right)\right)\right).
\end{align*}

In case 2, the targeted vertex might lose
connectivity to some part of the surviving network in case the connecting vertex gets attacked. Again this happens with probability of only $1/(2k+1)$.
However, the extra connectivity benefit from purchasing one additional edge is at most $2k$ (which happens when the edge
is bought from an immunized vertex in one cycle to another immunized vertex in the other cycle). Since $\c>1$, 
purchasing this additional edge would never be beneficial. Similarly, purchasing more than one additional edges
would not be beneficial as well because it is easy to observe that with only one additional edge each vertex (when survived)
will remain connected to every other surviving vertex.

In case 3, swapping her purchased edge to any vertex in her cycle cannot be strictly better than her current edge purchase because
the current edge purchase guarantees that she remains connected to every surviving vertex in her cycle after any attack.
Swapping her purchased to the other cycle will not be beneficial either because
\begin{align*}
k&\ge 3\implies \left(\frac{2k}{2k+1}\left(n-1\right)\right)-\c >
\left(\frac{k+1}{2k+1}\left(n-1\right)+\frac{1}{2k+1}\left(\sum_{i=1}^{k-1}\left(\frac{n-1}{2}+1+2i\right)\right)\right)-\c.
\end{align*}

In case 4, note that two edge purchases (one in each cycle) are sufficient to connect the targeted vertex
to any other surviving vertex. So it suffices to consider the sub-case that the targeted vertex swaps her 
current edge and buys one more edge. So the deviation in case 4 cannot more beneficial than the deviation
in case 2 which is not beneficial itself.

Changing the immunization is not beneficial in case 5 because
\begin{align*}
\b & > 2 \geq \frac{4k-1}{2k+1} \implies\left(\frac{2k-1}{2k+1}(n-1)+\frac{1}{2k+1}\frac{n-1}{2}\right)-\c\\
 &> \left(\frac{2k-1}{2k}(n-1)+\frac{1}{2k}\frac{n-1}{2}\right)-\c-\b.
\end{align*}

In case 6, 
\begin{align*}
\c < \frac{n}{4}-2 \text{ and } \b &> 1\implies \c-\b < \frac{n}{4}-3 \\
\left(\frac{2k}{2k+1}\left(n-1\right)\right)-\c &>
\left(\frac{k}{2k}\left(n-1\right)+\frac{1}{2k}\left(\frac{n-1}{2}\right)+\frac{1}{2k}\left(\sum_{i=1}^{k-1}\left(\frac{n-1}{2}+1+2i\right)\right)\right)-\b,
\end{align*}
so the deviation is not beneficial.

In case 7, with only one additional edge (to the other cycle) she remains connected to every surviving vertex. But this deviation is not beneficial
since 
\[
\c+\b >3 \implies \left(\frac{2k}{2k+1}(n-1)\right)-\c > (n-1)-2\c-\b.
\]

Cases 8 and 9 are not beneficial because their analog cases 3 and 4 were not beneficial and the cost of 
immunization dominates the extra connectivity benefit achieved after immunization.

Finally, it is easy to see that the configuration in Figure~\ref{fig:max-disruption-ne-carnage} is not an equilibrium with respect to the~\maxdisrupt~adversary since
the vertex connecting the two cycles would be the unique targeted region; a property that does not hold in equilibria
as stated in Lemma~\ref{lem:tree3}.
\end{proof}
\section{Original Reachability Network Formation Game}
\label{sec:no-attack}
In this section, for completeness, we prove properties of the
equilibrium networks in the \emph{original reachability network
  formation game}~\cite{BalaG00}.  This case coincides with setting
$\b$ equal to zero with respect to the \maxcarnage~adversary in our formulation.  
Obviously, when $\b=0$, it is a
\emph{dominant} strategy for any player to immunize.
\citet{BalaG00} state that for for a wide range of edge
purchasing cost $\c$, any equilibrium network is either a tree or the
empty network.
\begin{pro}[Proposition 4.1 of \citet{BalaG00}]
\label{thm:no-attack-tree-forward}
When~$\c\in(0, n-1)$, every network that forms in an equilibrium of
the original reachability game is either a tree or the empty network.
\end{pro}
\begin{proof}
  First observe that when $\c\in(1, n-1)$, the empty network
  is an equilibrium. Next, consider any equilibrium that is not the
  empty network when $\c\in(0, n-1)$. We claim such network
\begin{itemize}
\item cannot have any cycles.
\item cannot have more than one connected component.
\end{itemize}
These two properties imply that 
the equilibrium is indeed a tree as claimed. 

So, first, suppose there is an equilibrium network that has a cycle.
Pick any edge on this cycle.  Clearly, the player who purchased this
edge would remain connected to any other player that she was connected
to if she drops her edge. Since $\c>0$ she would strictly increase her
utility by dropping this edge, contradicting the assumption that the
network was an equilibrium.

Second, assume the equilibrium has more than one connected
component. Now, note that the size of each connected component that
has an edge should be at least $\c+1$. Otherwise, any vertex that is
in a component with size strictly less than $\c+1$ who purchased an
edge would strictly increase her utility by dropping all her purchased
edges.  Since the graph is non-empty there exists a connected
component in that network.  It is easy to see that any vertex in any
other connected component (which can be a singleton vertex) would
strictly increase her utility by purchasing an edge to the mentioned
connected component (since the size of the mentioned connected
component is at least $\c+1$). This contradicts the assumption that
the network was an equilibrium.
\end{proof}

Proposition~\ref{thm:no-attack-tree-forward} implies the following
about the welfare of the reachability game at equilibrium.
\begin{cor}
\label{cor:welfare-no-attack}
In the original reachability game when $\c\in(0,n-1)$, the maximum social
welfare achieved in \emph{any} equilibrium is $n^2-\c (n-1)$.
\end{cor}
\begin{proof}
  By proposition~\ref{thm:no-attack-tree-forward}, there are only two
  types of equilibrium networks. The empty graph equilibrium has a social
  welfare of $n$. Any tree equilibrium has a social welfare of
  $n^2-\c (n-1)$ which is straitly bigger than the social welfare
  of the empty graph equilibrium when $\c\in(0, n-1)$.
\end{proof}

We complement Proposition~\ref{thm:no-attack-tree-forward} by showing
that, indeed for \emph{any} tree, there exist a wide range of 
edge purchasing cost $\c$ and a specific edge
purchasing pattern that make that tree an equilibrium in the original reachability game.
\begin{pro}
\label{thm:no-attack-tree-backward}
For any tree on $n$ vertices and $c\in(0, n/2)$, there exists an edge
purchasing pattern which makes that tree an equilibrium of the original
reachability game.
\end{pro}

\begin{proof}
  Given any tree, pick any vertex which satisfies the following
  property as the root of the tree: no sub-tree of this root has size
  bigger than $n/2$ (see Lemma~\ref{lem:tree-root} for a proof that
  such vertex exists).  We claim that the pattern that every vertex
  buys an edge to its parent in this tree is an equilibrium. First of
  all it is easy to see that in the construction, the root does not
  purchase any edges and she is also connected to any other vertex in
  the network, so she does not want to purchase any edges. Now
  consider any other vertex that is not the root.  This vertex
  purchases \emph{exactly} one edge in the construction. If she drops
  that edge, her expenditure decreases by $\c$. However, she loses
  connectivity to at least $n/2$ vertices. Since $\c < n/2$, this
  deviation will strictly decrease her payoff. Finally, it is easy to
  see that such vertex does not benefit by purchasing more edges or
  changing her edge purchasing decision (i.e., dropping her currently
  purchased edge and buying one or more edges to other vertices in the
  network) because she is already connected to every other vertex in
  the network using a single edge. This completes the proof.
\end{proof}

Note that the range of edge purchasing cost in
Proposition~\ref{thm:no-attack-tree-backward} is a strict subset of
the range in Proposition~\ref{thm:no-attack-tree-forward}. This is
because for higher edge purchasing costs, only specific (and not all)
trees can form in an equilibrium. 
We wrap up this section by presenting Lemma~\ref{lem:tree-root}, which
we used in the proof of Proposition~\ref{thm:no-attack-tree-backward}.
\begin{lemma}[\citet{Jordan1869}]
\label{lem:tree-root}
Consider a graph $G=(V,E)$ where $|V|=n$. If $G$ is a tree, then there
exists a vertex $v\in V$ such that rooting the tree on $v$, no sub-tree 
has size more than $n/2$.
\end{lemma}
\section{Diversity in Equilibrium}
\label{sec:missing-proofs}
We showed in Section~\ref{sec:eq-ex} that the equilibria of our
network formation game can be quite diverse. In this section we
present the examples of this diversity more formally. 
We point out that we focus on \maxcarnage~adversary 
throughout this section and 
all the results hold for Nash, swapstable and linkstable
equilibria. Also similar to Section~\ref{sec:eq-ex} blue and red
vertices denote immunized and targeted vertices in the figures,
respectively. Moreover, directed edges are used in the figures to determine the
players who purchased the edges in the network. We proceed to present some showcases of the equilibria of
our game in the remainder of this section.
 
\subsection{Empty Graphs}
We first show that, when we focus on the \maxcarnage~adversary, an empty graph with all immunized vertices
or an empty graph with all targeted vertices can form in equilibria of our game.
Furthermore, we show that these are the only empty equilibrium networks of our game.
\begin{lemma}
There exists a range of values for parameters $\c$ and $\b$ such that the empty graph
is a (Nash, swapstable or linkstable) equilibrium network
with respect to the \maxcarnage~adversary.
\end{lemma}
\begin{proof}
First, it is easy to check that the empty network with all targeted
vertices is an equilibrium when $\c\ge1$ and $\b \geq 1/n$.  No player
would strictly prefer to purchase an edge, immunize or do them
both. Also, when $\c\ge1$ and $\b \leq 1/n$, the empty network with
all immunized vertices is an equilibrium. This shows that regardless
of value of $\b$ when $\c\ge1$, the empty network is an
equilibrium. 
\end{proof}

\begin{lemma}
Let $G$ be a (Nash, swapstable or linkstable) equilibrium network
with respect to the \maxcarnage~adversary. If $G$
is the empty network, then the vertices in $G$ are either all immunized or all targeted.
\end{lemma}
\begin{proof}
For contradiction, assume an empty equilibrium network with both targeted 
and immunized vertices.
  Let $k > 0$ denote the number of players that are targeted.  Since we are in an
  equilibrium, any immunized player (weakly) prefers immunization to remaining
  targeted:
\[
1 - \b \geq \left(1-\frac{1}{k+1}\right) \implies \b \leq \frac{1}{k+1}.
\]
Similarly, any targeted player (weakly) prefers to remain
targeted compared to immunizing herself:
\[
\left(1-\frac{1}{k}\right) \geq 1 - \b \implies \b \geq \frac{1}{k}.
\]
which contradicts with the range of $\b$ in the previous equation.
\end{proof}
\subsection{Trees}
As we pointed out earlier, all the nonempty equilibria of the original reachability
game are tress.  In this section we show that trees can also form in
the equilibria of our game when we focus on the \maxcarnage~adversary.  In particular, we show two specific tree
constructions: one in which all the vertices are immunized and the
other in which all the leaves are targeted.  Note that there are
indeed tree equilibria that fall outside of these two categories (as
some of the leaves can be targeted and some can be immunized).

We start by showing that a tree with all immunized vertices can form in  equilibria.
\begin{lemma}
\label{lem:eq-tree-1}
Consider any tree on $n$ vertices. Suppose $\c\in(0,n/2)$ and $\b\in(0,n/2)$. Then, 
there exists an edge 
purchasing pattern which makes that tree an equilibrium with respect to the \maxcarnage~adversary when all the vertices are immunized.
\end{lemma}
\begin{proof}
Since all the vertices are immunized and have a payoff of $n-\c-\b > 0$, no player would
change her immunization decision (regardless of how she changes her edge purchases), because 
if so, she would form the unique largest targeted region and will be killed by the adversary.

Now that immunization decision are fixed, it is easy to see that rooting the tree as in 
Proposition~\ref{thm:no-attack-tree-backward}
and the pattern of purchasing an edge towards the root will result in an equilibrium network.
\end{proof}

We then show that the equilibria proposed in Lemma~\ref{lem:eq-tree-1} can be used as 
a black-box to prove new tree equilibria in our game.
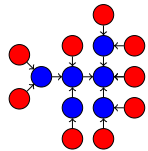
\begin{figure}[h]
\centering
\begin{minipage}[c]{.3\textwidth}
\scalebox{.75}{
\begin{tikzpicture}
[scale=0.55, every node/.style={circle,fill=red, draw=black}, gray node/.style = {circle, fill = blue, draw}]
\node[gray node] (1) at  (0, 0){};\node[gray node] (2) at  (1, 0){};\node[gray node] (3) at  (2, 0){};
\node[gray node] (4) at  (2, 1){};\node[gray node] (5) at  (1, -1){};\node[gray node] (6) at  (2, -1){};
\node (7) at  (-.71, .71){};\node (8) at  (-.71, -.71){};\node (9) at  (1, 1){};
\node (10) at  (2, 2){};\node (11) at  (3, 1){};\node (12) at  (3, 0){};
\node (13) at  (3, -1){};\node (14) at  (2, -2){};\node (15) at  (1, -2){};
\draw[->] (7) to (1);\draw[->] (1) to (2);\draw[->] (2) to (3);
\draw[->] (13) to (6);\draw[->] (6) to (3);\draw[->] (8) to (1);
\draw[->] (9) to (2);\draw[->] (12) to (3);\draw[->] (10) to (4);
\draw[->] (4) to (3);\draw[->] (12) to (3);\draw[->] (11) to (4);
\draw[->] (14) to (6);\draw[->] (15) to (5);\draw[->] (5) to (2);
\end{tikzpicture}}
\end{minipage}
\begin{minipage}[c]{.4\textwidth}
\caption{An example of tree equilibrium with respect to the \maxcarnage~adversary when $\c=2$ and $\b=1.9$. 
\label{fig:tree-2}}
\end{minipage}
\end{figure}

\begin{lemma}
\label{lem:eq-tree-2}
Consider any tree on $k\le n/2$ immunized vertices. Add $n-k$ targeted leaf vertices to this tree such that
every immunized vertex has at least one targeted neighbor (see Figure~\ref{fig:tree-2}). 
Then, for $\c\in(0, k/2)$, $\b\in\left((n-1)/(n-k), k/2-1\right)$ and 
$k\ge 7$, there exists an edge purchasing pattern the makes this network an equilibrium
with respect to the \maxcarnage~adversary.
\end{lemma}
\begin{proof}
Root the tree of immunized vertices as described in Proposition~\ref{thm:no-attack-tree-backward} and 
let any immunized vertex to purchase an edge towards the root. Also let all the targeted vertices to buy
the edge that connects them to the immunized tree. 

Consider any immunized vertex. She is connected to all the surviving vertices after any attack and she has
done so with (at most) one edge purchase. She would not change her immunization decision because she 
would form the unique targeted vertex. She would not add any more edges either and it is easy to see 
that her current edge purchase is the best in case she can only purchase one edge.

Consider any targeted vertex. She is connected to all the surviving vertices after any attack and she has
done so with only edge purchase. She would not change her immunization decision either because the 
cost of immunization dominate the benefit she would get for surviving with probability one.  
She would not add any more edge either and it is easy to see 
that her current edge purchase is the best in case she can only purchase one edge.
\end{proof}

\paragraph{Hub-Spoke}
$\newline$
We now prove that a hub-spoke network can form in the equilibria of our game when focusing on to the \maxcarnage~adversary. Although
a hub-spoke network is just a special case of the tree networks described in Lemma~\ref{lem:eq-tree-2},
we provide a separate proof for this network which applies to a wider range of parameters in 
comparison to Lemma~\ref{lem:eq-tree-2}.

Hub-spoke network is an interesting equilibrium because it satisfies all of the following properties at the same: \emph{(i)} it is also
an equilibrium in the original reachability game, \emph{(ii)} the network is very efficient because there are 
no more linkage than a tree (a minimum required to connect all the players) and the number of immunized vertices is \emph{only} one and 
\emph{(iii)} the welfare at equilibrium is high, because all the surviving vertices remain connected after each attack.
\begin{figure}[h]
\centering\begin{minipage}[c]{.3\textwidth}
\scalebox{.75}{
\begin{tikzpicture}
[scale=0.55, every node/.style={circle,fill=red, draw=black}, gray node/.style = {circle, fill = blue, draw}]
\node[gray node] (1) at  (0, 0){};
\node (2) at  (-2, 0){};\node (3) at  (-1.41, 1.41){};\node (4) at  (0, 2){};
\node (5) at  (1.41, 1.41){};\node (6) at  (2, 0){};\node (7) at  (1.41, -1.41){};
\node (8) at  (0,-2){};\node (9) at  (-1.41, -1.41){};
\draw[->] (2) to (1);\draw[->] (3) to (1);\draw[->] (4) to (1);
\draw[->] (5) to (1);\draw[->] (6) to (1);\draw[->] (7) to (1);
\draw[->] (8) to (1);\draw[->] (9) to (1);
\end{tikzpicture}}
\end{minipage}
\begin{minipage}[c]{.4\textwidth}
\caption{An example of hub-spoke equilibrium with respect to the \maxcarnage~adversary when $\c=1$ and $\b=1$
\label{fig:eq-star}}
\end{minipage}
\end{figure}
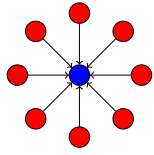
\begin{lemma}
\label{lem:eq-star}
If $\c\in(0,n-3]$ and $\b\in[1, n-1]$ then a hub-and-spoke network is an
equilibrium with respect to the \maxcarnage~adversary when the hub immunizes and the spokes buy the edges to the
hub.
\end{lemma}
\begin{proof}
  First observe that the expected size of the connected component of
  immunized and targeted vertices are $n-1$ and $(1-1/(n-1))(n-1)=n-2$,
  respectively.  Also the expenditure of immunized and targeted
  vertices are $\b$ and $\c$, respectively.

  Let us first consider the hub. The hub does not want to buy another
  edge because she already has an edge to any other vertex in the
  graph. So her only possible deviation is to change her immunization.
  If she changes her immunization decision, she will be the part of
  the unique targeted region. So her utility becomes zero while her
  expenditure becomes zero as well. So to prefer to not change her
  immunization decision her utility should be at least zero in the current strategy.
\[
\b \leq n-1 \implies (n-1) - \b \geq 0.
\]

Now consider any spoke vertex. Since the network is symmetric, it only
suffices to consider the possible deviations of one spoke vertex. The
spoke has the following deviations which we analyze one by one.
\begin{enumerate}
\item immunizing.
\item dropping her purchased edge.
\item dropping her purchased edge and immunizing.
\item dropping her purchased edge and adding new edge(s).
\item dropping her purchased edge, adding new edge(s) and immunizing.
\item adding more edges.
\item adding more edges and immunizing.
\end{enumerate}
In each case we consider the utilities before and after the deviation and show that
given the conditions in the statement of Lemma~\ref{lem:eq-star}, the deviation in not beneficial.

In case 1, the spoke survives after immunizing but has to pay the price of immunization.
\[
\b \geq 1\implies (n-2) - \c \geq (n-1) - \b - \c.
\]

In case 2, the spoke becomes disconnected from the rest of the network.
\[
\c \leq n-3 < n-3+\frac{1}{n-1} \implies (n-2)- \c \geq 1-\frac{1}{n-1}.
\]

In case 3, the spoke survives after immunization but similar to case 2 she becomes disconnected
from the rest of the network.
\begin{align*}
\c \leq n-3 \text{ and } \b \geq 1 &\implies \c - \b \leq n-4 < n-3
\implies (n-2) - \c \geq 1 - \b.
\end{align*}

In case 4, adding more edges will cause the vertex to form the unique
targeted region so as long as the her current utility before
deviation is non-negative, she would not drop her edge and add any new
edges.
\[
\c \leq n-3 < n-2 \implies (n-2) - \c > 0.
\]

In case 5, the spoke definitely survives after immunizing
but her payoff in case of survival is capped at $n-1$ (which will
happen if she purchases at least 2 edges to 2 other spokes).  
Suppose she adds
$i \geq 1$ edges. Then,
\[
\c\geq 0 \text{ and }\b \geq 1 \implies (n-2) - \c \geq n-1 - i \c - \b.
\]

In case 6, adding more edges will only cause the vertex to form the
largest targeted region which cannot happen in any equilibrium by Lemma~\ref{lem:tree2}. 
So similar to case 4, she does not add
more edges if her current payoff is strictly greater than zero.

In case 7, after the vertex immunizes, she becomes connected to every vertex that survives. 
So adding more edges strictly
decrease her payoff since $\c>0$. So as long as the condition for case 1 holds, the
spoke does not deviate in case 7 either.
\end{proof}
\subsection{Forest}
Nonempty networks with multiple connected components (with no immunized vertex) 
can form in the equilibria of our game when we focus on the \maxcarnage~adversary.
We start by showing that a forest consisting of targeted trees of equal size can form in equilibria.
\begin{lemma}
  Let $n=k F$. Then $k$ disjoint targeted trees of size $F$ can form in the
  equilibrium with respect to the \maxcarnage~adversary when $\c\in(0, F/4]$, $\b\geq (k-7/4) F$,
  $k\geq 4$ and $F\geq 2$.
\label{lem:eq-forest}
\end{lemma}	
\begin{proof}
Similar to the construction proposed in Proposition~\ref{thm:no-attack-tree-backward} in each tree, 
we can fix a root and guarantee that each player only purchases one edge (towards the root). 
Hence, the possible deviations of any player who purchased an edge are as follows:
\begin{enumerate}
\item dropping her purchased edge.
\item dropping her purchased edge and immunizing.
\item dropping her purchased edge and adding more edge(s).
\item dropping her purchased edge, adding more edge(s) and immunizing.
\item immunizing.
\item adding more edges.
\item adding more edge(s) and immunizing.
\end{enumerate}
By construction, the root of the tree is the only vertex that does not purchase any edge. 
Hence, the deviations of the root is 
only limited to cases 5-7. So in the first four cases we only consider a vertex who purchased 
an edge.

In case 1, dropping the edge will cause the vertex to survive but her connectivity benefit
decreases but it is at most $F/2$ due to Lemma~\ref{lem:tree-root}.
\begin{align*}
\c \leq  \frac{F}{4} \text{ and } k\ge 4 \implies (1-\frac{1}{k})F - \c \ge \frac{F}{2}.
\end{align*}

In case 2, when the vertex drops her purchased edges, she is not part of any targeted region anymore,
so she survives with probability $1$ even without immunization. So as long as the deviation in case 1
is not beneficial, the deviation in case 2 is not beneficial either. 

In case 3, if the vertex drops her purchased edge and buys any edge(s) to any other connected component
she will form the unique targeted region which cannot happen in any equilibrium by Lemma~\ref{lem:tree2}. 
So the only other case to consider is when she drops her edge and buys edge(s) to the same 
connected component she was a part of. Since 
she only requires one edge to connect to the component,
she does not have a better deviation in this case either.

In case 4, 
note that the payoff of a vertex who \emph{does not} purchase an edge and deviates according to case 6 is strictly better
than the deviation in this case. So showing that the deviation in case 6 is not beneficial (as we do shortly) is sufficient to show that the deviation
in case 4 is not beneficial either.

For cases 5-7, first consider a vertex who purchased one edge.
In case 5, when she immunizes, she survives with probability one.
\begin{align*}
\b&\ge (k-\frac{7}{4})F \text{ and } k \geq 4 \implies \b> \frac{3F}{4}
\implies (1-\frac{1}{k}) F - \c >  F-\c-\b.
\end{align*}

In case 6, adding more edges to other components will only result in forming the unique targeted region 
which cannot happen in any equilibrium by Lemma~\ref{lem:tree2}. Also, the vertex is already connected 
to the vertices in her component and any edge beyond a tree is redundant.

Finally, in case 7, when the vertex immunizes, she survives with probability of 1. Furthermore, since $\c<F/4$
she would benefit the most by buying an edge to any connected component she is not connected to.
\begin{align*}
\c &> 0 \text{ and } \b \ge (k-\frac{7}{4})F
\implies (1-\frac{1}{k}) F - \c \ge (k-1)F - k \c - \b.
\end{align*}

Now we consider cases 5-7 for  the vertex that did not purchase any edge. 
It's easy to verify that the argument in case 6 still holds. The argument in cases 5 and 7 also
hold with the only difference that now we have to subtract a $\c$ from both sides of the final 
inequalities in both cases, since the vertex initially has not purchased any edges.
\end{proof}

The forest network in Lemma~\ref{lem:eq-forest} was symmetric (all the trees have the 
same size). We now assert an example of non-symmetric forests that can form in equilibria.

\begin{lemma}
\label{lem:eq-forest-2}
  Let $n=k F+n'$. Then $k$ disjoint targeted trees of size $F$ along with $n'$ 
  vulnerable singleton vertices can form in the equilibrium with respect to the \maxcarnage~adversary if $\c\in(1, F/4)$, $\b\geq (k-1) (F-1)$,
  $k\ge 4$, $F\ge 5$ and $n'\ge 0$.
\end{lemma}	

\begin{proof}
Similar to the construction proposed in Proposition~\ref{thm:no-attack-tree-backward}, for each of the trees, 
we can fix a root and guarantee that each player only purchases one edge (towards the root). 

First, observe that no vertex would like to add or drop an edge without changing her immunization decision.
The singleton vertices, would not like to buy an edge to another singleton vertex (since $\c>1$) nor buy an
edge to a vertex in the trees (since they would form the unique targeted region). Vertices in
of a tree would not want to add an edge to another vertex either, they are already connected to other vertices in
their tree and would form the unique targeted region if they buy an edge to a vertex outside of their tree.
Furthermore, they would not drop their only purchased edge because
\begin{align*}
\c \leq  \frac{F}{4} \text{ and } k\ge 4 \implies (1-\frac{1}{k})F - \c \ge \frac{F}{2},
\end{align*}
where again $F/2$ is the maximum connectivity benefit they can receive according to Lemma~\ref{lem:tree-root}.

Now we show that no vertex in any of the trees could immunize and add more edges to strictly increase her 
expected payoff. We only consider the root vertices, since they are the only vertices who has not
purchased an edge yet (so given the same deviation, the payoff of no vertex who purchased an edge can be higher than the root).
We also point out that the most beneficial deviation is when the root buys an edge to every other tree (since
she pays $\c$ for an edge but gets a benefit of $(1-1/(k-1))F>\c$). So we have,
\begin{align*}
\c &\geq 1 \text{ and } \b \ge (k-1)(F-1)\implies (1-\frac{1}{k})F\geq (k-1)F-(k-1)\c-\b.
\end{align*}
Similarly, a singleton vulnerable vertex would like to buy an edge to every tree after immunization in her best deviation
which is still not beneficial.
\begin{align*}
\c &\geq 1 \text{ and } \b \ge (k-1)(F-1)\implies 1\geq (k-1)F-(k-1)\c-\b.
\end{align*}
\end{proof}

We conclude this section by pointing out that other non-symmetric forest equilibria can form
in the equilibria of our game e.g., when
there are 
some targeted trees along with vulnerable trees of smaller size and singleton vertices.
\subsection{Cycles}
We now assert that unlike the original reachability game, cycles can form in the equilibria of our game
when we focus on the \maxcarnage~adversary. 
Indeed, we show that an alternating cycle of immunized and targeted vertices can form in equilibria.
\begin{lemma}
\label{lem:eq-cycle}
A cycle of $n=2k$ alternating immunized and targeted vertices can form in
equilibria with respect to the \maxcarnage~adversary when \emph{(i)} every vertex buys an edge to the vertex in her clockwise
direction in the cycle, and \emph{(ii)} $\c \in (1, n/2-2)$, $\b\in (2, n/2+1)$ and $k\ge 4$.
\end{lemma}
\begin{proof}
  The proof is by case analysis for immunized and targeted
  vertices, respectively. First observe that the expected size of the connected
  component of immunized and vulnerable vertices are $n-1$ and
  $(1-1/k)(n-1)$, respectively.  Also the expenditure of immunized and
  vulnerable vertices are $\c+\b$ and $\c$, respectively.

  Let's start with an immunized vertex.  Note that if an immunized vertex
  changes her immunization decision, she will deterministically get
  killed by the adversary regardless of her edge purchasing decision
  because the immunized vertex is already connected to a targeted
  vertex. So as long as the payoff of an immunized vertex is greater than
  zero, she cannot change her immunization decision.
\begin{align*}
&  \c< \frac{n}{2}-2 \text{ and }  \b <  \frac{n}{2}+1
\implies (n-1)-\b-\c > 0.
\end{align*}
So for an immunized vertex, given that she never benefits by changing
her immunization decision, the possible deviations are as follows.
\begin{enumerate}
\item dropping her purchased edge and adding new edge(s).
\item dropping her purchased edge.
\item adding more edge(s).
\end{enumerate}
We compare the utilities after and before any of these
deviation and show that given the conditions in the 
statement of Lemma~\ref{lem:eq-cycle}, none of these deviations are beneficial.

In case 1, after purchasing $i \ge 1$ edge(s) the expected size of the connected 
component of the immunized vertex is (trivially) at most $n-1$. 
So the deviation is not beneficial because she is currently achieving the 
same expected size using only edge purchase.

In case 2, 
\begin{align*}
&\c < \frac{n}{2}-2 < \frac{n}{2} - 1 
\implies (n-1)-\c-\b > \frac{1}{k}\Big(1+3+\ldots + (2k-1)\Big)-\b.
\end{align*}
In case 3, no matter which new edge(s) the immunized vertex purchases,
her expected connected component has size at most $n-1$. So
adding more edges will only increase the expenditure.

Now consider a targeted vertex. Since the network is symmetric, it
suffices to consider one such vertex. Her deviations are as follows.
\begin{enumerate}
\item dropping her purchased edge and adding new edge(s).
\item dropping her purchased edge.
\item adding more edge(s).
\item immunizing.
\item dropping her purchased edge, adding new edge(s) and immunizing.
\item dropping her purchased edge and immunizing.
\item adding more edge(s) and immunizing.
\end{enumerate}
Again, we compare the utilities after and before each deviation and
show that given the conditions in the statement of
Lemma~\ref{lem:eq-cycle}, the deviations are not beneficial.

Case 1 is similar to the analysis of case 1 for the immunized vertex with the only difference
that $n-1$ should be replaced by $(1-1/k)(n-1)$.

In case 2,
\begin{align*}
&\c < \frac{n}{2} - 2 < \frac{n}{2} - 2 + \frac{2}{n}
\implies (1-\frac{1}{k})(n-1) - \c > \frac{1}{k}\Big(2 + \ldots + (2k-2)\Big).
\end{align*}

Similar to the case 3 for the immunized vertex, in case 3 
no matter which new edge(s) the targeted vertex purchases, 
her expected connected component is at most $(1-1/k)(n-1)$. So adding more edges 
will only increase her expenditure.

In case 4, the vertex will survive with probability 1, but she has to pay for immunization.
\begin{align*}
\b &> 2 > 2-\frac{1}{k}=\frac{1}{k}(n - 1)
\implies (1-\frac{1}{k})(n-1) - \c > (n-1) - \c - \b.
\end{align*}

In case 5, when 
the targeted vertex immunizes, she survives with probability one. 
Suppose she adds $i\geq 1$ edges, then
the size of her connected component after the attack is at most $n-1$.
\begin{align*}
&   \c > 1 \text{ and }  \b > 2 
\implies \b + i \c > 3 > \frac{1}{k}(n-1)
\implies (1-\frac{1}{k})(n-1) - \c > (n-1) - i \c - \b.
\end{align*}

In case 6, 
\begin{align*}
 \c < \frac{n}{2}-2 \text{ and } \b> 2
&\implies \c - \b < \frac{n}{2}-4 <  \frac{n}{2}-3+\frac{1}{k}\\
 &\implies (1-\frac{1}{k})(n-1) - \c > 
\frac{1}{k-1}\Big(2 + \ldots + (2k-2)\Big)-\b .
\end{align*}

In case 7, the targeted vertex survives after immunization. 
Suppose she adds $i\geq 1$ edges, then
the size of her connected components 
is at most $n-1$.
\begin{align*}
 \c > 1 \text{ and } \b > 2
&\implies \b + i \c > 3 > \frac{1}{k}(n-1)\\
&\implies
(1-\frac{1}{k})(n-1) - \c > (n-1) - (i+1) \c - \b.
\end{align*}
\end{proof}
\subsection{Flowers}
We next show that multiple cycles can form in equilibria when 
we focus on the \maxcarnage~adversary. The flower equilibrium
in this section is an illustration of such phenomenon. In the flower equilibrium each
petal has the same pattern of immunization as the cycle in Lemma~\ref{lem:eq-cycle}.

\begin{lemma}
\label{lem:eq-flower}
Let $n = F (2k-1)+1$. Consider a flower network containing $F$
petals (cycles) of size $2k$ where all the cycles share exactly one vertex.
Assume each petal is composed of alternating immunized and targeted
vertices, and the shared vertex is immunized. Then the flower network 
can form in the 
equilibrium with respect to the \maxcarnage~adversary when \emph{(i)} in each petal, the targeted vertices buy
both of the edges to their immunized neighbors, and \emph{(ii)}
$\b\in\left(2, (2k-1)F\right)$, $\c \in\left(0, \min\left\{(k-1)F-2, \left((k-1)^2+5\right)/(2kF)\right\}\right)$, $k\geq 2$ and
$F\geq 3$.
\end{lemma}
\begin{proof}
First note that the expected size of the connected component for
  immunized and targeted vertices are $(2k-1)F$ and
  $(1-1/(kF))(2k-1)F$, respectively.  Also the expenditure of immunized 
  and targeted vertices are $\b$ and $2 \c$, respectively.

First, consider any immunized vertex. Any such vertex is connected to
  every survived vertex in the network after any attack. So no such vertex
  wants to add any edges. Furthermore, she does not want to
  change her immunization decision because (regardless of her edge
  purchases) she will form the unique largest targeted region. So as
  long as her current payoff is bigger than zero, she would not change
  her action. This means
\[
\b < (2k-1)F \implies (2k-1)F - \b > 0.
\]

Now, consider any targeted vertex. Such vertex (when survives) is also
connected to to every survived vertex in the network after any
attack. Since she managed to do so with only two edge purchases, it
suffices for us to only consider deviations such that the number of
edges purchased by the targeted vertex is at most $2$.  So her
possible deviations are as follows.
\begin{enumerate}
\item buying two edges and immunizing.
\item buying one edge.
\item buying one edge and immunizing.
\item buying no edges.
\item buying no edges and immunizing.
\end{enumerate}
We compare the utilities before and after each deviation and show that
given that given the conditions in Lemma~\ref{lem:eq-flower} that none of the
deviations are beneficial.

Remind that the current edge purchases of any targeted vertex connect
her to any survived vertex. So if an targeted vertex is buying two
edges, she can do no better than her current purchases. So in case 1,
it suffices to check only the deviation in the immunization decision.
\begin{align*}
&\b > 2 > 2 - \frac{1}{k} 
\implies (1-\frac{1}{kF})(2k-1)F -2\c > (2k-1)F -2\c -\b.
\end{align*}

In case 2, first observe that if a targeted vertex is going to buy a
single edge, she will buy it to the central immunized vertex if she
wants to maximize her expected size of the connected component after
attack.\footnote{She would remain connected to at least $(F-1)$ of the petals if she survives.} 
Second, among all the targeted vertices in a petal, the targeted vertex with the maximum
expected size of the connected component is the vertex who is $k-1$ hops away from the 
central immunized vertex.\footnote{
Fix a petal and consider a targeted vertex who purchased an edge to the central immunized vertex.
If the attack happens outside of this petal, then the
  size of connected is the same for all the targeted vertices in the petal.  If
  the attack happens in the petal and the vertex survives, her expected utility is 
  at least half of her petal size (and sometimes more) regardless of the attack.}  
So it suffices to consider the deviation of such vertex.
\begin{align*}
\c &< \min\{(k-1)F-2,\frac{(k-1)^2+5}{2kF}\} <\frac{(k-1)^2+5}{2kF}
\implies\\
 &(1-\frac{1}{kF})(2k-1)F -2\c >\\ 
& (1-\frac{1}{kF})(2k-1)F
-\frac{1}{kF}\left(1+3+\ldots+(k-3)+1+3+\ldots+(k-1)\right)-\c.
\end{align*}

The same argument holds for case 3, with the only difference than the vertex will
survive with probability of $1$.
\begin{align*}
\c &< \min\{(k-1)F-2,\frac{(k-1)^2+5}{2kF}\} <\frac{(k-1)^2+5}{2kF}
\text{and }\b > 2-\frac{1}{kF}
\implies\\ 
&(1-\frac{1}{kF})(2k-1)F -2\c > \\
&(2k-1)F-\frac{1}{kF}\left(1+3+\ldots+(k-3)+1+3+\ldots+(k-1)\right)
-\c-\b.
\end{align*}

In case 4, she still survives with the same probability but her size of connected 
component is only 1 anytime she survives.
\begin{align*}
\c \leq \min\{(k-1)F-2,\frac{(k-1)^2+5}{2kF}\} & \leq (k-1)F-2
\leq  kF-\frac{F}{2}-\frac{3}{2} + \frac{1}{2k}+\frac{1}{2kF}\\
\implies &(1-\frac{1}{kF})(2k-1)F -2 \c \geq (1-\frac{1}{kF}).
\end{align*}

In case 5, she survives with probability 1, but the size of her connected component is 1.
\begin{align*}
\c \leq \min\{(k-1)F-2,\frac{(k-1)^2+5}{2kF}\} &\leq (k-1)F-2
\leq kF-\frac{F}{2}-\frac{1}{2} + \frac{1}{2k}\text{ and }\b > 2\\
\implies &(1-\frac{1}{kF})(2k-1)F -2\c > 1-\b.
\end{align*}
\end{proof}

The number of edges in the flower equilibrium is $n+F-1$. So to get the densest 
flower equilibrium, it suffices to set $F$ as large as possible or $k$ as small as possible. Setting $k=2$
will result in a flower equilibrium with $4n/3-O(1)$ edges.
We finally point out that among all the examples of equilibrium in this section with strictly more than $n$ edges, 
the flower is the only example that remains an equilibrium with respect to the \maxcarnage~adversary even when $\c>1$.
\subsection{Complete Bipartite Graph}
We finally show that a specific form of complete bipartite graph can form in equilibria
when we focus on the \maxcarnage~adversary.
The equilibria presented in Lemma~\ref{lem:eq-biclique} have $2n-4$ edges which shows
that our upper bound on the density of equilibria (Theorem~\ref{thm:sparse-general}) is tight.
\begin{lemma}\label{lem:eq-biclique}
  Consider a complete bipartite graph $G=(U\cup V, E)$ with $|U|=2$
  and $|V|\ge1$.  $G$ can form in the equilibrium 
  with respect to the \maxcarnage~adversary if all the vertices
  in $U$ are targeted, all the vertices in $V$ are immunized, the
  vertices in $U$ purchase all the edges in $E$, $\c\in(0,1/2]$ and
  $\b\in\left((n-1)/2, n-1\right)$.
\end{lemma}
\begin{proof}
  The proof is by case analysis for immunized and targeted
  vertices, respectively. First observe that the expected size of the connected
  component of immunized and targeted vertices are $n-1$ and
  $(n-1)/2$, respectively.  Also the expenditure of immunized and
  targeted vertices are $\b$ and $(n-2)\c$, respectively. 
  
   Consider an immunized vertex first. If she changes her immunization decision, 
   she will deterministically get killed by the adversary regardless of her edge purchasing 
   decision because the immunized vertex is already connected to a targeted vertex and hence she will
   form the unique largest targeted region. So as long as her payoff is greater than zero,
    she would not change her immunization decision.
    \[
    \b\le n-1\implies (n-1)-\b\geq 0.
    \]
   Also the immunized vertex remains connected to every vertex that survives, regardless of the attack. So she 
   would not want to by any edges.
   
   Now consider a targeted vertex. First, note that since $\c\leq 1/2$, the current utility of a targeted vertex is 
   at least $1/2$. Next, it is easy to observe that no deviation of a targeted vertex can be beneficial if she purchases
   an edge to the other targeted vertex, regardless of her choice of immunization or her other edge purchases. 
   If she does not immunize, buying an edge to the other targeted vertex will result in forming the largest unique targeted 
   region which cannot happen in any equilibrium by Lemma~\ref{lem:tree2}. If she immunizes, the other targeted
   vertex becomes the unique largest targeted region, so she would not benefit by purchasing an edge in this case
   either.
   
   This, together with the symmetry of the network with respect to immunized vertices imply that the deviations 
   of a targeted vertex that we need to consider are as follows.
   \begin{enumerate}
   \item purchasing $k\in\{0, \ldots, n-3\}$ edges to immunized vertices.
   \item purchasing $k\in\{0, \ldots, n-3\}$ edges to immunized vertices and immunizing.
   \end{enumerate}
   
   We compare the utilities of the targeted vertex before and after the deviation and show that
   none of the above deviations are beneficial.
   
   In case 1,
   \begin{align*}
   k\leq n-3 < n-2 \text{ and } \c \leq \frac{1}{2}
    &\implies (n-k-2) \c \le \frac{n-k-2}{2}\\
    &\implies\frac{n-1}{2} - (n-2)\c\geq \frac{k+1}{2}-k\c.
   \end{align*}
   
In case 2,
\begin{align*}
\c \leq \frac{1}{2} \text{ and } \b\geq \frac{n-1}{2}\ge\frac{k+1}{2}
&\implies (n-k-2)\c-\b\leq \frac{n-k-2}{2} - \frac{k+1}{2}\\
&\implies \frac{n-1}{2}-(n-2) \c\geq (k+1)-k \c-\b.
\end{align*}
\end{proof}
\section{Missing Proofs from Section~\ref{sec:welfare}}
\label{sec:missing-connectivity}

To prove Theorem~\ref{thm:connect} with respect to the \maxcarnage~adversary,
first in Lemma~\ref{lem:edge-immunized} we show that in a non-trivial equilibrium network
 every immunized vertex has an adjacent edge. Then in Lemma~\ref{lem:no-isolated-immunized}
we show that all the immunized vertices are in the same connected component of the non-trivial
equilibrium network.
\begin{lemma}\label{lem:edge-immunized}
  Let $G = (V,E)$ be a non-trivial
  Nash, swapstable or linkstable equilibrium network with respect to the \maxcarnage~adversary. 
  Then, for all $u\in \I$, there is an edge $(u,v)\in E$.
\end{lemma}
\begin{proof}
  Suppose not. Then there exists an immunized vertex $u$ with no adjacent edge.
  Since $G$ is non-trivial, there exists an edge $(x,y)\in E$. Without loss of generality
  assume that $x$ has purchased the edge $(x,y)$ and 
   let $p$ denote the probability of attack to $x$ in $G$. We know $p<1$, otherwise
  $x$ would benefit by dropping her edge to $y$.
   
  Since we are in an equilibrium $x$ does not strictly prefer to 
  drop any of her edges. Let $\mu$ and $\mu'$ denote the 
  expected connectivity benefit of $x$ before and  after the deviation that she drops her edge to $y$. Then $\mu-\c\geq \mu'\geq 1-p$.
  The last inequality comes from the fact that the size of the connected component of $x$
  after the deviation is at least $1$ and the probability of attack to
  $x$ after deviation is at most $p$. Remind that with respect to the \maxcarnage~adversary
  the attack is characterized by the size of the maximum vulnerable region. So if $x$ is targeted then
  the size of the targeted region she belongs to does not increase by dropping an edge. So the 
  probability of attack to $x$ does not increase after the deviation.
  
  Consider the deviation that $u$ purchases an edge to $x$. Since $u$ 
  is immunized, this deviation would not change the distribution of the attack.
  Therefore, the change in $u$'s expected utility after the deviation is $\mu-\c\geq 1-p$ which is strictly
  bigger than 0 since $p<1$; a contradiction.
\end{proof}

\begin{lemma}
\label{lem:no-isolated-immunized}
Suppose $G= (V,E)$ is a non-trivial Nash, swapstable or linkstable equilibrium network with respect to
the \maxcarnage~adversary.
Then all the immunized vertices of $G$ are in the same connected component.
\end{lemma}
\begin{proof}
Suppose not. Then the immunized vertices are in multiple connected components. 
Let $G_1$ and $G_2$ be two such components. Pick immunized vertices $u_1\in G_1$ and $u_2\in G_2$ arbitrarily. 

Consider an edge $(u, u_2)$ adjacent to $u_2$ which exists due to Lemma~\ref{lem:edge-immunized}. 
Let $v\in\{u, u_2\}$ denote the vertex who purchased the edge $(u, u_2)$.
Consider the deviation that $v$ drops that edge and let $\mu$ and $\mu'$ denote the connectivity benefit of $v$ before and after
the deviation, respectively. Since we are in an equilibrium, then $v$ (weakly) prefers to 
maintain this edge. So
$
\mu -\c \geq \mu'
$.
We show that $\mu' > 0$. In the case that $v=u_2$, clearly $\mu'\geq 1 > 0$ since $v$ is immunized. 
In the case that $v=u$, let $p$ denote the probability of attack to $v$ pre-deviation. 
In this case $\mu'\geq 1-p$ since the size of the connected component of $v$ after the deviation is at least 1 
and the probability of attack to $v$ would not increase after dropping an edge (see the proof of Lemma~\ref{lem:edge-immunized} for more discussion).
Finally, we know that 
$p<1$ otherwise dropping the edge to $v_2$ would improve $v$'s utility. So $\mu'\geq 1-p > 0$ in this case as well.

Finally, consider the deviation that $u_1$ purchases an edge to $u_2$. The expected utility of $u_1$ 
after the deviation will increase by at least $\mu-\c \geq \mu'> 0$ since the distribution of attack after
this deviation remains unchanged; 
a contradiction.
\end{proof}

We are now ready to prove Theorem~\ref{thm:connect}.
\begin{proof}[Proof of Theorem~\ref{thm:connect}]
Lemma~\ref{lem:edge-immunized} implies that all the immunized vertices are 
in the same connected component. 
Let $\hat{G}$ denote this component.
In the rest of the proof we show
that all the vulnerable vertices are also part of this connected component.
Hence, $\hat{G}=G$ and $G$ 
is connected

By the way of the contradiction assume there exists a vulnerable vertex $w$ outside of $\hat{G}$.
We consider two cases: (1) $w$ is not targeted or (2) $w$ is targeted.

In case (1), pick any immunized vertex $u\in\hat{G}$. $u$ has an adjacent edge $(u_1, u)$
by Lemma~\ref{lem:edge-immunized}.
Let $v\in\{u_1, u\}$ be the vertex who purchased the edge $(u_1, u)$.
Consider the deviation that $v$ drops that edge and let $\mu$ and $\mu'$ denote the connectivity benefit of $v$ before and after
the deviation, respectively. Since we are in an equilibrium, then $v$ (weakly) prefers to 
maintain this edge. So $\mu -\c \geq \mu'$.
Also $\mu' > 0$ with the exact same argument as in the proof of Lemma~\ref{lem:no-isolated-immunized}.
Now, consider the deviation that $w$ purchases an edge to $u$. The expected utility of $w$ 
after deviation will increase by at least $\mu-\c \geq \mu'> 0$ because the distribution of attack is 
unchanged after this deviation; a contradiction.

In case (2), we consider two sub-cases: 
2(a) there exists a targeted region with size strictly bigger than 1 in $G$
or 2(b) the size of all the targeted regions are exactly 1 in $G$.

In case 2(a), again consider the targeted vertex $w$
in a connected component of size bigger than 1 
and let $(w, w_1)$ denote the adjacent edge 
to $w$ which exists by the assumption of the case.
Let $v\in\{w, w_1\}$ be the vertex who purchased this edge.
Consider the deviation that $v$ drops this edge and let $\mu$ and $\mu'$
denote $v$'s connectivity benefit before and after the deviation, respectively.
Since we are in an equilibrium $v$ (weakly) prefers to keep this edge which implies $\mu-\c\geq \mu'$. 
We show $\mu'>0$. Let $p$ denote the probability
 of attack to $v$. Then $\mu'\geq 1-p$ because the size of connected component of $v$
 after the deviation is at least $1$ and the probability of attack to $v$ after the deviation is at most 
 $p$ since dropping an edge would not increase the probability of attack to $v$. 
 Finally observe that $p<1$ otherwise $v$ would not have purchased any edge. Hence, $\mu'\geq 1-p>0$.
Now consider an immunized vertex $u$ in $\hat{G}$ and a deviation that $u$ purchases
an edge to $v$. $u$'s expected utility after this deviation is increased by at least $\mu-\c\geq \mu' > 0$
since the distribution of attack is unchanged after this deviation; a contradiction.

In case 2(b), observe that the adversary's attack distribution is uniform over 
all the targeted vertices. Furthermore, all the vulnerable vertices are targeted
because all the targeted regions have size 1.
So let $0<k < |V|$ denote the number of targeted vertices in $G$. Then the 
expected connectivity benefit (and utility) of $w$ which is a singleton vertex outside of $\hat{G}$ is $1-1/k$. 
We will show that some vertex in $G$ has a beneficial deviation by
considering the following sub-cases:
2(b$\alpha$) there is no targeted vertex in $\hat{G}$ and 2(b$\beta$) there is at least one targeted vertex in $\hat{G}$.

In case 2(b$\alpha$) all the vertices in $\hat{G}$ are immunized. Furthermore, there is an edge 
in $\hat{G}$ by Lemma~\ref{lem:edge-immunized}. 
As a result, $\hat{G}$
is a tree of immunized vertices because any edge beyond the tree would be redundant.
Pick a leaf vertex $v\in \hat{G}$. $v$ has purchased an edge in $\hat{G}$ since $\c>1$.
Then the connectivity benefit of $v$ is $|\hat{G}|$ which is at least $\b+\c+1/2$ (otherwise $v$ would better off
dropping the edge and immunization in which case she would die with probability of at most $1/2$
because $w$ is also targeted and the adversary would attack them with equal probability). 
Now consider the deviation
that $w$ immunizes and buys an edge to $u$. The change in $w$'s utility is
$|\hat{G}|+1-\b-\c-(1-1/k)>0$ since $|\hat{G}|\geq\b+\c+1/2$ and $k>0$; a contradiction.

The analysis of case 2(b$\beta$) is more delicate. 
Remind that by the assumptions so far, all the sub-cases below share the following common assumptions: there exists a vulnerable vertex $w$ outside of $\hat{G}$
which is targeted. Furthermore, targeted regions are all singletons and there is at least one targeted vertex in $\hat{G}$.
We consider the following exhaustive sub-cases.
\begin{enumerate}
\item[(i)] All the edges in $\hat{G}$ are purchased by immunized vertices.\\
Since $\c>1$ no immunized vertex in $\hat{G}$ would buy an edge to a targeted vertex in $\hat{G}$ unless that targeted 
vertex itself is connected to some other vertices (immunized in this case since targeted regions are singletons in case 2(b)). 
At first glance, the immunized vertex would be better off swapping the edge that connects her to this
targeted vertex to any of the immunized vertices that the targeted vertex itself is connected to. If so,
even when the targeted vertex gets attacked, the immunized vertex would remain connected to other neighboring 
vertices of the targeted vertex. However, the immunized vertex might be indifferent between her current action and swapping
which means there is another path that she has to any of the immunized vertices she would have lost connectivity
to when this particular targeted vertex is attacked. Note that every such path will also get disconnected in some other attack 
(otherwise there is no need to purchase the edge to the targeted vertex at the first place). This implies that every targeted vertex in $\hat{G}$ is a part of a cycle in $\hat{G}$.
Let $p\in\{1/k, \ldots,(k-1)/k\}$ denote the total probability of attack to any of the targeted vertices inside of $\hat{G}$.~\footnote{Note that this 
probability is in the increments of $1/k$ because the attack distribution is uniform over $k$ targeted vertex. So the total probability
of attack to targeted vertices in $\hat{G}$ is simply the number of targeted vertices in $\hat{G}$ times $1/k$.}
Then the expected
connectivity benefit of an immunized vertex who purchased an edge in $\hat{G}$ is $(1-p)|\hat{G}|+p(|\hat{G}|-1)=|\hat{G}|-p$.
\footnote{With probability $1-p$ the attack happens outside of $\hat{G}$ in which case the connectivity benefit is $|\hat{G}|$.
With probability $p$ the attack happens inside of $\hat{G}$ in which case the connectivity benefit is $|\hat{G}|-1$ because
the attack kills exactly one targeted vertex and that vertex is a part of a cycle in $\hat{G}$.}
Consider the deviation that an immunized vertex that purchased an edge drops her purchased edge. She decreases her expenditure by $\c$
and her connectivity benefit is at least $(1-p)|\hat{G}|+p$ after the deviation. 
\footnote{With probability $1-p$ the attack happens outside of $\hat{G}$ in which case the connectivity benefit is $|\hat{G}|$.
With probability $p$ the attack happens inside of $\hat{G}$ in which case the connectivity benefit is at least $1$ since the vertex is immunized.}
Since the immunized vertex, (weakly) prefers her current strategy then, $\c\leq  p|\hat{G}|-2p$. Finally, consider the deviation
that $w$ (which is outside of $\hat{G}$) purchases an edge to any immunized vertex in $\hat{G}$. 
Note that the attack distribution remains unchanged after this deviation.
Hence the change in $w$'s expected utility
is at least $p(|\hat{G}|)+(1-p-1/k)(|\hat{G}|+1)-\c -(1-1/k) \geq p > 0$; a contradiction.
Note that after the deviation, with probability $1-p-1/k$ the attack happens outside of $\hat{G}$ (which now contains $w$) in which case the connectivity benefit is $|\hat{G}|+1$.
With probability $p$ the attack happens inside of $\hat{G}$ and does not kill $w$ in which case the connectivity benefit is $|\hat{G}|$.

\item[(ii)] There exists a targeted vertex in $\hat{G}$ which purchased an edge.
Define the marginal benefit for an edge purchase by a vertex to be the difference between the 
expected utility of the vertex with and without the purchased edge.
\begin{enumerate}
\item[(I)] There exists a targeted vertex $u\in\hat{G}$ 
which has marginal benefit
of strictly bigger than $\c$ for one of her edge purchases.\\
Suppose $(u, v)$ is the edge purchased by $u$ which has a marginal
benefit strictly higher than $\c$. We know $v$ is immunized because targeted regions are singletons.
So the deviation that $w$ also purchases an edge to $v$ would have a marginal benefit of 
strictly bigger than $\c$ as well because the attack distribution remains unchanged after the deviation; a contradiction.
\item[(II)] For all the targeted vertices $u\in\hat{G}$ the marginal benefit is exactly $\c$ for all of the edge purchases made by $u$.
\begin{enumerate}
\item There is exactly one targeted vertex in $\hat{G}$.\\
Let $u$ be the sole targeted vertex in $\hat{G}$ and assume $u$ purchased $i>0$ edges. 
All the edge purchases of $u$ are to immunized vertices.
Furthermore, each of these immunized vertices are themselves connected 
to other immunized vertices otherwise $u$ would not have bought an edge to any of such vertices 
(remind that $\c>1$). So we can think of an edge purchased by $u$ as
an edge that connects $u$ to a fully immunized component. Finally, note that any such immunized component
is a tree because no vertex in that component can get attacked so any edge beyond a tree is redundant
in that component.

Remind that $k$ denote the total number of targeted vertices in $G$.
If $k>2$, $w$ can buy an edge to any vertex in one of these immunized components 
and get a marginal benefit of strictly bigger than $\c$. \footnote{The benefit of strictly bigger than $\c$ happens
when a vertex other than $u$ and $w$ are attacked. Such vertex exists when $k>2$.} So suppose $u$ and $w$
are the only targeted vertices in the network and therefore $k=2$.
Consider one immunized component that $u$ has purchased an edge to and let $X$ denote the size of this immunized component. 
Since the marginal utility of $u$ from this purchase is exactly $\c$
then $(1-1/k)X = \c$ or $X = (k)\c/(k-1)$. Replacing $k=2$ we get $X=2\c$. 
And this equation should hold for each of 
the immunized components that $u$ has purchased an edge to because $u$'s marginal benefit
for each edge purchase is $\c$ based on the assumption of this case. 
Furthermore, $u$ gets attacked with probability $1/2$ so for her to not immunize $\b\geq (iX+1)/2$.
As we mentioned before, all of the immunized components that $u$ has purchased an edge to are trees. Finally, consider a leaf
in one such tree and the deviation that the leaf drops her edge and becomes targeted. 
Her change in the utility is $2/3-\left((iX+1)/2+X/2-\b-\c\right)\geq 2/3$ because the adversary
now attacks each of the three targeted vertices with probability of $1/3$; a contradiction.
\item There are strictly more than one targeted vertex in $\hat{G}$.\\
Suppose $(u, v)$ is the edge purchased by a targeted vertex $u\in\hat{G}$. $v$ is immunized since
targeted regions are singletons.
Consider the deviation that $w$ purchases an edge to $v$. The marginal benefit of
this purchase is strictly bigger than $\c$ because (1) $w$ would get strictly higher benefit when any 
other targeted vertex besides $w$ and $u$ is attacked and such vertex exists by the assumption of this case
(since we assumed there are more than one targeted vertex in $\hat{G}$) and (2) 
the distribution of attack will not change after the deviation.
\end{enumerate}
\end{enumerate}
\end{enumerate}

\end{proof}

\begin{proof}[Proof of Theorem~\ref{thm:welfare-new}]
  First, Theorem~\ref{thm:connect} implies that $G$ is connected.
  Furthermore, the application of Lemma~\ref{lem:singletons} implies 
  that all the targeted regions of $G$
  (if there are any) are singletons.
Finally the number of immunized
  vertices is (trivially) at most $n$ and by Theorem~\ref{thm:sparse-general},
  there are at most $2n-4$ edges in $G$. So the collective expenditure of vertices in
  $G$ is at most $\cm:=(2n-4) \c + n \b$.

  Let $T=(B\cup C, E')$ be the block-cut tree decomposition of $G$.
  An attack to targeted non-cut vertices in
  any block of $T$ leaves $G$ with a single connected component after
  attack. However, an attack to targeted cut vertices of $T$ can
  disconnect $G$. So to analyze the welfare, we only consider the
  targeted cut vertices in $T$ and in particular we only focus on
  targeted cut vertices of $T$ with the property that the attack on
  such a vertex sufficiently reduces the size of the largest connected
  component in the resulting graph. More precisely, let
  $\epsilon=2 \sqrt{\c}/n^{1/3}$.  We refer to a targeted cut vertex
  $v$ as a \emph{heavy cut vertex} if after an attack to $v$, the size
  of the largest connected component in $G\setminus\{v\}$ is strictly
  less than $(1-\epsilon)n$.  We then show that the total probability
  of attack to heavy cut vertices is small if $G$ is a non-trivial
  equilibrium. This implies that with high probability (which we
  specify shortly) the network retains a large connected component
  after the attack, hence, the welfare is high.

  We root $T$ arbitrarily on some targeted cut vertex $r\in C$.  If there
  is no such cut vertex, then the size of the largest connected component
  in $G$ after any attack is at least $n-1$. So the social welfare in
  this case is at least $(n-1)^2-\cm$ and we are done.  So assume $r$
  exists. For any vertex $v$, let $T_v$ denote the subtree of $T$ rooted at $v$. 
  Consider the set of cut vertices $\D_r\subseteq C$ such
  that for all $v\in \D_r$:
(a) $v$ is targeted,
(b) $|T_v|\ge \epsilon  n$, and
(c) no targeted cut vertex $v' \in T_v \setminus \{ v \}$ has the property that
$|T_{v'}|\ge \epsilon  n$ i.e. $v$ is the deepest vertex in the
tree $T_v$ that satisfies property (b).

Note that each $v\in \D_r$ is a heavy cut vertex (but there might be other 
heavy cut vertices in $T$ that are not in $\D_r$).  We consider two
cases: (1) $|\D_r|=1$ and (2)
$|\D_r| > 1$.

Consider case (1) where $|\D_r|=1$. Let $\D_r=\{v\}$. Consider the
following two cases: 1(a) $v=r$ and 1(b) $v\ne r$ where $r$ is the root
of the tree.

\begin{figure}[h]
\centering
\begin{minipage}[c]{.22\textwidth}
\centering
\scalebox{.6 }{
\begin{tikzpicture}
[scale=0.7, every node/.style={circle,draw=black, minimum size=0.7cm}, red node/.style = {circle, fill = red, draw},  gray node/.style = {circle, fill = blue, draw}]
\node [red node] (2) at  (0, 10){$v$};
\node [draw,rectangle,color=white,minimum width=1cm,minimum height=0.6cm,label=$$] (4) at (0, 8.4) {$$};
\node [draw,rectangle,color=white,minimum width=1.2cm,minimum height=1cm,label=$$] (7) at (-2, 8.5) {$$};
\node [draw,rectangle,color=white,minimum width=1cm,minimum height=1.5cm,label=$$] (8) at (2, 8.5) {$$};
\draw(2) to (4);\draw(2) to (7);\draw(2) to (8);
\draw (0,8.9)--(0.5,8)--(-0.5,8)--cycle;
\draw (1.3,9.05)--(0.8,7.1)--(1.8,7.1)--cycle;
\draw (-1.1,9.15)--(-0.6,7.2)--(-1.6,7.2)--cycle;
\end{tikzpicture}}
\caption{\label{fig:newcase2a}}
\end{minipage}
\begin{minipage}[c]{.27\textwidth}
\centering
\scalebox{.6}{
\begin{tikzpicture}
[scale=0.7, every node/.style={circle,draw=black, minimum size=0.7cm}, red node/.style = {circle, fill = red, draw},  gray node/.style = {circle, fill = blue, draw}]
\node [red node] (1) at  (0, 4){$v$};
\node [red node] (2) at  (0, 10){$r$};
\node (3) at  (0, 7){};
\node (9) at (2,7){};
\node [draw,rectangle,color=black,minimum width=1cm,minimum height=0.6cm,label=$$] (4) at (0, 8.5) {$$};
\node [draw,rectangle,color=black,minimum width=1cm,minimum height=0.6cm,label=$$] (7) at (-2, 8.5) {$b$};
\node [draw,rectangle,color=black,minimum width=1cm,minimum height=0.6cm,label=$$] (8) at (2, 8.5) {$$};
\node [draw,rectangle,color=black,minimum width=1cm,minimum height=0.6cm,label=$$] (5) at (0, 5.5) {$$};
\node [draw,rectangle,color=black,minimum width=1cm,minimum height=0.6cm,label=$$] (10) at (2, 5.5) {$$};
\draw (0,4.8)--(-1.2,3.5)--(1.2,3.5)--cycle;
\draw(2) to (4);\draw(4) to (3);\draw(2) to (7);\draw(2) to (8);\draw(3) to (5);\draw(1) to (5);
\draw(9) to (10); \draw(9) to (8);
\end{tikzpicture}}
\caption{\label{fig:newcase2b2}}
\end{minipage}
\begin{minipage}[c]{0.35\textwidth}
\centering
\scriptsize{
Figure~\ref{fig:newcase2a}: Case 1(a); $v$ is the only heavy cut vertex
 and is the root of $T$. The triangles denote the subtrees
 rooted at the child blocks of $v$.\\
Figure~\ref{fig:newcase2b2}: Case 1(b2); $v\ne r$ and either $r$ or a vertex in $b$ has a 
beneficial deviation. The triangle denotes
the subtree rooted at $v$.}
\end{minipage}
\end{figure}

In case 1(a), let $p$ be the probability of attack to $v$.
Consider the deviation that $v$ immunizes but maintains the same edge purchases as in her current strategy. 
Since $G$ is an equilibrium, $v$ (weakly) prefers her current strategy to the deviation.  The
connectivity of $v$ after an attack to any targeted vertex other than $v$ is at least $n-\epsilon
n$.  So for $v$ to not prefer immunizing: $p \left(n-\epsilon n\right)\leq \b$.  Moreover, if any vertex other than
$v$ is attacked, the size of the largest connected component after the
attack is at least $(1-\epsilon)n$ (see Figure~\ref{fig:newcase2a}). This implies the welfare is at least
\begin{align*}
\left(1-p\right)\left(\left(1-\epsilon\right)n\right)^2-\cm &> (1-\frac{\b}{(1-\epsilon)n})\left(1-2\epsilon\right)n^2-\cm\\
&> n^2-4\sqrt{\c} n^{5/3} - \frac{\b n^{4/3}}{n^{1/3}-2\sqrt{\c}}-\cm = n^2-O(n^{5/3}).
\end{align*}

For case $1(b)$, observe that the targeted cut vertices on the path
from $v$ to $r$ (the root of $T$) are the only possible heavy cut
vertices in the network (counting both $v$ and $r$ to be on
  the path).  So let $p_v$ denote the probability that some heavy cut vertex
on the path from $v$ to $r$ is attacked.  We consider
two cases: 1(b1) $ p_v \leq \sqrt{\c} n^{-1/3} $, and 1(b2)
$p_v > \sqrt{\c} n^{-1/3}$.  We show that in case 1(b1) the welfare is
as claimed in the statement of Theorem~\ref{thm:welfare-new} and case
1(b2) cannot happen.

In case 1(b1), with probability $1-p_v$, the size of the largest
connected component after the attack is at least
$(1-\epsilon)n$. Hence the welfare in case 1(b1) is at least
\begin{align*}
&\left(1-p_v\right)\left(\left(1-\epsilon\right)n\right)^2-\cm 
\geq \left(1-\sqrt{\c} n^{-1/3}\right)\left(1-2\epsilon\right)n^2-\cm\\
&=\left(1-\sqrt{\c} n^{-1/3}\right)\left(1-\frac{4 \sqrt{\c}}{n^{1/3}}\right)n^2-\cm
> n^2-5\sqrt{\c} n^{5/3}-\cm = n^2-O(n^{5/3}).
\end{align*}

In case 1(b2), since $r$ is a cut vertex, $r$ has at least two child
blocks. Consider any child block of $r$ that is not in the same
subtree of $r$ as $v$ (e.g. $b$ in Figure~\ref{fig:newcase2b2}) and call
this child block $b$. Since all the targeted regions are singletons, $r$
is only connected to immunized vertices in $b$ -- let $w$ be one such
immunized vertex. Now consider the deviation that $w$ purchases an
edge to an immunized vertex $w'$ in $T_v$ ($w'$ exists because by the
choice of $\epsilon$, $|T_v|\geq 2$, $v$ is targeted and targeted regions are singletons).  
Note that this deviation does not change the distribution of the attack.
Thus, the deviation will give $w$ additional benefit of at least
$|T_v| -1 \geq \epsilon n - 1$ (for the entirety of $T_v$ other than
$v$) whenever the attack occurs on the path from $v$ to $r$ (which
happens with probability $p_v$). So, $w$'s marginal increase in her expected
utility for this purchase will be at least
\[
p_v(\epsilon n -1) > \left(\sqrt{\c} n^{-1/3}\right)\left(2 \sqrt{\c} n^{2/3}-1\right) = 2 \c  n^{1/3} - \sqrt{\c} n^{-1/3} > \c n^{1/3} >\c,
\]
which shows that the $w$ can \emph{strictly} increase her expected
utility in the deviation; a contradiction.

In case (2), let $r'$ be a cut vertex that is the \emph{lowest common
  ancestor} of vertices in $\D_{r}$. If $r'\ne r$, we root the tree on
$r'$ and repeat the process of finding heavy cut vertices. Note that
$\D_{r}\subseteq \D_{r'}$ since we might add some additional heavy cut
vertices to $\D_{r'}$ (vertices that used to be ancestors of $r'$ in $T$). See
Figures~\ref{fig:newbefore}~and~\ref{fig:newafter} for an example.

\begin{figure}[h]
\centering
\begin{minipage}[b]{0.22\textwidth}
\centering
\scalebox{.6}{
\begin{tikzpicture}
[scale=0.7, every node/.style={circle,draw=black, minimum size=0.7cm}, red node/.style = {circle, fill = red, draw},  gray node/.style = {circle, fill = blue, draw}]
\node [red node] (1) at  (0, 4){$v_2$};
\node [red node] (12) at  (-2.5, 4){$v_1$};
\node [red node] (2) at  (0, 10){$r$};
\node (3) at  (0, 7){$r'$};
\node [draw,rectangle,color=black,minimum width=1cm,minimum height=0.6cm,label=$$] (4) at (0, 8.5) {$$};
\node [draw,rectangle,color=black,minimum width=1cm,minimum height=0.6cm,label=$$] (7) at (-2, 8.5) {$$};
\node [draw,rectangle,color=black,minimum width=1cm,minimum height=0.6cm,label=$$] (8) at (2, 8.5) {$$};
\node [draw,rectangle,color=black,minimum width=1cm,minimum height=0.6cm,label=$$] (11) at (-2, 5.5) {$$};
\node [draw,rectangle,color=black,minimum width=1cm,minimum height=0.6cm,label=$$] (5) at (0, 5.5) {$$};
\draw (0,4.8)--(-1.2,3.5)--(1.2,3.5)--cycle;
\draw (-2.2, 4.9)--(-4.1,3.5)--(-1.8,3.5)--cycle;
\draw(2) to (4);\draw(4) to (3);\draw(2) to (7);\draw(2) to (8);\draw(3) to (5);\draw(1) to (5);
\draw(11) to (12); \draw(11) to (3);
\end{tikzpicture}}
\caption{\label{fig:newbefore}}
\end{minipage}
\begin{minipage}[b]{0.27\textwidth}
\centering
\scalebox{.6}{
\begin{tikzpicture}
[scale=0.7, every node/.style={circle,draw=black, minimum size=0.7cm}, red node/.style = {circle, fill = red, draw},  gray node/.style = {circle, fill = blue, draw}]
\node [red node] (1) at  (0, 4){$v_2$};
\node [red node] (12) at  (-2.5, 4){$v_1$};
\node [red node] (15) at  (2.5, 4){$r$};
\node (3) at  (0, 7){$r'$};
\node [draw,rectangle,color=black,minimum width=1cm,minimum height=0.6cm,label=$$] (11) at (-2, 5.5) {$$};
\node [draw,rectangle,color=black,minimum width=1cm,minimum height=0.6cm,label=$$] (5) at (0, 5.5) {$$};
\node [draw,rectangle,color=black,minimum width=1cm,minimum height=0.6cm,label=$$] (10) at (2, 5.5) {$$};
\draw(3) to (5);\draw(1) to (5);\draw(11) to (12); \draw(11) to (3);\draw(10) to (3);\draw(15) to (10);
\draw (0,4.8)--(-1.2,3.5)--(1.2,3.5)--cycle;
\draw (-2.2, 4.9)--(-4.1,3.5)--(-1.8,3.5)--cycle;
\draw (2.2, 4.9)--(1.8,3.5)--(4,3.5)--cycle;
\end{tikzpicture}}
\caption{\label{fig:newafter}}
\end{minipage}
\begin{minipage}[b]{0.3\textwidth}
  \scriptsize{An example of re-rooting in case 2. Heavy cut vertices in $\D$ are
  in red. The small rectangles and circles denote blocks and cut
  vertices, respectively. The triangles denote the subtrees rooted at
  critical cut vertices.~\ref{fig:newbefore} is before
  and~\ref{fig:newafter} is after re-rooting.}
\end{minipage}
\end{figure}

Observe that the vertices in $\D_{r'}$ and the targeted cut vertices
on the path from some $v\in\D_{r'}$ to $r'$ (the new root) are the
only possible heavy cut vertices in the tree.  Let $p_v$ denote the
probability that some targeted cut vertex on the path from $v$ to $r'$
is attacked when $v\in\D_{r'}$.  We consider two cases: 2(a)
$\Sigma_{v\in\D_{r'}} p_v \leq n^{-1/3}$, and 2(b)
$\Sigma_{v\in\D_{r'}} p_v > n^{-1/3}$. We show that in case 2(a) the
welfare is as claimed in the statement of
Theorem~\ref{thm:welfare-new} and case 2(b) cannot happen.

In case 2(a), with probability of at least
$1-\Sigma_{v\in\D_{r'}} p_v$, the attack does not occur on a path from
any $v\in \D_{r'}$ to $r'$. Thus, in these cases, the size of the
largest connected component after an attack is at least
$(1-\epsilon)n$. Hence the welfare in this case is at least
\begin{align*}
&(1-\sum_{v\in\D_{r'}} p_v)\left(\left(1-\epsilon\right)n\right)^2-\cm
\geq
\left(1-n^{-1/3}\right)\left(1-2\epsilon\right)n^2-\cm\\
&=
\left(1-n^{-1/3}\right)\left(1-\frac{4 \sqrt{\c}}{n^{1/3}}\right)n^2-\cm
> n^2-\left(1+4\sqrt{\c}\right)n^{5/3}-\cm= n^2-O\left(n^{5/3}\right).
\end{align*}

Next, we consider case 2(b), namely
$\Sigma_{v\in\D_{r'}} p_v > n^{-1/3}$.  Since $|T_v| \geq n\epsilon$,
we know that $|\D_{r'}|\le 1/\epsilon$. Therefore, there exists a
$v^*\in\D_{r'}$ such that
\begin{equation}
\label{eq:delta}
p_{v^*} > \frac{n^{-1/3}}{|\D_{r'}|} \geq n^{-1/3}\epsilon = 2\sqrt{\c} n^{-2/3},
\end{equation}
by the pigeonhole principle. Also, since each $v'\in \D_{r'}$ is a cut
vertex, and is unimmunized, $v'$ must have a child block. Since
unimmunized vertices are singletons, $v'$ must be connected to her
child blocks through an immunized vertex.

By the choice of the root in $\D_{r'}$ (i.e. the least common ancestor
in $\D_{r}$ before re-rooting), there exists a $v'\in \D_{r'}$ such that every time a
heavy cut vertex on the path from $v^*$ to the root is attacked then
$T_{v^*}$ and $T_{v'}$ end up in different connected components.  Now,
consider the deviation that an immunized vertex $w$ in $T_{v^*}$
purchases an edge to an immunized vertex in $T_{v'}$.  
Note that this deviation does not change the distribution of the attack.
So after this
deviation, $w$ would get an additional connectivity benefit of at
least $p_{v^*}(|T_{v'}|-1)$ -- this benefit occurs whenever there is
an attack to a cut vertex on the path from $v^*$ to the root (which
happens with probability of $p_{v^*}$) and the connectivity benefit in
this case is at least $|T_{v'}|-1\geq(\epsilon n-1)$. Moreover, the extra
expenditure of $w$ in this deviation is $\c$. However,
\[
p_{v^*}\left(\epsilon n-1\right) \geq \left(2\sqrt{\c} n^{-2/3}\right)\left(2 \sqrt{\c} n^{-1/3}n-1\right) = 4 \c - 2\sqrt{\c} n^{-2/3} \geq 2\c > \c,
\]
which shows that $w$ can increase her expected utility strictly in the deviation; a contradiction.
\end{proof}
\section{Connectivity and Social Welfare in Equilibria -- Maximum Disruption Adversary}
\label{sec:max-disruption-welfare}
In this section we prove the analog results of Section~\ref{sec:welfare} regarding the social welfare with respect to the \maxdisrupt~adversary.
While the the results in this section look almost identical to the statements in Section~\ref{sec:welfare}, 
we explicitly point out some of the differences. First, in Section~\ref{sec:welfare} we could show that when $\c>1$,
any non-trivial\footnote{Remind that an equilibrium network is non-trivial if it contains at least one immunized vertex and one edge.}  
Nash or swapstable equilibrium network with respect to \maxcarnage~adversary is connected, has targeted regions of size at most 1 and enjoys high social welfare.
While we suspect that all these statements hold for Nash equilibrium networks with respect to~\maxdisrupt~adversary,
we can only show that when $\c>1$ every non-trivial and connected 
Nash equilibrium network with respect to the \maxdisrupt~adversary has targeted regions of size at most 1 and 
enjoys high social welfare. Hence, we leave the question of whether non-trivial Nash equilibrium networks with respect to
\maxdisrupt~adversary are connected when $\c>1$ as an open question. 
Second, while in the welfare results of Section~\ref{sec:welfare} also hold for non-trivial swapstable equilibrium networks with respect to
\maxcarnage~adversary, we show that when $\c>1$,
non-trivial swapstable equilibrium networks with respect to~\maxdisrupt~adversary can be disconnected, can have targeted regions of size bigger than one
and in general can have pretty low social welfare.
\begin{thm}\label{lem:singles-disruption}
  Let $\c>1$, and consider a Nash equilibrium network
  $G$ with respect to the \maxdisrupt~adversary.  If $G$ is non-trivial
  and connected, then the size of
  all targeted regions, if there are any, is
  exactly $1$.
\end{thm}

\begin{thm}
\label{thm:welfare-new-disruption}
Let $\c>1$, and consider a Nash equilibrium network
$G=(V,E)$ with respect to the \maxdisrupt~adversary over $n$ vertices.  If $G$ is non-trivial and connected and $\c$ and $\b$
are constants (independent of $n$), then the welfare of $G$ is
$n^2 - O(n^{5/3})$.
\end{thm}

\begin{lemma}
\label{ex:swap-disconnected}
When $\c>1$, there exists a non-trivial swapstable (and hence linkstable) equilibrium network $G$ with respect to the \maxdisrupt~adversary such
that $G$ has more than one connected component and some targeted regions have size strictly bigger than 1.
\end{lemma}

The example in the proof Lemma~\ref{ex:swap-disconnected} immediately implies the following corollary.
\begin{cor}
When $\c>1$, there exists a non-trivial swapstable (and hence linkstable) equilibrium network $G=(V,E)$ with respect to the \maxdisrupt~adversary such
that the welfare of $G$ is $O(n)$ where $n=|V|$.
\end{cor}

\subsection{Proof of Theorem~\ref{lem:singles-disruption}}
We first prove the following useful result which is the analog of Lemma~\ref{lem:tree2}
for the \maxdisrupt~adversary.
\begin{lemma}
\label{lem:tree3}
Let $G=(V,E)$ be a Nash, swapstable or linkstable equilibrium network
with respect to the \maxdisrupt~adversary. 
The number of targeted regions cannot be one when $|V|>1$.
\end{lemma}
\begin{proof}
So consider the case that there exists a unique
singleton targeted vertex $u$. When $G$ is an empty graph then the
same argument shows that $u$ cannot exist. So suppose $G$ is non-empty
and $u$ is the unique singleton targeted vertex.  So there is some
immunized vertex $v\in V\setminus \{u\}$ who purchases an edge to some
other $v'\in V\setminus \{u, v\}$. Let $B$ be the partition that
$v, v'$ belong to. Since $v$ is best responding, it must be that
$|B| - \b - \c \geq 0$, since $v$ could choose to not buy $(v, v')$
and not to immunize for expected utility of at least $0$. This implies
that $u$ cannot be best responding in this case, since buying an edge
to $v$ and immunizing would give $u$ an expected utility of
$(|B| +1) - \b - \c \geq 1 > 0$, a contradiction to $G$ being an equilibrium.\\
\end{proof}
We are now ready to prove Theorem~\ref{lem:singles-disruption}.
\noindent\emph{Proof of Theorem~\ref{lem:singles-disruption}.}
Suppose not. Then there exists some targeted region $\T$ with
$|\T| >1$. 
Note that $\T$ is a vulnerable region such that an attack to $\T$
will minimize the social welfare in this case.
By Lemma~\ref{lem:tree}, the subgraph of $G$ on $\T$
forms a tree. Then, this tree must have at least two leaves
$x,y\in \T$.  We claim that there is some vertex in $\T$ who would
strictly prefer to \emph{swap} her edge to some immunized vertex in
$G$ rather than an edge which connects her to the remainder of $\T$.
  
Since $G$ contains some immunized vertex (since $G$ is non-trivial),
any connection between $\T$ and the rest of $G$ is through immunized
vertices.  We consider two cases and show that none of them is
possible.
\begin{enumerate}
\item {\it One of $x$ or $y$ buys her edge in the tree.} Suppose
  without loss of generality $x$ buys an edge in the tree. Since $G$
  is connected, there exists an immunized vertex $z$ which is
  connected to some vertex in $\T$.  If $x$ is not connected to $z$,
  then $x$ would strictly prefer to buy an edge to $z$ over buying her
  tree edge. This is because the targeted regions are not unique in
  equilibria by Lemma~\ref{lem:tree3} and after the deviation all the
  previous targeted regions remain targeted, no new targeted region
  would be added and the targeted region that $x$ was a part of would
  become non-targeted. So $x$'s utility would only strictly increase
  by this deviation; a contradiction.  
  So suppose $x$ is connected to
  $z$. Then if $y$ also bought her tree edge, she would also strictly
  prefer an edge to $z$ for the same reason. 
  Observe that
  $y$ cannot be connected to $z$ because one of the edges $(x,z)$ or $(y,z)$ would 
  be redundant.
  So suppose $y$ did not buy her tree edge. Observe that
  $y$ cannot be connected to $z$ because one of the edges $(x,z)$ or $(y,z)$ would 
  be redundant.
  Now consider the edge that connects $y$ to the tree $\T$. Then $y$'s 
  parent in the tree must have bought this edge; since $\c>1$,
  this implies $y$ must be connected to some immunized vertex $z'$ (or
  it would not be worth connecting to $y$); 
  Also observe that $y$'s parent can be connected to $z$ because either the edge
  between $x$ and $z$ or $y$'s parent and $z$ is redundant.
  However, $y$'s parent would
  strictly prefer to buy an edge to $z'$ over an edge to $y$.  
  Thus, $x$ cannot have bought her tree edge; either $y$ or her parent would
  like to re-wire if this were the case.
\item {\it Neither $x$ nor $y$ buys her connecting edge in the tree.}
  Since $\c>1$, both $x$ and $y$ must have immunized neighbors (or
  their edges being purchased by $x$'s targeted parent and $y$'s
  targeted parent would not be best responses by those vertices). 
    Let $z$ and $z'$ denote the immunized vertices connected to $x$ and $y$, respectively.
  Note that $z\ne z'$ otherwise one of the edges $(z,x)$ or $(z', y)$ would be redundant.
  But then, both $x$'s parent and $y$'s parent in the tree $\T$ would
  strictly prefer to buy an edge to $z$ and $z'$
  rather than to $x$ and $y$, respectively.
\end{enumerate}
\qed

\subsection{Proof of Theorem~\ref{thm:welfare-new-disruption}}
Before proving
Theorem~\ref{thm:welfare-new-disruption} we state Lemma~\ref{lem:same-max}
that would be useful in the proof.

\begin{lemma}
\label{lem:same-max}
Let $G=(V,E)$ be a connected graph with at least two immunized
vertices. Suppose the welfare is at least $W$ with respect to the \maxdisrupt adversary 
when an attack starts at
any vertex $v\in V$. Consider the graph $G'=(V', E')$ where $V'=V$ and
$E'=E\cup (v_1, v_2)$ and $v_1$ and $v_2$ are any two immunized
vertices i.e., $G'$ is the same as $G$ with only an edge added between
$v_1$ and $v_2$.  Then the welfare is at least $W-\c$ 
with respect to the \maxdisrupt adversary when an attack
starts at any vertex $v'\in V'$ in $G'$.
\end{lemma}
\begin{proof}
The statement is trivial when $E=E\cup (v_1, v_2)$. So suppose $(v_1, v_2)\notin E$. 
\begin{itemize}
\item First, consider any vertex $v\in V$ that realizes the welfare
  $W$ post-attack to $v$ in $G$. Let $G_1, \ldots G_k$ be the
  connected components in $G\setminus \{v\}$. If $v_1$ and $v_2$ are
  in the same connected component, then $G_1, \ldots G_k$ would be the
  also the connected components in $G'\setminus \{v\}$. In such case
  the sum of connectivity benefits remains the same while the
  collective expenditure increases by $\c$; so the welfare is
  $W-\c$. Otherwise, $v_1$ and $v_2$ would be in different connected
  components in $G\setminus \{v\}$.  Without loss of generality, let $G_1$ and $G_2$ to be
  such components, respectively.  In this case
  $G_1\cup G_2, G_3, \ldots G_k$ would be connected components in
  $G'\setminus \{v\}$. This means that after the attack to $v$ in $G'$
  the sum of connectivity benefits only increases (since
  $|G_1+G_2|^2 > |G_1|^2 +|G_2|^2$).  Since the collective expenditure
  also increases by $\c$, then the welfare is strictly bigger than
  $W-\c$ in this case.
\item Second, consider any vertex $v\in V$ such that the welfare
  post-attack to $v$ in $G$ is $W_v > W$. Now similar to the above
  case after adding the edge $(v_1, v_2)$, the welfare in $G'$ after
  an attack to $v$ is either $W_v-\c$ or strictly bigger than
  $W_v-\c$. And both of these values are strictly bigger than $W-\c$
  as claimed.
\end{itemize}
\end{proof}

An immediate consequence of Lemma~\ref{lem:same-max} is the following
scenario. Let $G=(V,E)$ be a non-trivial 
and connected
equilibrium 
network with respect to the \maxdisrupt~adversary when $\c>1$.  
By Lemma~\ref{lem:tree3} we know that the number
of targeted regions (if there are any) in $G$ is strictly bigger than1. 
Furthermore, 
By Lemma~\ref{lem:singles-disruption} we know that 
when $G$ is connected 
the targeted
regions (if they exist) are singletons.  
So let $u_1$ and $u_2$ be two such targeted
vertices in $G$. Suppose there exist immunized vertices $v_1$ and
$v_2$ such that $v_1$ and $v_2$ remain in the same connected component
after an attack to $u_1$ but end up in different connected components
after an attack to $u_2$. Now consider the graph $G'=(V', E')$ where
$V'=V$ and $E'=E\cup (v_1, v_2)$ i.e., $G$ with added edge
$(v_1, v_2)$. Let $\T'\subseteq V'$ be the set of targeted vertices in
$G'$. Then $u_1\in \T'$ and $u_2\notin \T'$. Furthermore, every vertex
$v\in \T'$ was also a targeted vertex in $G$. So targeted regions in
$G'$ are also singletons.

Intuitively, the above describes a situation where after the deviation, 
the targeted regions could be identified
easily i.e., we only need to consider targeted regions before the
deviation and check which one of them remains targeted after the
deviation. We use this observation extensively in the proof of Theorem~\ref{thm:welfare-new-disruption}.



$\newline$
\noindent\emph{Proof of Theorem~\ref{thm:welfare-new-disruption}.}
First of all, by Lemma~\ref{lem:singles-disruption}  
the size
of targeted regions (if there are any) in $G$ is exactly $1$.
Also since
  there are at most $2n-4$ edges in $G$ by
  Theorem~\ref{thm:sparse-general} and the number of immunized
  vertices is at most $n$, the collective expenditure of vertices in
  $G$ is at most $\cm=(2n-4) \c + n \b$.

  Let $T=(B\cup C, E')$ be the block-cut tree decomposition of
  $G$.~\footnote{Recall that for any
    $v\in B\cup C$, we denote $T_v$ to be the subtree rooted at
    $v$. We define the size of $T_v$ (denoted by $|T_v|$)
    to be the cardinality of the union of all the blocks and cut
    vertices in $T_v$. In contrast to the standard convention that cut vertices are also part
of the blocks their removal would disconnect, we will assume throughout that cut vertices are not part of the blocks
to avoid overcounting.}  
    The decomposition has the nice property that
  an attack to a  targeted vertex in any block of $T$ leaves $G$ with a
  single connected component after attack. However, an attack to a
  targeted cut vertex of $T$ can disconnect $G$. 
  Due to the choice of the adversary, either targeted vertices are all in blocks, or all 
  cut vertices. In the former case, the statement of the theorem is immediate.
  In the latter case, to analyze the
  welfare we only consider the targeted cut vertices in $T$ and in
  particular we only focus on targeted cut vertices of $T$ with the
  property that the attack on such a vertex sufficiently reduces the size of the largest 
connected component in the resulting graph. More precisely, let
  $\epsilon=2 \sqrt{\c}/n^{1/3}$.  We refer to a targeted cut
  vertex $v$ as a \emph{heavy targeted cut vertex} (heavy for short when it is clear from the context) 
  if after an attack to $v$, the size of
  the largest connected component in $G\setminus\{v\}$ is strictly
  less than $(1-\epsilon)n$.  We then show that the total probability of
  attack to heavy targeted cut vertices is small if $G$ is a
  non-trivial equilibrium. This implies that with high probability (which we specify shortly) 
  the network retain a large connected component after the attack and, hence, the
  welfare is high.

  Root $T$ arbitrarily on some targeted cut vertex $r\in C$. 
  If there is no such cut vertex, then the size of largest connected
  component in $G$ after any attack is at least $n-1$. So the social welfare 
  in that case is at least $(n-1)^2-\cm$ and we are done.
  So assume $r$ exists and consider the
  set of cut vertices $\D_r\subseteq C$ such that for all $v\in \D_r$

\begin{enumerate}
\item[(a)] $v$ is targeted,
\item[(b)] $|T_v|\ge \epsilon  n$, and
\item[(c)] no targeted cut vertex $v' \in T_v \setminus \{ v \}$ has the property that
$|T_{v'}|\ge \epsilon  n$ (i.e., $v$ is the deepest vertex in the
tree $T_v$ that satisfies property (b)).
\end{enumerate}
Observe that each $v\in \D_r$ is a heavy targeted cut vertex (but there might be other 
heavy targeted cut vertices in $T$ that are not in $\D_r$).  We consider two
cases based on the size of $\D_r$: (1) $|\D_r|=1$ and (2)
$|\D_r| > 1$. 
We first outline the structure of the proof as follows.
\begin{enumerate}[label=(\arabic*)]
\item $|\D_r|=1$:
\begin{enumerate}[label=\alph*]
\item $\D_r = \{r\}$ where $r$ is the root of the tree: In this case we show that the welfare is as claimed.
\item $\D_r = \{v\}$ and $v \neq r$ where $r$ is the root of the tree: Let $p_v$ denote the probability of attack to heavy targeted cut vertices
on the path from $v$ to $r$ (including $r$ and $v$). Then
\begin{enumerate}[label=\arabic*]
\item $p_v\leq \sqrt{\c\b}n^{-1/3}$: In this case the welfare is as claimed.
\item $p_v> \sqrt{\c\b}n^{-1/3}$:
\begin{enumerate}[label=\arabic*]
\item There exists a non-heavy targeted cut vertex: In this case we propose a beneficial deviation for a vertex contradicting the assumption that the network was an equilibrium.
\item All the targeted cut vertices are heavy: In this case we propose a beneficial deviation for a vertex contradicting the assumption that the network was an equilibrium.
\end{enumerate}
\end{enumerate}
\end{enumerate}
\item $|\D_r| > 1$: Let $p_F$ denote the total probability of attack to all the heavy targeted cut vertices.
\begin{enumerate}[label=\alph*]
\item $p_F\leq 8\sqrt{\c}n^{-1/3}$: In this case we show the welfare is as claimed.
\item $p_F > 8\sqrt{\c}n^{-1/3}$: In this case we propose a beneficial deviation for a vertex contradicting the assumption that the network was an equilibrium.
\end{enumerate}
\end{enumerate}

Consider case (1) where $|\D_r|=1$. Let $\D_r=\{v\}$. Consider the following two cases:
1(a) $v=r$ and 1(b) $v\ne r$ where $r$ is the root of the tree.

In case 1(a), let $p$
be the probability of attack to $v$.
Now consider the deviation in which $v$
immunizes but maintains the same edge purchases as in her current
strategy in $G$. Since $G$ is an equilibrium, $v$
(weakly) prefers her current strategy to this deviation.  The
connectivity of $v$ after an attack to any targeted vertex other than $v$
is at least $n-\epsilon n$.  So for $v$ to not prefer immunizing: $p \left(n-\epsilon
  n\right)\leq \b$.  Furthermore, if any vertex other than
$v$ is attacked, the size of the largest connected component after the
attack is at least $(1-\epsilon)n$. This implies the welfare is at least
\begin{align*}
\left(1-p\right)\left(\left(1-\epsilon\right)n\right)^2-\cm &> \left(1-\frac{\b}{(1-\epsilon)n}\right)\left(1-2\epsilon\right)n^2-\cm\\
&> n^2-4\sqrt{\c} n^{5/3} - \frac{n^{4/3}}{n^{1/3}-2\sqrt{\c}}-\cm = n^2-O(n^{5/3}).
\end{align*}

For case
$1(b)$, 
observe that the targeted cut vertices on the path from $v$ to
$r$ (the root of $T$) are the only possible heavy targeted cut vertices in the
network. \footnote{We count both $v$ and $r$ to be on the path.}  So
let $p_v$ denote the probability that some heavy targeted cut vertex on the path from
$v$ to $r$ (including $r$ and $v$) is attacked.  We consider two cases: 1(b1)
$ p_v \leq \sqrt{\b\c} n^{-1/3} $, and 1(b2)
$p_v > \sqrt{\b\c} n^{-1/3}$.  We show that in case 1(b1) the welfare
is as claimed in the statement of Theorem~\ref{thm:welfare-new} and
case 1(b2) cannot happen at equilibrium.

In case 1(b1), with probability $1-p_v$, the size of the largest
connected component after the attack is at least
$(1-\epsilon)n$. Hence the welfare in case 1(b1) is at least
\begin{align*}
\left(1-p_v\right)\left(\left(1-\epsilon\right)n\right)^2-\cm 
&\geq \left(1-\sqrt{\b\c} n^{-1/3}\right)\left(1-2\epsilon\right)n^2-\cm\\
&=\left(1-\sqrt{\b\c} n^{-1/3}\right)\left(1-\frac{4 \sqrt{\c}}{n^{1/3}}\right)n^2-\cm\\
&> n^2-\left(\sqrt{\b\c}+4\sqrt{\c}\right) n^{5/3}-\cm = n^2-O(n^{5/3}).
\end{align*}

In case 1(b2), consider two sub-cases: 1(b2-1) there exists a targeted
cut vertex $v'$ such that after an attack to $v'$ the size of the
largest connected component is at least $(1-\epsilon)n$ or 1(b2-2) not.

\begin{figure}[h]
\centering\begin{minipage}[c]{.4\textwidth}
\scalebox{.5}{
\centering
\begin{tikzpicture}
[scale = 0.7, every node/.style={circle,draw=black, minimum size=0.7cm}, red node/.style = {circle, fill = red, draw},  gray node/.style = {circle, fill = blue, draw}]
\node [red node] (1) at  (0, 3){$v$};
\node [red node] (2) at  (0, 10){$$};
\node [red node] (10) at  (0, 13){$r$};
\node [red node] (3) at  (0, 7){$v^*$};
\node (12) at  (2.5, 10){$v'$};
\node [draw,rectangle,color=black,minimum width=1cm,minimum height=0.6cm,label=$$] (11) at (0, 11.5) {$$};
\node [draw,rectangle,color=black,minimum width=1cm,minimum height=0.6cm,label=$$] (4) at (0, 8.5) {$b^*$};
\node [draw,rectangle,color=black,minimum width=1cm,minimum height=0.6cm,label=$$] (13) at (2.5, 8.5) {$$};
\node [draw,rectangle,color=black,minimum width=1cm,minimum height=0.6cm,label=$$] (7) at (-2, 11.5) {$$};
\node [draw,rectangle,color=black,minimum width=1cm,minimum height=0.6cm,label=$$] (5) at (0, 5.5) {$$};
\node [draw,rectangle,color=black,minimum width=1cm,minimum height=0.6cm,label=$$] (6) at (0, 1.5) {$b$};
\draw (0,4.5)--(-2,0.9)--(2,0.9)--cycle;
\draw (1) to (6);\draw(2) to (4);\draw(4) to (3);\draw(3) to (5);\draw(1) to (5);
\draw (10) to (11); \draw (10) to (7);\draw (12) to (11);\draw(11) to (2);\draw (12) to (13);
\end{tikzpicture}}
\end{minipage}
\begin{minipage}[c]{.5\textwidth}
\centering
\caption{Case 1(b2-1): All the heavy targeted cut vertices are red and $v'$ is a non-heavy targeted cut vertex. The deviation involves an immunized
vertex in $b$ to purchase to an immunized vertex in $b^*$.
\label{fig:1b21}}
\end{minipage}
\end{figure}
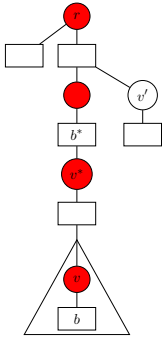

Consider first case 1(b2-1).  Notice that the heavy targeted cut vertices
from $v$ to $r$ are \emph{monotonic}, in that some prefix of the
targeted cut vertices on the path are heavy, and the remaining suffix
is not (though the suffix may be empty). Let $v^*$ be the most shallow
heavy targeted cut vertex on the path from $v$ to $r$: then, all targeted
vertices between $v$ and $v^*$ are heavy, and all heavy targeted cut vertices from
$v$ to $r$ are between $v$ and $v^*$. Since $v^*$ is a cut vertex, it
has at least two child blocks; let $b^*$ be the one which isn't on the
path from $v^*$ to $v$. Again, the connection of $v^*$ to $b^*$ must be through
an immunized vertex; call that vertex $w^*$. Then, $p_v$ is equal to the
probability that an attack will occur on the path from $v$ to $v^*$,
since all targeted vertices on that path are heavy (and all heavy
targeted cut vertices are on that path). 
Consider any child block $b$ of $v$ and let $w$ be an immunized
vertex in $b$ (again $w$ exists because $v$ is targeted).
We will argue about the deviation
of $w^*$ buying an edge to $w$. See Figure~\ref{fig:1b21}.
In the figures in the proof, circles denote targeted cut vertices 
and rectangles denote blocks. Triangles are used to represent 
the subtree of a heavy targeted cut vertex in $\D_r$.

Let $p$ be the probability of attack to any targeted vertex. Since $v$ is a heavy targeted cut vertex and $v$
did not immunize then $p \epsilon n\le \b$, so
$p\leq \b/(\epsilon n) = \b n^{-2/3}/(2\sqrt{\c})$, and the same
argument holds for any heavy targeted cut vertex on the path from $v$ to
$v^*$. Since $p_v > \sqrt{\b\c} n^{-1/3}$ there are
$p_v / p > 2\sqrt{\c/\b} n^{1/3}$ heavy targeted cut vertices (all on the path
from $v$ to $v^*$).  Let $u_{w^*}(x)$ denote the size of the connected
component of $w^*$ after an attack to vertex $x$.  The utility of $w^*$
before this deviation is
\[
\sum_{x \text{ is heavy}} p \cdot u_{w^*}(x) + \sum_{x \text{ is not heavy}}p\cdot  u_{w^*}(x)- \text{ expenditure of } w^*.
\]
After the deviation, Lemma~\ref{lem:same-max} will imply that all
targeted cut vertices between $v$ and $v^*$ will no longer be
targeted. All of these targeted cut vertices are heavy, so $w^*$'s utility will be
exactly

\begin{align*}
& \frac{1}{1-\Pr[\text{attack on path from } v \text{ to } v^*]}\sum_{x \text{ is not on the path from $v$ to $v^*$}}p \cdot u_{w^*}(x)- \text{ expenditure of } w^* -\c\\
&= \frac{1}{1-p_v}\sum_{x \text{ is not heavy}} p \cdot u_{w^*}(x)- \text{ expenditure of } w^* -\c
\end{align*}

where the equality follows from the fact that the set of heavy
targeted cut vertices is exactly the set of all targeted vertices on the
path from $v$ to $v^*$. The scaling in the utility after the deviation
was calculated by the observation that all non-heavy targeted cut
vertices before the deviation remain targeted and none of the targeted
vertices on the path from $v$ to $v^*$ remain targeted after the
deviation.

By definition, $u_{w^*}(x) > n-\epsilon n$ when $x$ is not heavy. So
\[
\sum_{x \text{ is not heavy}}p \cdot u_{w^*}(x) > \left(n-\epsilon n\right) \sum_{x \text{ is not  heavy}}p  = \left(n-\epsilon n\right)\left(1-p_v\right).
\]
Also
\[
\sum_{x \text{ is heavy }} p\cdot u_{w^*}(x) \leq p_v \left(n-\epsilon n\right) - p \sum_{i=1}^{p_v/p} (i-1)
\]
because there are $p_v/p$ heavy targeted cut vertices and each heavy targeted cut vertex cause a
loss in connectivity of at least $\epsilon n$. Furthermore, as we
traverse heavy targeted cut vertices from $v$ towards $v^*$, the amount of loss
in connectivity beyond $\epsilon n$ increases. We simply bound this
increment by $1$ when we go to the next heavy targeted cut vertex on the path
from $v$ to $v^*$ because there is at least one immunized vertex in each
block between consecutive targeted cut vertices.

We will now show that it is strictly beneficial for $w^*$ to buy this
edge. We subtract the utility $w^*$ gets without this deviation from the
(lower bound on the) utility she gets with the deviation, and will
show that this difference is strictly positive:
\begin{align*}
\left(\frac{1}{1-p_v}-1\right)&\sum_{x \text{ is not heavy}}p \cdot u_{w^*}(x) - \sum_{x \text{ is heavy }} p \cdot u_{w^*}(x) -\c \\ 
& > 
p_v\left(n-\epsilon n\right)  - \left(p_v \left(n-\epsilon n\right) - p \sum_{i=1}^{p_v/p} (i-1)\right)-\c
\end{align*}
After simplification and replacing the bounds for $p_v$ and $p_v/p$ we get
\[
p \sum_{i=1}^{p_v/p} (i-1)-\c = p \frac{\frac{p_v}{p}(\frac{p_v}{p}-1)}{2} -\c > \frac{p_v^2}{4p} -\c  = 2\c-\c = \c > 0,
\]
which shows that $w^*$ buying an edge to $w$ is strictly beneficial,
so this cannot be an equilibrium.

Consider case 1(b2-2). Since no targeted cut vertex leaves a connected
component of size at least $(1-\epsilon)n$, all targeted cut vertices are
heavy. Thus, if $p$ is the probability that some targeted vertex is
attacked, $1/p$ is the number of targeted vertices, which is also the
number of heavy targeted cut vertices.

Pick any child block of $r$ that is not in the same subtree as $v$ and
call it $b$; $b$ must exist because $r$ is a cut vertex so it
has at least $2$ child block. Let $w$ be an immunized vertex
in $b$; $w$ must exist because $r$ is a singleton targeted vertex and
can only be connected to non-targeted regions through immunized vertices.
Let
$v_1=v\rightarrow v_2 \rightarrow v_3 \rightarrow \ldots \rightarrow
v_k=r$
denote the path (ignoring the blocks and non-targeted cut vertices on
the path).  Consider the child block of
$v_{k/2}$ on the path from $v$ to $r$. Let $b'$ be such child block
and $w'$ be an immunized vertex in $b'$ (similar to the previous case,
$w'$ exists because the targeted regions are singletons).  Again, we
consider the deviation that $w$ purchase an edge to $w'$ and show that
it is beneficial. See Figure~\ref{fig:1b22}.

\begin{figure}[h]
\centering
\begin{minipage}[c]{.4\textwidth}
\centering
\scalebox{.6}{
\begin{tikzpicture}
[scale = 0.7, every node/.style={circle,draw=black, minimum size=0.7cm}, red node/.style = {circle, fill = red, draw},  gray node/.style = {circle, fill = blue, draw}]
\node [red node] (1) at  (0, 4){$v$};
\node [red node] (2) at  (0, 10){$r$};
\node [red node] (3) at  (0, 7){$$};
\node [draw,rectangle,color=black,minimum width=1cm,minimum height=0.6cm,label=$$] (4) at (0, 8.5) {$$};
\node [draw,rectangle,color=black,minimum width=1cm,minimum height=0.6cm,label=$$] (7) at (-2, 8.5) {$b$};
\node [draw,rectangle,color=black,minimum width=1cm,minimum height=0.6cm,label=$$] (5) at (0, 5.5) {$b'$};
\node [draw,rectangle,color=white,minimum width=0.2cm,minimum height=0.2cm,label=$v_1$] at (1.2,3.2) {};
\node [draw,rectangle,color=white,minimum width=0.2cm,minimum height=0.2cm,label=$v_k$] at (1.2,9.2) {};
\node [draw,rectangle,color=white,minimum width=0.2cm,minimum height=0.2cm,label=$v_{k/2}$] at (1.2,6.2) {};
\draw (0,5)--(-1,3.5)--(1,3.5)--cycle;
\draw(2) to (4);\draw(4) to (3);\draw(2) to (7);\draw(3) to (5);\draw(1) to (5);
\end{tikzpicture}}
\end{minipage}
\begin{minipage}[c]{.5\textwidth}
\centering
\caption{Case 1(b2-2): Heavy targeted cut vertices are in red and we denote such vertices
 on the path from $r$ to $v$ by $v_k$ to $v_1$. The deviation involves an immunized vertex 
 in $b$ to purchase an edge to an immunized vertex in $b'$.
\label{fig:1b22}}
\end{minipage}
\end{figure}
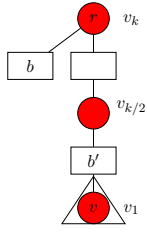

Notice that $k = 1/p$, and furthermore, since all targeted cut
vertices are heavy, $p_v = 1$.  After the deviation cut vertices
$v_{k/2}, \ldots, v_k$ become non-targeted and cut vertices
$v_1, \ldots v_{k/2-1}$ remain targeted by applying
Lemma~\ref{lem:same-max}. Furthermore, there are no other targeted
vertices.

Let $p$ be the probability of attack to any targeted vertex before the
deviation. Since $v$ is a heavy targeted cut vertex and $v$ did not immunize then $p \epsilon n\le \b$ or
$p\leq \b/(\epsilon n) = \b n^{-2/3}/(2\sqrt{\c})$. Since $p_v = 1$
there are $k = p_v / p > n^{2/3}2\sqrt{\c}/\b$ heavy targeted cut vertices
before the deviation.  Let $u_w(x)$ denote the size of the connected
component of $w$ after an attack to $x$.  The utility of $w$ before
the deviation is
$$
\sum_{x \text{ is targeted in } G} p \cdot u_w(x) - \text{ expenditure of } w.
$$
Let $p'$ be the probability of attack to any targeted vertex after the deviation. The utility of $w$
after the deviation is 
$$
\sum_{x \text{ is targeted in } G'}p' \cdot u_w(x)- \text{ expenditure of } w -\c.
$$
Note that by construction $p' \geq 2p$. Also
$u_w(v_{1+i}) - u_w(v_{k/2+1+i}) \geq k/2$ because there is at least a
block of size 1 between each consecutive targeted cut vertices before
the deviation. Combining these two facts and subtracting the second
term from the first term we get
\begin{align*}
& \sum_{x \text{ is targeted in } G'}p' \cdot u_w(x) - \sum_{x \text{ is targeted in } G} p \cdot u_w(x) - \c \\
& = p'  \sum_{x =  v_i, i\in [k/2, k]} u_w(x)  - p  \left( \sum_{x =  v_i, i\in [k/2, k]} u_w(x)  + \sum_{x =  v_i, i\in [1, k/2)}  u_w(x)\right) - \c\\
& \geq 2p  \sum_{x =  v_i, i\in [k/2, k]} u_w(x)  - p  \left( \sum_{x =  v_i, i\in [k/2, k]} u_w(x)  + \sum_{x =  v_i, i\in [1,k/2)}  u_w(x)\right) - \c\\
& \geq p  \sum_{x =  v_i, i\in [k/2, k]} u_w(x)  - p  \sum_{x =  v_i, i\in [1,k/2)}  u_w(x) - \c\\
& \geq 
p\sum_{i=1}^{1/p} \frac{k}{2} -\c = p\frac{\frac{1}{p}(\frac{1}{p}+1)}{2} \frac{k}{2} -\c \geq \frac{k}{4} -\c > 0,
\end{align*}
because $k$ is the number of heavy targeted cut vertices before the deviation
which is at least $n^{2/3}2\sqrt{\c}/\b$ (grows with $n$). So this deviation
is strictly beneficial for $w$; a contradiction.

In case (2), let $r'$ be a cut vertex that is the \emph{lowest common
  ancestor} of vertices in $\D_{r}$. If $r'\ne r$, we root the tree on
$r'$ and repeat the process of finding heavy targeted cut vertices. Note that
$\D_{r}\subseteq \D_{r'}$ since we might add some additional heavy targeted cut
vertices to $\D_{r'}$ (vertices that used to be ancestors of $r'$ in $T$).

Observe that the vertices in $\D_{r'}$ and
the targeted cut vertices on the path from some $v\in\D_{r'}$ to $r'$
(new root) are the only possible heavy targeted cut vertices in the network.
Let $p_F$ denote the total probability of attack to all the heavy targeted cut vertices.  
We consider two cases: 2(a)
$ p_F \leq 8\sqrt{\c} n^{-1/3}$, and 2(b)
$p_F > 8\sqrt{\c}n^{-1/3}$. We show that in case 2(a) the welfare
is as claimed in the statement of Theorem~\ref{thm:welfare-new}, and
that case 2(b) cannot happen.

First consider case 2(a). In this case, with probability of
$1-p_F$, the size of the largest connected component after an attack
is at least $(1-\epsilon)n$. Hence the welfare in this case is at
least
\begin{align*}
\left(1-p_F\right)\left(\left(1-\epsilon\right)n\right)^2-\cm
&\geq
\left(1-8\sqrt{\c}n^{-1/3}\right)\left(1-2\epsilon\right)n^2-\cm\\
&=
\left(1-8\sqrt{\c}n^{-1/3}\right)\left(1-\frac{4 \sqrt{\c}}{n^{1/3}}\right)n^2-\cm\\
&> n^2-\left(8\sqrt{\c}+4\sqrt{\c}\right)n^{5/3}-\cm= n^2-O\left(n^{5/3}\right).
\end{align*}


Next, we consider case 2(b), where $p_F > 8\sqrt{\c}n^{-1/3}$. First, 
note that by re-rooting there are at least two sub-tress of the new root $r'$
that contain heavy targeted cut vertices in $\D_{r'}$. 
Second, observe that by pigeonhole principle there exists
a vertex $v\in\D_{r'}$ such that the probability of attack to (heavy) targeted 
cut vertices on the path from $v$ to $r'$ (counting both $v$ and $r'$) is at least $8\sqrt{\c} n^{-1/3}\epsilon \geq 16 \c n^{-2/3}>8\c n^{-2/3}$.
Let $b$ be any child block of $v$ and $w$ any immunized vertex in $b$.
We propose a deviation for $w$ that strictly increases her utility.

Before we describe the deviation, consider the subtree of $r'$ that contains $v$ and let $N$ denote the size of the 
subtree. We first assume $N < n - n^{5/6}$ and relax this assumption later. 
Intuitively, we would like to make targeted cut vertices that would cause $w$ a connectivity loss of strictly bigger than
$n^{5/6}/2$ non-targeted.\footnote{By connectivity loss of strictly bigger than
$n^{5/6}/2$ we mean that the size of the connected component of $w$ after the attack is 
strictly less than $n-n^{5/6}/2$.}
Observe that all the heavy targeted cut vertices on the path from $v$ to $r'$ satisfy this property.
Also other than the heavy targeted cut vertices on the path from $v$ to $r'$ 
there might be other vertices that would cause a connectivity loss of strictly bigger than 
$n^{5/6}/2$ for $w$. So $w$'s deviation should make such targeted vertices non-targeted as well.
Furthermore, we would like to ensure that any targeted vertex that would cause $w$ a connectivity loss of at most $n^{5/6}/2$
targeted before the deviation would remain targeted after the deviation. 
Let $p$ denote the probability of attack to a targeted cut vertex before the deviation.
Such deviation implies that the 
connectivity benefit of $w$ after the deviation is increased by at least 
\[
p \left (\frac{8\c n^{-2/3}}{p} \right)\frac{n^{5/6}}{2} = 4\c n^{1/6}
\]
because there were at least $8\c n^{-1/3}/p$ targeted cut vertices which cause $w$ to have a connectivity
loss of at least $n^{5/6}$ before the deviation (these are exactly the heavy targeted cut vertices on the path from $v$ to $r'$). 
Moreover, after the deviation the connectivity loss of $w$ in any attack 
is bounded by $n^{5/6}/2$.
Finally, we need to guarantee that the deviation does not involve buying a lot of additional edges for $w$. 
In what follows we propose a deviation with at most $2 n^{1/6}$ additional edge purchases (for a cost of $2\c n^{1/6}$) 
which shows that $w$'s deviation was beneficial; a contradiction to network being an equilibirum. 
\begin{figure}[h]
\centering
\begin{minipage}[c]{.4\textwidth}
\centering
\scalebox{.6}{
\begin{tikzpicture}
[scale = 0.55, every node/.style={circle,draw=black, minimum size=0.7cm}, red node/.style = {circle, fill = red, draw},  gray node/.style = {circle, fill = blue, draw}]
\node [red node] (1) at  (0, 3){$v$};
\node [red node] (2) at  (0, 10){$$};
\node [red node] (10) at  (-1, 13){$r$};
\node [red node] (3) at  (0, 7){$$};
\node (12) at  (-2, 10){$$};
\node [draw,rectangle,color=black,minimum width=1cm,minimum height=0.6cm,label=$$] (11) at (0, 11.5) {$$};
\node [draw,rectangle,color=black,minimum width=1cm,minimum height=0.6cm,label=$$] (4) at (0, 8.5) {$$};
\node [draw,rectangle,color=black,minimum width=1cm,minimum height=0.6cm,label=$$] (7) at (-2, 11.5) {$b'$};
\node [draw,rectangle,color=black,minimum width=1cm,minimum height=0.6cm,label=$$] (5) at (0, 5.5) {$$};
\node [draw,rectangle,color=black,minimum width=1cm,minimum height=0.6cm,label=$$] (6) at (0, 1.5) {$b$};
\draw (0,4.5)--(-2,0.9)--(2,0.9)--cycle;
\draw (-2,11)--(-3.2,9.4)--(-0.8,9.4)--cycle;
\draw (1) to (6);\draw(2) to (4);\draw(4) to (3);\draw(3) to (5);\draw(1) to (5);
\draw (10) to (11); \draw (10) to (7);\draw(11) to (2);\draw(12) to (7);
\end{tikzpicture}
}
\end{minipage}
\begin{minipage}[c]{.5\textwidth}
\centering
\caption{The deviation for $w\in b$ is to purchase an edge to $w'\in b'$. 
The red denote the  targeted cut vertices that would result in a connectivity
loss of strictly bigger than $n^{5/6}/2$ for $w$.
\label{fig:2b-1}}
\end{minipage}
\end{figure}
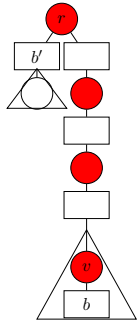
\begin{figure}[h]
\begin{minipage}[c]{.4\textwidth}
\centering
\scalebox{.6}{
\begin{tikzpicture}
[scale = 0.55, every node/.style={circle,draw=black, minimum size=0.7cm}, red node/.style = {circle, fill = red, draw},  gray node/.style = {circle, fill = blue, draw}]
\node [red node] (1) at  (0, 3){$v$};
\node [red node] (2) at  (0, 10){$$};
\node [red node] (10) at  (-2.5, 13){$r$};
\node [red node] (3) at  (0, 7){$$};
\node [draw,rectangle,color=black,minimum width=1cm,minimum height=0.6cm,label=$$] (11) at (0, 11.5) {$$};
\node [draw,rectangle,color=black,minimum width=1cm,minimum height=0.6cm,label=$$] (4) at (0, 8.5) {$$};
\node [draw,rectangle,color=black,minimum width=1cm,minimum height=0.6cm,label=$$] (5) at (0, 5.5) {$$};
\node [draw,rectangle,color=black,minimum width=1cm,minimum height=0.6cm,label=$$] (6) at (0, 1.5) {$b$};
\node  (201) at  (-5, 3){$$};
\node [red node] (202) at  (-5, 10){$$};
\node [red node] (203) at  (-5, 7){$u^*$};
\node [draw,rectangle,color=black,minimum width=1cm,minimum height=0.6cm,label=$$] (2011) at (-5, 11.5) {$$};
\node [draw,rectangle,color=black,minimum width=1cm,minimum height=0.6cm,label=$$] (204) at (-5, 8.5) {$$};
\node [draw,rectangle,color=black,minimum width=1cm,minimum height=0.6cm,label=$$] (205) at (-5, 5.5) {$b_{u^*}$};
\draw(202) to (204);\draw(204) to (203);\draw(203) to (205);\draw(201) to (205);
\draw (10) to (2011);\draw(2011) to (202);
\draw (-5,4.5)--(-6,2.4)--(-4,2.4)--cycle;
\draw (0,4.5)--(-2,0.9)--(2,0.9)--cycle;
\draw (1) to (6);\draw(2) to (4);\draw(4) to (3);\draw(3) to (5);\draw(1) to (5);
\draw (10) to (11);\draw(11) to (2);
\end{tikzpicture}
}
\end{minipage}
\begin{minipage}[c]{.5\textwidth}
\centering
\caption{The deviation for $w\in b$ is to purchase an edge to $w_{u^*} \in b_{u^*}$.
The red denote the  targeted cut vertices that would result in a connectivity
loss of strictly bigger than $n^{5/6}/2$ for $w$.
\label{fig:2b-2}}
\end{minipage}
\end{figure}
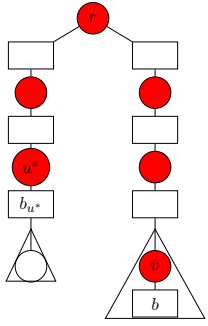

First consider the subtrees of $r'$ that do not contain $v$. If there are no targeted cut vertices that
would cause a connectivity loss of strictly bigger than $n^{5/6}/2$, $w$ would purchase an edge to an immunized vertex $w'$ in a child block $b'$ of 
$r'$ that is not on the path from $r'$ to $v$. $b'$ exists because $r'$ is a cut vertex and $w'\in b'$
exists because targeted regions are singletons. See  Figure~\ref{fig:2b-1}. This deviation by Lemma~\ref{lem:same-max}
would only cause all the heavy targeted cut vertices on the path from $v$ to $r'$ (including $v$ and $r'$) to become non-targeted 
without adding any new targeted vertex. Now suppose there are targeted cut vertices 
on the subtrees of $r'$ that do not contain $v$ which would cause $w$ to have a connectivity loss of strictly bigger than $n^{5/6}/2$.
Let $u$ be one such cut vertex. Clearly an attack to any targeted vertex on the path from $u$ to $r'$ would also cause $w$ 
to have a connectivity loss of strictly bigger than $n^{5/6}/2$. 
We define a targeted cut vertex $u^*$
in the subtrees of $r'$ that do not contain 
$v$ as \emph{troublesome} if (1) an attack to any targeted cut vertex in $T_{u^*}$ would cause $w$ a connectivity loss of at most 
$n^{5/6}/2$ and (2) an attack to $u^*$ (and any targeted cut vertex on the path from $u^*$ to $r'$) would cause 
a connectivity loss of strictly bigger than $n^{5/6}/2$ for $w$. Let $U^*$ be the set of all troublesome vertices in the subtrees of $r'$
that do not contain $v$. Observe that $|U^*|\le 2(n-N)/n^{5/6}$. This is because $|T_{u^*}| \geq n^{5/6}/2$ for all $u^*\in U^*$ and the
size of subtrees of $r'$ that contain $v$ is $n-N$. For any $u^*\in U^*$ let $b_{u^*}$ denote the child block of $u^*$
and $w_{u^*}$ denote any immunized vertex in $b_{u^*}$. Consider the deviation that $w$
purchases an edge to $w_{u^*}$. See Figure~\ref{fig:2b-2}. This deviation by Lemma~\ref{lem:same-max} would make all the targeted cut vertices
on the path from $v$ to $u^*$ (including both $v$ and $u^*$) non-targeted. Note that these are the only targeted vertices 
that would become non-targeted after the deviation. Furthermore, by Lemma~\ref{lem:same-max}
the deviation does not add any new targeted vertex. 
So $\max\{1, |U^*|\}$ edge purchases 
guarantee that in the subtrees of $r'$ that 
do not contain $v$ all attacks would cause a connectivity loss of at most $n^{5/6}/2$ for $w$.

\begin{figure}[h]
\centering\begin{minipage}[c]{.4\textwidth}
\centering
\scalebox{.5}{
\begin{tikzpicture}
[scale = 0.6, every node/.style={circle,draw=black, minimum size=0.7cm}, red node/.style = {circle, fill = red, draw},  gray node/.style = {circle, fill = blue, draw}]
\node [red node] (1) at  (0, 3){$v$};
\node [red node] (2) at  (0, 10){$$};
\node [red node] (10) at  (-1, 13){$r$};
\node [red node] (3) at  (0, 7){$$};
\node (16) at  (3, 3){$$};
\node (12) at  (-2, 10){$$};
\node [draw,rectangle,color=black,minimum width=1cm,minimum height=0.6cm,label=$$] (11) at (0, 11.5) {$$};
\node [draw,rectangle,color=black,minimum width=1cm,minimum height=0.6cm,label=$$] (4) at (0, 8.5) {$$};
\node [draw,rectangle,color=black,minimum width=1cm,minimum height=0.6cm,label=$$] (13) at (3, 8.5) {$$};
\node [draw,rectangle,color=black,minimum width=1cm,minimum height=0.6cm,label=$$] (7) at (-2, 11.5) {$$};
\node [draw,rectangle,color=black,minimum width=1cm,minimum height=0.6cm,label=$$] (5) at (0, 5.5) {$$};
\node [draw,rectangle,color=black,minimum width=1cm,minimum height=0.6cm,label=$$] (6) at (0, 1.5) {$b$};
\node [draw,rectangle,color=black,minimum width=1cm,minimum height=0.6cm,label=$$] (14) at (3, 5.5) {$b_{u}$};
\node [red node] (15) at  (3, 7){$u$};
\draw (0,4.5)--(-2,0.9)--(2,0.9)--cycle;
\draw (3,4.5)--(2,2.3)--(4,2.3)--cycle;
\draw (-2,11)--(-3.2,9.4)--(-0.8,9.4)--cycle;
\draw (1) to (6);\draw(2) to (4);\draw(4) to (3);\draw(3) to (5);\draw(1) to (5);
\draw (10) to (11); \draw (10) to (7);\draw(11) to (2);\draw(12) to (7);
\draw (13) to (15); \draw (14) to (15); 
\draw (14) to (16);
\draw (13) to (2);
\end{tikzpicture}}
\end{minipage}
\begin{minipage}[c]{.5\textwidth}
\caption{The red denote the  targeted cut vertices that would result in a connectivity
loss of strictly bigger than $n^{5/6}/2$ for $w$. The deviation is for $w\in b$ to purchase an edge to $w_{u} \in b_{u}$.
\label{fig:2b-3}}
\end{minipage}
\end{figure}
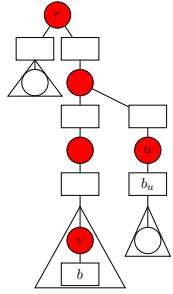

Now consider the subtree of $r'$ that contains $v$. In this subtree other than
the heavy targeted cut vertices on the path from $v$ to $r'$ there might be other targeted
cut vertices that would result in a connectivity loss of strictly bigger than $n^{5/6}/2$ to $w$.
See Figure~\ref{fig:2b-3}.
Again we define a targeted cut vertex $u$
in the subtree of $r'$ that contains
$v$ \emph{troublesome} if (1) an attack to any targeted cut vertex in $T_{u}$ would cause $w$ a connectivity loss of at most 
$n^{5/6}/2$, (2) an attack to $u$ (and any targeted cut vertex on the path from $u$ to $r'$) would cause 
a connectivity loss of strictly bigger than $n^{5/6}/2$ for $w$ and (3) $u$ is not on the path from $v$ to $r'$. 
Let $U$ be the set of all troublesome vertices in the subtree of $r$
that contains $v$. Observe that $|U|\le 2N/n^{5/6}$. This is because $|T_{u}| \geq n^{5/6}/2$ for all $u\in U$ and the
size of subtree of $r'$ that contains $v$ is $N$. For any $u\in U$, let $b_{u}$ denote the child block of $u$
and $w_{u}$ denote any immunized vertex in $b_{u}$. Consider the deviation that $w$
purchases an edge to $w_{u}$. See Figure~\ref{fig:2b-2}. This deviation by Lemma~\ref{lem:same-max} would make all the targeted cut vertices
on the path from $v$ to $u$ (including both $v$ and $u$) non-targeted. Note that these are the only targeted vertices 
that would become non-targeted after the deviation. Furthermore, by Lemma~\ref{lem:same-max}
the deviation does not add any new targeted vertices. 
So $|U|$ edge purchases guarantees that all attacks in the
subtree of $r'$ that cause $v$ would cause a connectivity loss of at most $n^{5/6}/2$ for $w$ after the deviation.

So the total number of edge purchases in the deviation described for $w$ 
is  $\max\{1, |U^*|\} + |U| \leq 2 n^1/6$, as claimed.

Finally, we show that the assumption we made about $N$ (the size of the subtree of $r'$
that contains $v$) can be made without loss of generality. In particular we show
that in case 2(b) there exists a targeted cut vertex $r^*$ such that if we root
the block-cut tree in $r^*$ then: (1) there exists a $v^*\in\D_{r^*}$
such that the probability of attack to the heavy targeted cut vertices on the path from
$v^*$ to $r^*$ is at least $8\c n^{-2/3}$ and (2) the size of the subtree
of $r^*$ that contains $v^*$ is less than $n-n^{5/6}$.

By contradiction, suppose not. Find a vertex $v\in H_{r'}$ such that the probability of
attack to (heavy) targeted cut vertices on the path from $v$ to $r'$ is at least $8\sqrt{\c} n^{-1/3}/|\D_{r'}| \geq 8\sqrt{\c} n^{-1/3}\epsilon =  16\c n^{-2/3}$
(such vertex exists by pigeonhole principle). 
The choice of re-rooting to $r'$ in case 2 guarantees that there exist a vertex $v'\in\D_{r'}$ ($v'\ne v$) such that 
$v$ and $v'$ are in different subtrees of $r'$. In particular the probability of attack on the path from $v'$ to $v$ is also at least $16\c n^{-2/3}$.
If $|T_{v}|$ (the size of subtree of $v$) is at least $n^{5/6}$ then we are done because $r^*=v$ and $v'$ witness the property
that we claim. Otherwise, we show that there exists a targeted cut vertex $r^*$ on the path from $v$ to $r'$ such that 
if we root the block cut-tree on $r^*$ then either $v$ or $v'$ satisfy the property we were looking for.

Note that in this case the size of subtree of $r'$ that contains $v$ is at least $n-n^{5/6}$. Furthermore, $|T_{v}| < n^{5/6}$.
Let $r_1, \ldots, r_k$ be the (heavy) targeted cut vertices on the path from $r'$ to $v$. 
As we move along from $r_1$ to $r_k$ (increasing the index $i$ in $r_i$ along the path) the size of the subtree of $r_i$ that contains $v$ decreases.
Furthermore, as we move along from $r_k$ to $r_1$ (decreasing the index $i$ in $r_i$ along the path) $|T_{r_i}|$ increases. So there 
exist $i$ and $j$ such that (1) the size of subtree of $r_i'$ that contains $v$ is less than $n-n^{5/6}$ for all $i'\geq i$ and (2) $|T_{r_{j'}}|> n^{5/6}$ for all $j'\leq j$.
Note that either the probability of attack on the path from $r_i$ to $r'$ or the probability of attack on the path from $r_j$ to $v$ is at least 
$16\c n^{-2/3}/2 = 8\c n^{-2/3}$ (remind that the probability of attack between $v$ and $v'$ is at least $16 \c n^{-2/3}$).
In the former case $r^* = r_i$ and $v$ witness the property we claimed. In the latter case $r^* = r_j$ and $v'$ witness the property we claimed.
\qed

\subsection{Proof of Lemma~\ref{ex:swap-disconnected}}
\noindent\emph{Proof of Lemma~\ref{ex:swap-disconnected}.}
\begin{figure}[h]
\centering\begin{minipage}[c]{.5\textwidth}
\centering
\scalebox{.5}{
\begin{tikzpicture}
[scale=0.6, every node/.style={circle,draw=black}, gray node/.style = {circle, fill = blue, draw}, red node/.style = {circle, fill = red, draw}]
\node [red node] (1) at  (0, 1){};
\node [red node] (3) at  (-2, 0){};\node [red node] (4) at  (-2, 1){};
\node [red node] (5) at  (-2, 3){};\node [red node] (6) at  (-2, 2){};
\node [red node] (7) at  (-2, 4){};\node [red node] (8) at  (-2, -1){};
\draw[->] (3) to (1);\draw[->] (4) to (1);\draw[->] (5) to (1);\draw[->] (6) to (1);\draw[->] (7) to (1);\draw[->] (8) to (1);
\node [red node] (11) at  (3, 1){};
\node [red node] (13) at  (1, 0){};\node [red node] (14) at  (1, 1){};
\node [red node] (15) at  (1, 3){};\node [red node] (16) at  (1, 2){};
\node [red node] (17) at  (1, 4){};\node [red node] (18) at  (1, -1){};
\draw[->] (13) to (11);\draw[->] (14) to (11);\draw[->] (15) to (11);\draw[->] (16) to (11);\draw[->] (17) to (11);\draw[->] (18) to (11);
\node [red node] (21) at  (0, -5){};
\node [red node] (23) at  (-2, -6){};\node [red node] (24) at  (-2, -5){};
\node [red node] (25) at  (-2, -3){};\node [red node] (26) at  (-2, -4){};
\node [red node] (27) at  (-2, -2){};\node [red node] (28) at  (-2, -7){};
\draw[->] (23) to (21);\draw[->] (24) to (21);\draw[->] (25) to (21);\draw[->] (26) to (21);\draw[->] (27) to (21);\draw[->] (28) to (21);
\node [red node] (111) at  (3, -5){};
\node [red node] (113) at  (1, -6){};\node [red node] (114) at  (1, -5){};
\node [red node] (115) at  (1, -3){};\node [red node] (116) at  (1, -4){};
\node [red node] (117) at  (1, -2){};\node [red node] (118) at  (1, -7){};
\draw[->] (113) to (111);\draw[->] (114) to (111);\draw[->] (115) to (111);\draw[->] (116) to (111);\draw[->] (117) to (111);\draw[->] (118) to (111);
\node [red node] (41) at  (7, -2){};
\node [gray node] (42) at  (6, -2){};\node [red node] (43) at  (5, -1){};\node [red node] (44) at  (5, -2){};\node [red node] (45) at  (5, -3){};
\node [gray node] (46) at  (8, -2){};\node [red node] (47) at  (9, -1){};\node [red node] (48) at  (9, -2){};\node [red node] (49) at  (9, -3){};
\draw[->] (41) to (42);\draw[->] (41) to (46);
\draw[->] (43) to (42);\draw[->] (44) to (42);\draw[->] (45) to (42);
\draw[->] (47) to (46);\draw[->] (48) to (46);\draw[->] (49) to (46);
\node [draw,rectangle,color=black,minimum width=0.5cm,minimum height=0.5cm,label=$e$] at (1,4) {};
\node [draw,rectangle,color=black,minimum width=0.5cm,minimum height=0.5cm,label=$d$] at (3,1) {};
\node [draw,rectangle,color=black,minimum width=0.5cm,minimum height=0.5cm,label=$a$] at (5,-1) {};
\node [draw,rectangle,color=black,minimum width=0.5cm,minimum height=0.5cm,label=$b$] at (7,-2) {};
\node [draw,rectangle,color=black,minimum width=0.5cm,minimum height=0.5cm,label=$c$] at (8,-2) {};
\end{tikzpicture}}
\end{minipage}
\begin{minipage}[c]{.45\textwidth}
\centering
\caption{An example of non-trivial swapstable (and hence linkstable) equilibrium network with respect to the \maxdisrupt~adversary with 
more that one connected component. $\c=1.5$ and $\b=6.5$.}
\label{fig:example-c-1-swap}
\end{minipage}
\end{figure}
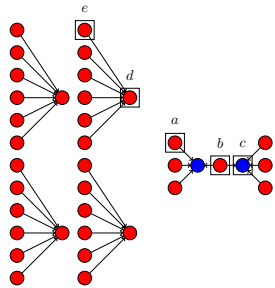
We consider the non-trivial network in Figure~\ref{fig:example-c-1-swap} and show that this network is a swapstable equilibrium network with respect to the \maxdisrupt~adversary when
$\c=1.5$ and $\b=6.5$.~\footnote{We point out that the network in Figure~\ref{fig:example-c-1-swap} is not a Nash equilibrium network with respect to the \maxdisrupt~adversary because
the immunized vertex of type (b) has a profitable Nash deviation which is to buy an edge to all the connected components with no immunization plus an additional edge 
to the other immunized vertex.}
By symmetry, we just consider the deviations in 5 types of vertices denoted by (a)-(e) in Figure~\ref{fig:example-c-1-swap}. 
We show that none of the deviations
can strictly increase the utility. 

For type (a) vertices the utility pre-deviation is is $9 (4/5) + 4 (1/5) - \c = 8-\c$. The deviations for such vertex are as follows:
\begin{enumerate}
\item dropping the purchased edge.
\item keeping the purchased edge and adding one more edge.
\item swapping the purchased edge.
\item immunizing.
\item dropping the purchased edge and immunizing.
\item keeping the purchased edge and adding one more edge and immunizing
\item swapping the purchased edge and immunizing.
\end{enumerate}

In case 1, the utility would become 1 after the deviation so as long as $\c < 7$, the deviation is not beneficial.

In case 2, adding an edge to any vertex outside of the current connected component, would result in the vertex of type (a) to
form the unique targeted region. So the utility would become $-2\c$ which is strictly less than the current utility. Adding an edge to a vertex inside 
of the current component can make vertex $(c)$ non-targeted. So the utility of type (a) vertex will become $9-2\c$ after the  deviation 
which is not beneficial since $\c=1.5$.

In case 3, swapping the edge to any vertex outside of the current connected component would result in 
the vertex of type (a) to
form the unique targeted region. So the utility would become $-\c$ which is strictly less than the current utility.
Swapping the edge to a vertex inside of the current component will either not change the utility or make (c) the 
unique targeted vertex. In the latter case, the utility of the vertex after the deviation would be at most $5-\c$ 
which is strictly smaller 
than the current utility of $8-\c$.

Type (a) vertices are not targeted, so immunization in case 4 only increases the cost without adding any benefit.

In case 5, after dropping the edge, the vertex remains non-targeted so the deviation in case 1 strictly
dominates the deviation in case 5.

In case 6, adding an edge to any vertex in a component outside of the current connected component, 
would result in that component to be the unique targeted region; so the added edge would not 
provide any direct benefits but it will cause the vertex of type $(c)$ to be non-targeted. 
However, we showed in
case 2 that this would not still be a beneficial deviation.
Adding an edge to a vertex inside of the current component of a type (a) vertex can make the vertex of type $(c)$ non-targeted. 
However, again, this deviation is strictly worse than the deviation in case 2 because the vertex remain non-targeted 
after the deviation so the immunization is not needed.

In case 7, swapping the edge to any vertex outside of the current connected component would result in 
that component to be the unique targeted region. So the utility of the type (a) vertex after the deviation 
would become $1-\c-\b$ which is strictly less than the 
current utility of $8-\c$.
Swapping the edge to a vertex inside of the current component will either not change the utility or make (c) the 
unique targeted vertex. In the latter case, the utility of the vertex after the deviation would be at most $5-\c-\b$ which is strictly smaller 
than the current utility of $8-\c$.

For type (b) vertices, the utility pre-deviation is is $9 (4/5) - \b = 0.7$ and the deviations are as follows:
\begin{enumerate}
\item changing the immunization (with or without adding an edge).
\item adding an edge.
\end{enumerate}

In case 1, the vertex would form the unique targeted region if she changes her immunization (regardless of whether 
she adds an edge or not). So this case will not happen as long as the current utility of type (b) vertex is strictly bigger than zero.
Since the utility of type (b) vertex is positive before the deviation, then the deviation is not beneficial. 

In case 2, adding an edge to any vertex outside of the current connected component, 
would result in that component to be the unique targeted region; so the added edge would not provide any benefit.
However, the added edge would result in vertex of 
The best edge in the current component to purchase an edge to is the other immunized vertex. In this case
the utility of type (b) vertex would become $9-\b-\c < 9(4/5)-\b$. So this deviation is not beneficial for the type (b) vertex either.

Next consider the only type (c) vertex. The utility of such vertex is $9(4/5)-2\c$. The deviations of such vertex are as follows:
\begin{enumerate}
\item dropping one of the purchased edges.
\item keeping the purchased edges and adding one more edge.
\item swapping one of the purchased edges.
\item immunizing.
\item dropping one of the purchased edges and immunizing.
\item keeping the purchased edges, adding one more edge and immunizing.
\item swapping one of the purchased edges and immunizing.
\end{enumerate} 

In case 1, dropping the edge would make (c) non-targeted but her benefit after the deviation would be
$5-\c$ which is strictly smaller than $9(4/5)-2\c$ when $\c=1.5$.

In case 2, adding an additional edge to any vertex outside of the current connected component, would make (c)
part of the unique targeted region; so (c)'s utility after the deviation would be $-3\c$ which is strictly less than her current
utility. Furthermore, in her current component, (c) is already connected to both type (b) vertices. So adding an edge 
to any of type (a) vertices again would result in (c) to form the unique targeted region; which is not beneficial.

In case 3, swapping one the edges to 
to any vertex outside of the current connected component, would make (c)
part of the unique targeted region; so (c)'s utility after the deviation would be $-2\c$ which is strictly less than her current
utility. In her connected component, if (c) swaps one of her edges to a type (a) vertex, then again she would 
form the unique targeted region; so this deviation is not beneficial as well.

In case 4, after the immunization, vertex (c) become non-targeted so her utility would be $9-2\c-\b$ after deviation
which is strictly less than $9(4/5)-2\c$ when $b=6.5$.

In case 5, after dropping an edge the vertex (c) becomes non-targeted. So this deviation is strictly 
dominated by the deviation in case 1, which is also not beneficial. 

In case 6, after immunization no vertex in the connected component of (c) is targeted so adding an edge
inside of her current component would be redundant. Outside of her current component, adding an edge
would form a unique targeted region which would result in only an additional edge cost without any 
connectivity benefit for (c). 

In case 7, swapping one the edges 
to any vertex outside of the current connected component of (c) would make that component the unique targeted region.
(c)'s utility in this case would be $5-2\c-\b$ which is strictly less than her current utility.
Furthermore, after immunization, no vertex in (c)'s connected component is targeted and (c) requires at least two 
edges to remain connected to all the vertices in her connected component. So swapping an edge in this case would 
not be beneficial as well.

For type (d) vertices, the utility pre-deviation is  $7(4/5)-\c$. The deviations of one such vertex are as follows:
\begin{enumerate}
\item dropping the purchased edge.
\item keeping the purchased edge and adding one more edge.
\item swapping the purchased edge.
\item immunizing
\item dropping the purchased edge and immunizing.
\item keeping the purchased edge and adding one more edge and immunizing.
\item swapping the purchased edge and immunizing.
\end{enumerate}

In case 1, the utility after the deviation would be $1$ which is strictly less than the current utility.

In case 2, adding an edge inside of her current connected component would be redundant. Adding
an edge to any other vertex would make the pre-deviation connected component of type (d) vertex as
the unique targeted region. In which case, her utility becomes $-2\c$ which is less than her current utility.

In case 4, after immunization, no vertex in the current connected component of the type (d) vertex remains 
targeted. Hence her utility would be $7-\c-\b$ after deviation which is less than her current utility.

In case 3, swapping the edge to any other vertex in her current connected component would 
not change (d)'s utility. Swapping her edge to any other connected component with no immunized
vertex would make (d) part of the unique targeted region. In which case her utility would be $-\c$; strictly less
than her current utility. Swapping her edge to any vertex in the connected component with the immunized
vertex would make (c) the unique targeted region. In which case (d)'s utility would be at most $5-\c$
which is strictly less than her current utility.

In case 5, after dropping the purchased edge, the type (d) vertex would not be targeted anymore. So this 
deviation is strictly dominated by the deviation in case 1 which is also not beneficial.

In case 6, after immunization, no vertex in the connected component of (d) remains targeted. 
Adding an edge inside of her current connected component would be redundant.
Adding an edge to a connected component with no immunization would make such connected component
the unique targeted region. In which case the type (d) vertex achieves no connectivity benefit and only suffers the
linkage cost. Adding an edge to any vertex in the connected component with immunization would make (c)
the unique targeted region. In this case (d)'s utility would be at most $11-2\c-\b$ which is still less
than her current utility.

In case 7, swapping the edge to any other vertex in her current connected component would 
not change (d)'s utility. Swapping her edge to any other connected component with no immunized
vertex would make that component the unique targeted region. In which case (d)'s utility would be $1-\c-\b$; 
strictly less than her current utility. Swapping her edge to any vertex in the connected component with the immunized
vertex would make (c) the unique targeted region. In which case (d)'s utility would be at most $5-\c-\b$
which is strictly less than her current utility.

For type (e) vertices, the utility pre-deviation is $7(4/5)$. The deviations of one such vertex are as follows.
\begin{enumerate}
\item adding an edge.
\item immunizing.
\item adding an edge and immunizing.
\end{enumerate}

In case 1, adding an edge inside of her current connected component would be redundant. Adding
an edge to any other vertex would make the pre-deviation connected component of type (e) vertex as
the unique targeted region. In which case, her utility becomes $-\c$ which is less than her current utility.

In case 2, after immunization, no vertex in the current connected component of the type (e) vertex remains 
targeted. Hence her utility would be $7-\b$ after deviation which is less than her current utility.

In case 3, after immunization, no vertex in the connected component of (e) remains targeted. 
Adding an edge inside of her current connected component would be redundant.
Adding an edge to a connected component with no immunization would make such connected component
the unique targeted region. In which case the type (e) vertex achieves no connectivity benefit and only suffers the
linkage cost. Adding an edge to any vertex in the connected component with immunization would make (c)
the unique targeted region. In this case (e)'s utility would be at most $11-\c-\b$ which is still less
than her current utility.
\qed
\section{Convergence of Best Response Dynamics}
\label{sec:br-cycles}
  We conjectured the general and fast convergence of swapstable (and linkstable) 
  dynamics with respect to \maxcarnage~adversary
  in Section~\ref{sec:exp}. However, this conjecture needs some specification. 
In this section, we show in Example~\ref{ex:nash-cycles} that Nash (and also swapstable and
linkstable) best response dynamics can cycle. 
We again 
focus on the \maxcarnage~adversary and show that cycles can happen in best response dynamics
when we start from a specific initial graph, the players best respond in a fix order and ties
are broken adversarially.
However, this construction heavily relies on a worst-case rule for breaking best response ties, and thus we suspect
the more natural variant with randomized ordering and randomized tie-breaking converges generally. Also, 
to our knowledge, standard potential
game arguments do not seem to apply here. 

\begin{figure}[h]
\centering
\begin{minipage}[c]{.15\textwidth}
\centering
\scalebox{.5}{
\begin{tikzpicture}
[scale = 0.45, every node/.style={circle, fill = red, draw=black, minimum size=0.7cm}]
\node (1) at  (0, 0){1};\node (2) at (2,0){2};\node (3) at (4, 0){3};\node (4) at  (6, 0){4};
\node (5) at  (0, 2){5};\node (6) at (2,2){6};\node (7) at (4, 2){7};\node (8) at  (6, 2){8};
\node (9) at  (0, 4){9};\node (10) at (2,4){10};\node (11) at (4, 4){11};\node (12) at  (6, 4){12};
\node (13) at  (0, 6){13};\node (14) at (2,6){14};\node (15) at (4, 6){15};\node (16) at  (6, 6){16};
\node (17) at  (0, 8){17};\node (18) at (2,8){18};\node (19) at (4, 8){19};\node (20) at  (6, 8){20};
\draw[->] (1)  to (2);\draw[->] (2)  to (3);\draw[->] (3)  to (4);
\draw[->] (5)  to (6);\draw[->] (6)  to (7);\draw[->] (7)  to (8);
\draw[->] (9)  to (10);\draw[->] (10)  to (11);\draw[->] (11)  to (12);
\draw[->] (13)  to (14);\draw[->] (14)  to (15);\draw[->] (15)  to (16);
\draw[->] (17)  to (18);\draw[->] (18)  to (19);\draw[->] (19)  to (20);
\end{tikzpicture}}
\end{minipage}
\begin{minipage}[c]{.3\textwidth}
\centering
\caption{Nash best response cycles with respect to \maxcarnage~adversary. $\c=7/6$ and $\b=20$.
\label{fig:br-cycles}}
\end{minipage}
\begin{minipage}[c]{.15\textwidth}
\centering
\scalebox{.50}{
\begin{tikzpicture}
[scale = 0.45, every node/.style={circle, draw=black}, red node/.style = {circle, fill=red, draw, minimum size=0.7cm}]
\node [red node](21) at  (0, -2){4};\node [red node] (22) at (2,-2){1};\node [red node] (23) at (4, -2){2};\node [red node] (24) at  (6, -2){3};
\node [red node] (1) at  (0, 0){2};\node [red node] (2) at (2,0){3};\node [red node] (3) at (4, 0){4};\node [red node] (4) at  (6, 0){1};
\node [red node] (5) at  (0, 2){2};\node [red node] (6) at (2,2){3};\node [red node] (7) at (4, 2){4};\node [red node] (8) at  (6, 2){1};
\node [red node] (9) at  (0, 4){3};\node [red node] (10) at (2,4){4};\node [red node] (11) at (4, 4){1};\node [red node] (12) at  (6, 4){2};
\node [red node] (13) at  (0, 6){4};\node [red node] (14) at (2,6){1};\node [red node] (15) at (4, 6){2};\node [red node] (16) at  (6, 6){3};
\node [red node] (17) at  (0, 8){1};\node [red node] (18) at (2,8){2};\node [red node] (19) at (4, 8){3};\node [red node] (20) at  (6, 8){4};
\draw[->] (1)  to (2);\draw[->] (2)  to (3);
\draw[->] (5)  to (6);\draw[->] (6)  to (7);\draw[->] (7)  to (8);
\draw[->] (9)  to (10);\draw[->] (10)  to (11);\draw[->] (11)  to (12);
\draw[->] (13)  to (14);\draw[->] (14)  to (15);\draw[->] (15)  to (16);
\draw[->] (17)  to (18);\draw[->] (18)  to (19);\draw[->] (19)  to (20);
\draw[->] (21)  to (22);\draw[->] (22)  to (23);\draw[->] (23)  to (24);
\node [draw=none,rectangle,minimum width=1cm,minimum height=0.5cm] at (-3,8) {round 1};
\node [draw=none,rectangle,minimum width=1cm,minimum height=0.5cm,label=] at (-3,6) {round 2};
\node [draw=none,rectangle,minimum width=1cm,minimum height=0.5cm,label=] at (-3,4) {round 3};
\node [draw=none,rectangle,minimum width=1cm,minimum height=0.5cm,label=] at (-3,2) {round 4};
\node [draw=none,rectangle,minimum width=1cm,minimum height=0.5cm,label=] at (-3,0) {round 5};
\node [draw=none,rectangle,minimum width=1cm,minimum height=0.5cm,label=] at (-3,-2) {round 6};
\end{tikzpicture}
}
\end{minipage}
\begin{minipage}[c]{.3\textwidth}
\centering
\caption{The status of the first component at the beginning of the first 6 rounds of best response dynamics.
\label{fig:br-cycles-2}}
\end{minipage}
\end{figure}

\begin{example}
\label{ex:nash-cycles}
Consider the network in Figure~\ref{fig:br-cycles} (with all vulnerable vertices) with $n=20$, $\c=7/6$ and $\b=20$ to be 
the initial configuration in running the Nash best response dynamics. If
the vertices Nash best respond in the increasing order of their labels, then there exists a tie breaking
rule which causes the best response dynamics to cycle with respect to a \maxcarnage~adversary.
\end{example}
\begin{proof}
  Since the components are symmetric, we only analyze one of the components.
  Vertices 1 and 2 are currently best responding (although each has a
  deviation with the same payoff but we break ties in favor of their
  current action).  Vertex 3's best response is to drop her
  edge. Vertex 4's best response is to connect back to the same
  component she was a part of before vertex 3's best response.  We
  break ties by forcing vertex 4 to purchase an edge to vertex 1.

After the first round, we are in the same pattern as before but the labels of the 
vertices are different. So in the next round vertex 2 would drop
her edge and vertex 3 would buy an edge to vertex 4. In the third
round, vertex 1 would drop her edge. 
In the fourth round, 
and vertex 2 would buy an edge to
vertex 3. In the fourth round, vertex 4 would drop her edge. 
In the fifth round, vertex
1 would buy an edge to 2, vertex 3 would drop her edge and vertex 4 would buy an edge to vertex 1. 
So we are back in the same configuration that we were at the beginning or round 2 (see Figure~\ref{fig:br-cycles}).
\end{proof}

Since we considered Nash best responses, but all the best responses
chosen by the adversary were linkstable deviations,
Example~\ref{ex:nash-cycles} also shows that swapstable and linkstable
best response dynamics can cycle with respect to the \maxcarnage adversray 
if the order of vertices who best
respond are fixed but the ties in the best responses of a vertex are broken
adversarially.

We suspect that this phenomenon is the result of adversarial
tie-breaking and/or the ordering on the vertices as a similar observation has been made for the
  convergence of Nash best responses in the original reachability
  game~\cite{BalaG00}. We point out that in our experimental results
  in Section~\ref{sec:exp}, we used a fixed tie-breaking rule and yet
  the simulations always converged to an equilibrium. 
\end{document}